\documentclass[aps,amssymb,showkeys]{revtex4}

\usepackage{array}
\usepackage{multirow}
\usepackage{amsmath}
\usepackage{amsthm}
\usepackage{graphicx} 
\usepackage{amssymb}
\usepackage{kbordermatrix}
\usepackage[pdftex, bookmarks, colorlinks=true, plainpages = false, citecolor = red, urlcolor = blue, filecolor = blue]{hyperref}
\usepackage{color}
\usepackage{slashed}
\usepackage{yhmath}
\usepackage{mathrsfs}
\usepackage{enumerate}

\newcommand{\red}{\textcolor{red}}
\newcommand{\be}{\begin{equation}}
\newcommand{\ee}{\end{equation}}
\newcommand{\bea}{\begin{eqnarray}}
\newcommand{\eea}{\end{eqnarray}}
\newcommand{\bean}{\begin{eqnarray*}}
\newcommand{\eean}{\end{eqnarray*}}

\theoremstyle{plain}
\newtheorem{theorem}{Theorem}

\newtheorem{cor}[theorem]{Corollary}
\newtheorem{lem}[theorem]{Lemma}

\theoremstyle{definition}
\newtheorem{defn}[theorem]{Definition}
\newtheorem{notation}[theorem]{Notation}

\allowdisplaybreaks

\begin{document}
\title{A solution space for a system of null-state partial differential equations I}

\date{\today}

\author{Steven M.\ Flores}
\email{steven.flores@helsinki.fi} 
\affiliation{
Department of Mathematics, University of Michigan, Ann Arbor, Michigan, 48109-2136, USA \\ and \\
Department of Mathematics \& Statistics, University of Helsinki, P.O. Box 68, 00014, Finland}

\author{Peter Kleban}
\email{kleban@maine.edu} 
\affiliation{LASST and Department of Physics \& Astronomy, University of Maine, Orono, Maine, 04469-5708, USA}

\begin{abstract}  

This article is the first of four that completely and rigorously characterize a solution space $\mathcal{S}_N$ for a homogeneous system of $2N+3$ linear partial differential equations (PDEs) in $2N$ variables that arises in conformal field theory (CFT) and multiple Schramm-L\"owner evolution (SLE$_\kappa$).  In CFT, these are null-state equations and conformal Ward identities.  They govern partition functions for the continuum limit of a statistical cluster or loop-gas model, such as percolation, or more generally the Potts models and O$(n)$ models, at the statistical mechanical critical point.  (SLE$_\kappa$ partition functions also satisfy these equations.)  For such a lattice model in a polygon $\mathcal{P}$ with its $2N$ sides exhibiting a free/fixed side-alternating boundary condition $\vartheta$, this partition function is proportional to the CFT correlation function
\[\langle\psi_1^c(w_1)\psi_1^c(w_2)\dotsm\psi_1^c(w_{2N-1})\psi_1^c(w_{2N})\rangle^{\mathcal{P}}_\vartheta,\]
where the $w_i$ are the vertices of $\mathcal{P}$ and where $\psi_1^c$ is a  one-leg corner operator.  (Partition functions for ``crossing events" in which clusters join the fixed sides of $\mathcal{P}$ in some specified connectivity are linear combinations of such correlation functions.)  When conformally mapped onto the upper half-plane, methods of CFT show that this correlation function satisfies the system of PDEs that we consider.

In this first article, we use methods of analysis to prove that the dimension of this solution space is no more than $C_N$, the $N$th Catalan number.  While our motivations are based in CFT, our proofs are completely rigorous.  This proof is contained entirely within this article, except for the proof of lemma \ref{alltwoleglem}, which constitutes the second article \cite{florkleb2}.  In the third article \cite{florkleb3}, we use the results of this article to prove that the solution space of this system of PDEs has dimension $C_N$ and is spanned by solutions constructed with the CFT Coulomb gas (contour integral) formalism.  In the fourth article \cite{florkleb4}, we prove further CFT-related properties about these solutions, some useful for calculating cluster-crossing probabilities of critical lattice models in polygons.

\end{abstract}

\maketitle

\section{Introduction}\label{intro}

We consider critical bond percolation on a very fine square lattice inside a rectangle $\mathcal{R}:=\{x+iy\,|\,0<x<R,0<y<1\}$ with \emph{wired} (or \emph{fixed}) left and right sides (i.e., all bonds are activated on these sides) and \emph{free} top and bottom sides (i.e., we do not condition the state of any of the bonds on these sides).  In \cite{c3}, J.\ Cardy used conformal field theory (CFT) \cite{bpz, fms, henkel} methods to argue that, at the critical point and in the continuum limit, the partition function for this system is proportional to the CFT correlation function
\be\label{corrfunc}\langle\psi_1^c(w_1)\psi_1^c(w_2)\psi_1^c(w_3)\psi_1^c(w_4)\rangle^{\mathcal{R}},\ee
where $w_i\in\mathbb{C}$ is the $i$th vertex of $\mathcal{R}$ and $\psi_1^c(w_i)$ is a $c=0$ CFT one-leg corner operator \cite{c1,sk,skfz,fkz,gruz, rgbw} that implements the boundary condition change (BCC) from free to fixed \cite{c2,c1,c3} at this vertex.  The essence of this argument supposed the emergence of conformal invariance at the critical point for bond percolation in the continuum limit, a feature that was previously observed in computer simulations \cite{lps}.  By considering the $Q\rightarrow1$ limit of a particular combination of certain $Q$-state random cluster model partition functions given by (\ref{corrfunc}), Cardy then predicted a formula for the probability that the wired sides of $\mathcal{R}$ are joined by a cluster of activated bonds (figure \ref{CrossingFig}).  His prediction, called \emph{Cardy's formula}, is \cite{c3}
\be\label{Cardy}\mathbb{P}\{\text{left-right crossing}\}=\frac{3\Gamma(2/3)}{\Gamma(1/3)^2}m^{1/3}\,_2F_1\left(\frac{1}{3},\frac{2}{3};\frac{4}{3}\,\Bigg|\,m\right), \quad R=K(1-m)/K(m).\ee
Here, $m\in(0,1)$ corresponds one-to-one with the aspect ratio $R\in(0,\infty)$ of the rectangle $\mathcal{R}$ via the second equation in (\ref{Cardy}), with $K$ the complete elliptic function of the first kind \cite{morsefesh}.  Computer simulations \cite{lps,ziff} have numerically verified this prediction (\ref{Cardy}), thus giving very strong evidence for the presence of conformal symmetry in the continuum limit of critical percolation.  Other simulations \cite{lpps} consistently suggest that many observables, such as the probability of the left-right cluster-crossing event, common to different homogeneous models of critical percolation (e.g., site vs.\ bond percolation and percolation on different regular lattices) converge to the same value in the continuum limit, a phenomenon called \emph{universality}.  Later, S.\ Smirnov rigorously proved Cardy's formula for site percolation on the triangular lattice \cite{smirnov}.  Also after Cardy's result (\ref{Cardy}), researchers have used CFT to predict other formulas involving critical percolation cluster crossings \cite{watts, skz1}, densities \cite{skz2, skdz, skz3, skfz}, and pinch points \cite{pinchpoint}.

\begin{figure}[t!]
\centering
\includegraphics[scale=0.6]{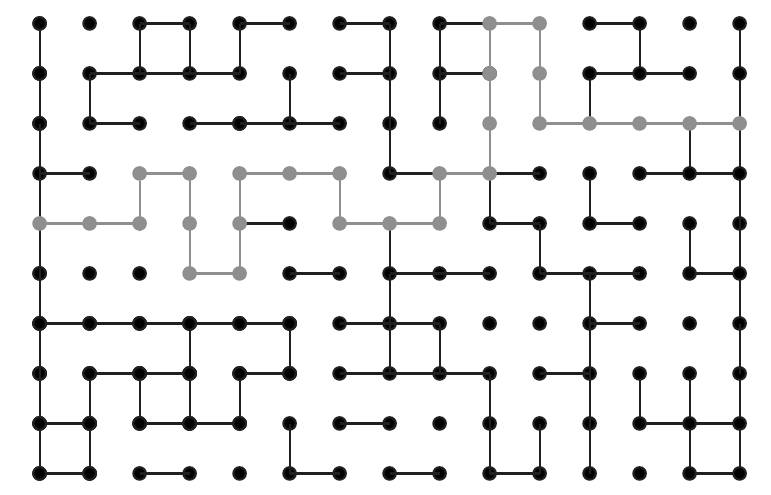}
\caption{A bond percolation configuration in a rectangle with a crossing (gray bonds) from the left side of the rectangle to the right side.}
\label{CrossingFig}
\end{figure}

\begin{figure}[b!]
\centering
\includegraphics[scale=0.28]{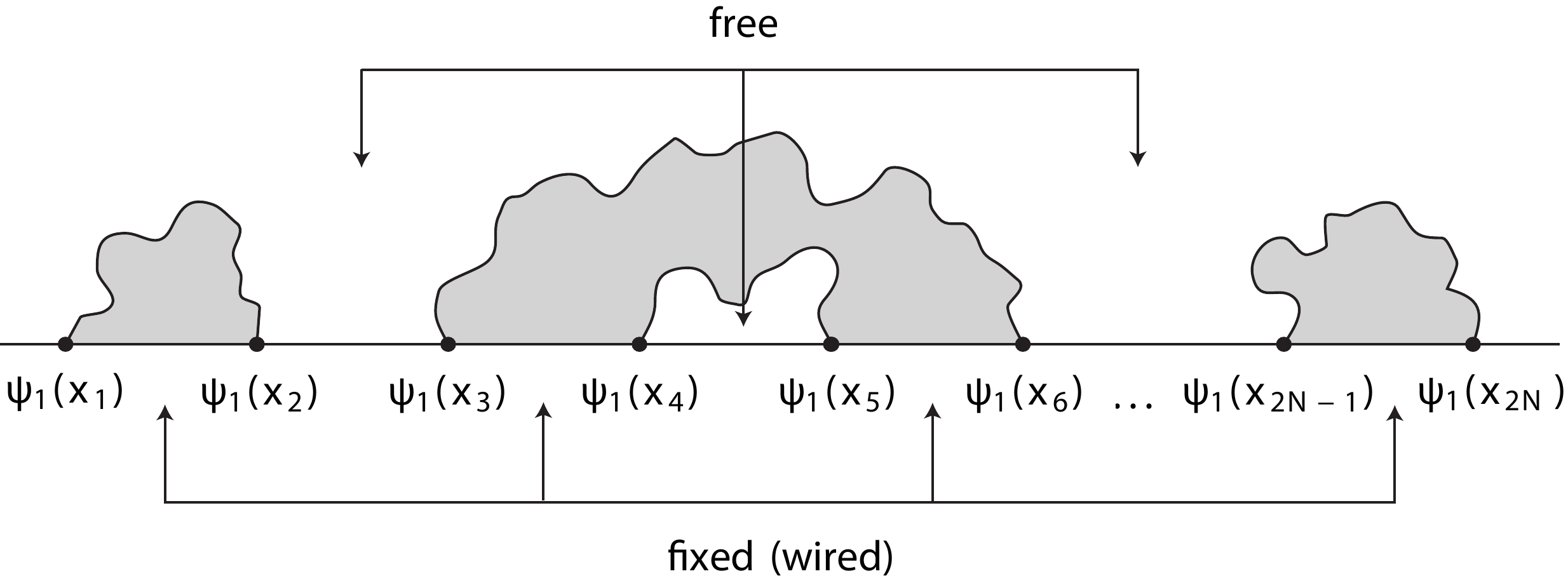}\\
\caption{An FFBC sample, with boundary clusters anchored to wired segments.  We have conformally mapped the polygon interior onto the half-plane, with $x_i$ the image of the $i$th vertex.}
\label{FFBC}
\end{figure}

The setup for Cardy's formula has interesting generalizations that motivate the analysis presented in this article.  Looking beyond rectangles, we may consider system domains that are even-sided polygons $\mathcal{P}$,  with the boundary condition (BC) alternating from wired to free to wired, etc., as we trace the boundary of $\mathcal{P}$ from side to side.  We call this a \emph{free/fixed side-alternating boundary condition} (FFBC) (figure \ref{FFBC}).  And looking beyond percolation, we may consider other lattice models with  critical points that have CFT descriptions in the continuum limit.  These include the Potts model \cite{wu} and its close relative, the random cluster model \cite{fk}.  If we enumerate the FFBC events and condition one of these systems to exhibit the $\vartheta$th FFBC event on the boundary of $\mathcal{P}$, then we may adapt Cardy's argument to predict that the conditioned partition function is proportional to the $2N$-point CFT correlation function
\be\label{2Npoint}\langle\psi_1^c(w_1)\psi_1^c(w_2)\dotsm\psi_1^c(w_{2N-1})\psi_1^c(w_{2N})\rangle^{\mathcal{P}}_\vartheta. \ee
Here, $w_i\in\mathbb{C}$ is the $i$th vertex of $\mathcal{P}$, and $\psi_1^c(w_i)$ is a CFT one-leg corner operator \cite{c1,sk,skfz,fkz,gruz, rgbw} that implements the BCC from free to fixed \cite{c2,c1,c3} at this vertex (appendix \ref{preliminaries}).  If we condition the system to exhibit the $\vartheta$th FFBC event, then Potts model spin clusters or FK clusters, anchor to the wired sides of $\mathcal{P}$ and join these sides in some topological crossing configuration with some non-trivial probability.  We call these \emph{boundary clusters}.  An induction argument \cite{dubedat} shows that there are $C_N$ such configurations (figure \ref{CatalanFig}), with $C_N$ the $N$th Catalan number given by
\be\label{catalan}C_N=\frac{(2N)!}{N!(N+1)!}.\ee
Formulas for \emph{crossing probabilities}, or probabilities of these crossing events, as functions of the shape of $\mathcal{P}$ generalize Cardy's formula (\ref{Cardy}), which corresponds to the case of critical percolation with $N=2$.  In \cite{fkz}, we use results from this article and its sequels \cite{florkleb2, florkleb3, florkleb4} to predict some of these formulas, extending recent results on crossing probabilities for hexagons \cite{js,dub}.

To calculate the $2N$-point function (\ref{2Npoint}), we conformally map $\mathcal{P}$ onto the upper half-plane (figure \ref{FFBC}).  After we continuously extend it to the boundary of $\mathcal{P}$, this map also sends the vertices $w_1$, $w_2,\ldots,w_{2N}$ onto real numbers $x_1<x_2<\ldots<x_{2N}$, and it sends the one-leg corner operator $\psi_1^c(w_i)$ hosted by the $i$th vertex of $\mathcal{P}$ to a one-leg boundary operator $\psi_1(x_i)$ at $x_i$ \cite{gruz, rgbw}.  In the Potts model (resp.\ random cluster model) \cite{wu, fk}, the one-leg boundary operator is a primary operator that belongs to the $(2,1)$ (resp.\ $(1,2)$) position of the Kac table \cite{c3, bpz, fms, henkel, c2}.  (We discuss one-leg boundary operators further in appendix \ref{preliminaries}.)  Thus, the CFT null-state condition implies that this half-plane version of (\ref{2Npoint}) satisfies the system of $2N$ \emph{null-state partial differential equations (PDEs)} \cite{bpz, fms}
\be\label{nullstateCFT}\Bigg[\frac{3\partial_j^2}{2(2\theta_1+1)}+\sum_{k\neq j}^{2N}\left(\frac{\partial_k}{x_k-x_j}-\frac{\theta_1}{(x_k-x_j)^2}\right)\Bigg]F(x_1,x_2,\ldots,x_{2N})=0,\quad j\in\{1,2,\ldots,2N\},\ee
where the \emph{one-leg boundary weight} $\theta_1$ is the conformal weight of the one-leg boundary operator and is given by \cite{bpz, fms, henkel} (appendix \ref{preliminaries})
\be\label{theta1cft}\theta_1=\frac{1}{16}\left[5-c\pm\sqrt{(c-1)(c-25)}\right].\ee
Here, $c$ is the central charge of the CFT, and it corresponds to the model under consideration.  For example, $c=0$ corresponds to percolation \cite{c3}, and $c=1/2$ corresponds to the Ising model (i.e., the two-state Potts model) and the two-state random cluster model \cite{bpz}.  Also, the sign to be used in (\ref{theta1cft}) depends on the model.  For example, we use the $+$ sign for the Potts model and the $-$ sign for the random cluster model.  

\begin{figure}[b!]
\centering
\includegraphics[scale=0.27]{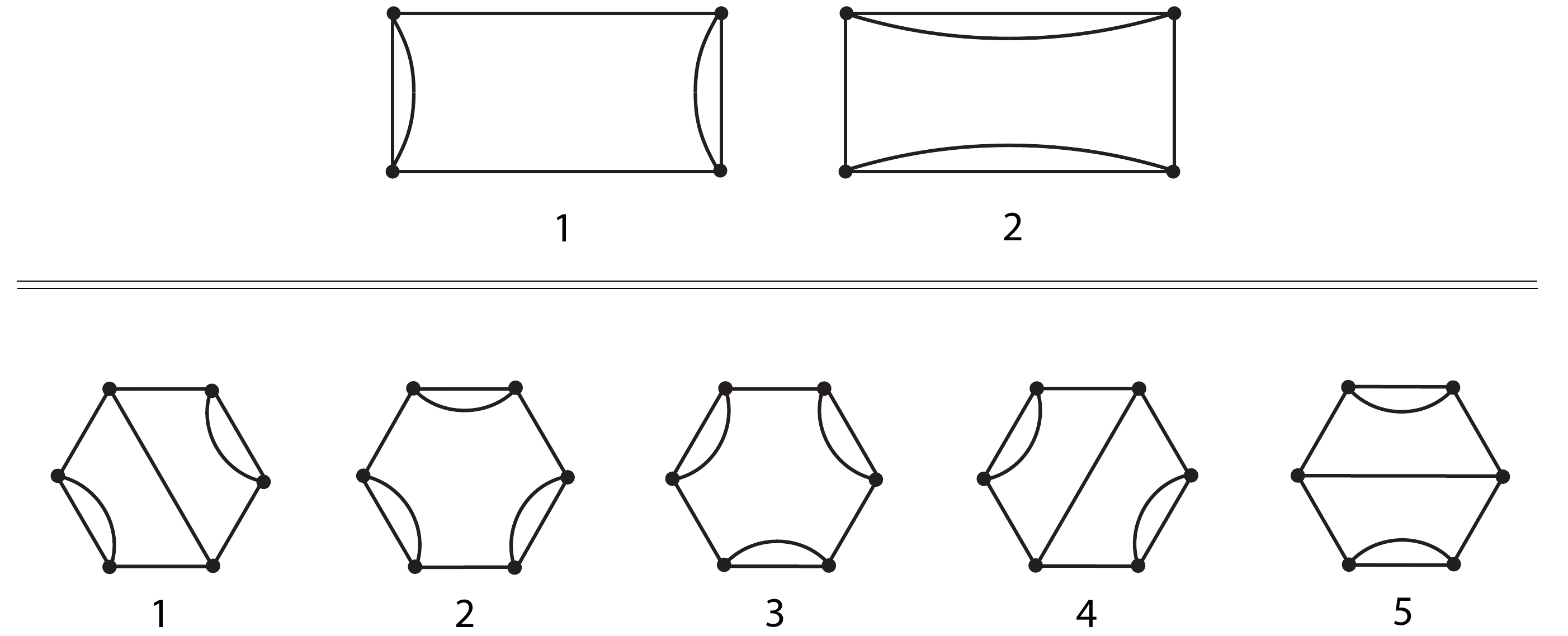}
\caption{The number of crossing configurations (or boundary arc connectivities) in the rectangle (resp.\ hexagon) equals the second (resp.\ third) Catalan number $C_2=2$ (resp.\ $C_3=5$).}
\label{CatalanFig}
\end{figure}

Aside from the null-state PDEs (\ref{nullstateCFT}), any CFT correlation function must satisfy three conformal Ward identities \cite{bpz, fms, henkel}.  For the half-plane version of the particular $2N$-point function (\ref{2Npoint}) above, these are
\be\begin{gathered}\label{wardidCFT}\sum_{k=1}^{2N}\partial_kF(x_1,x_2,\ldots,x_{2N})=0,\quad \sum_{k=1}^{2N}(x_k\partial_k+\theta_1)F(x_1,x_2,\ldots,x_{2N})=0,\\
\sum_{k=1}^{2N}(x_k^2\partial_k+2\theta_1x_k)F(x_1,x_2,\ldots,x_{2N})=0.\end{gathered}\ee
If a function $F$ satisfies these PDEs (\ref{wardidCFT}), then it is covariant with respect to conformal bijections of the upper half-plane onto itself, with each coordinate $x_i$ having conformal weight $\theta_1$.  We elaborate on this in section \ref{objorg} below.

If $N=1$, then it is easy to show that any solution to the system (\ref{nullstateCFT}, \ref{wardidCFT}) is of the form $C(x_2-x_1)^{-2\theta_1}$ for some arbitrary constant $C$.  Thus, the rank (i.e., the dimension of the solution space) equals the first Catalan number, $C_1=1$.  If $N=2$, then we may use the conformal Ward identities (\ref{wardidCFT}) to convert the system of four null-state PDEs (\ref{nullstateCFT}) into a single hypergeometric differential equation \cite{bpz}.  The general solution of this differential equation completely determines the solution space, so the rank of the system equals the second Catalan number, $C_2=2$.  (After setting $c=0$, an appropriate boundary condition argument gives Cardy's formula (\ref{Cardy}).)  If $N=3$, then a similar but more complicated argument \cite{dub} shows that the rank of the system equals the third Catalan number, $C_3=5$, at least when $c=0$.  Beyond this, the need for $C_N$ linearly independent crossing-probability formulas suggests that the rank of the system is at least $C_N$, but it does not apparently suggest that the rank is $C_N$ exactly.  In spite of this, the Coulomb gas formalism \cite{dub, df1,df2} allows us to construct many explicit classical (in the sense of \cite{giltru}) solutions of the system for any $N\in\mathbb{Z}^+$, a remarkable feat!  Neither this article nor its first sequel \cite{florkleb2} uses these solutions, but the latter sequels \cite{florkleb3, florkleb4} do.

\begin{figure}[b!]
\centering
\includegraphics[scale=0.5]{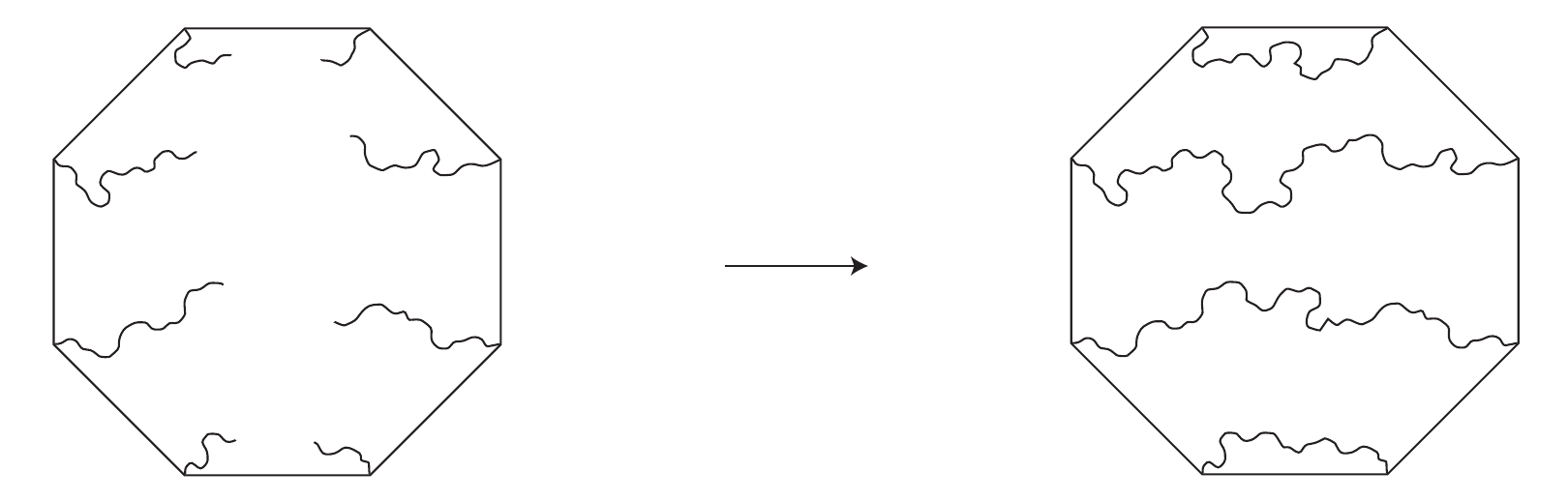}
\caption{Multiple SLE$_\kappa$ in an octagon.  A multiple-SLE$_\kappa$ curve grows from each vertex into the octagon, and these curves join pairwise to form four distinct, non-crossing boundary arcs.}
\label{MultiSLE}
\end{figure}

In addition to CFT, we may use multiple SLE$_\kappa$ \cite{dub2, bbk, graham, kl, sakai}, a generalization of SLE$_\kappa$ (Schramm-L\"owner evolution) \cite{rohshr, knk, lawler}, to study the continuum limit of a critical lattice model inside a polygon $\mathcal{P}$ with an FFBC.  As we use this approach, we forsake the boundary clusters and study their perimeters instead.  These perimeters, called \emph{boundary arcs}, are random fractal curves that fluctuate inside $\mathcal{P}$.  Their law is conjectured, (and proven for some models in the case of (ordinary) SLE$_\kappa$, see table \ref{kappavalues}), to be that of multiple SLE$_\kappa$, a stochastic process that simultaneously grows $2N$ fractal curves, one from each vertex, inside $\mathcal{P}$.  These curves explore the interior of $\mathcal{P}$ without crossing themselves or each other until they join to form $N$ distinct, non-crossing boundary arcs that connect the vertices of $\mathcal{P}$ pairwise (figure \ref{MultiSLE}) in a specified connectivity \cite{js,bbk}.  An induction argument \cite{dubedat} shows that these $2N$ curves join in one of $C_N$ possible connectivities, called \emph{boundary arc connectivities} \cite{pinchpoint} (figure \ref{CatalanFig}).  Furthermore, we identify each boundary arc connectivity with the particular cluster-crossing event whose boundary arcs join in that connectivity.

\begin{table}
\centering
\begin{tabular}{lcccl} 
Random walk or critical lattice model\,&$\kappa$ &$c$ & & Current status \\
\hline
The loop-erased random walk \cite{law1}& 2 & $-2$ & &proven \cite{lsw}\\
The self-avoiding random walk \cite{madraslade} & 8/3 & 0 & & conjectured \cite{lsw2}\\
$Q=2$ Potts spin cluster perimeters \cite{wu}& 3 & 1/2 & & proven \cite{smir3}\\
$Q=3$ Potts spin cluster perimeters  \cite{wu} & 10/3& 4/5 & & conjectured \cite{smir}\\
$Q=4$ Potts spin/FK cluster perimeters \cite{wu, fk} & 4 & 1 & & conjectured \cite{smir}\\
The level line of a Gaussian free field \cite{schrsheff} & 4 & 1 & & proven \cite{schrsheff}\\
The harmonic explorer \cite{schrsheff} & 4 & 1 & & proven \cite{schrsheff}\\
$Q=3$ Potts FK cluster perimeters \cite{fk} & 24/5 & 4/5 & & conjectured \cite{smir}\\
$Q=2$ Potts FK cluster perimeters \cite{fk} & 16/3 & 1/2 & & proven \cite{smir3}\\
Percolation and smart-kinetic walks \cite{grim, weintru} & 6 & 0 & & proven \cite{smir2}\\
Uniform spanning trees \cite{lsw} & 8 & $-2$ & & proven \cite{lsw}\\
\hline
\end{tabular}
\caption{Models conjectured or proven to have conformally invariant continuum limits and with random curves whose law converges to SLE$_\kappa$ as we approach this continuum limit.}
\label{kappavalues}
\end{table}

Multiple SLE$_\kappa$ provides a different, rigorous approach to calculating some observables that may be predicted via CFT, and although these two approaches are fundamentally different, they are closely related \cite{bauber}.  Two entities determine the multiple-SLE$_\kappa$ process \cite{dub2,bbk}:
\begin{enumerate}[I.]
\item\label{condition1}The first is stochastic: a collection of $2N$ absolutely continuous martingales with zero cross-variation and total quadratic variation $\kappa t$, with $\kappa>0$ the SLE$_\kappa$ \emph{speed} or \emph{parameter} and $t>0$ \emph{the evolution time}.  The equation \cite{bauber}
\be\label{central}c(\kappa)=\frac{(6-\kappa)(3\kappa-8)}{2\kappa}\ee
relates a CFT of central charge $c<1$ (resp.\ $c=1$) to a multiple SLE$_\kappa$ with one of two possible speeds, one in the \emph{dilute phase} $\kappa\in(0,4]$, and one in the \emph{dense phase} $\kappa\in(4,\infty)$ (resp.\ with one speed $\kappa=4$ in the dilute phase) \cite{knk} of SLE$_\kappa$.  Further arguments provided in appendix \ref{preliminaries} show that if we substitute (\ref{central}) into (\ref{theta1cft}), then we must use the $+$ (resp.\ $-$) sign in the dilute (resp.\ dense) phase, so \cite{bauber}
\be\label{theta1}\theta_1=\frac{6-\kappa}{2\kappa}.\ee
\item\label{condition2}The second is deterministic: a nonzero function $F$, which we call an {\it $SLE_\kappa$ partition function}.  (This is similar to but slightly different from the actual partition function of the critical system under consideration.  See appendix \ref{preliminaries} and \cite{fkz}.)  The only condition imposed on  $F$ is that it satisfies the system of null-state PDEs (\ref{nullstateCFT}) and the three conformal Ward identities (\ref{wardidCFT}) (in the classical sense of \cite{giltru}) with $\theta_1$ given by (\ref{theta1}), and that it never equals zero.
\end{enumerate}
Because all boundary arcs have the statistics of SLE$_\kappa$ curves in the small regardless of our choice of SLE$_\kappa$ partition function, the partition function that we do use may only influence a large-scale property of multiple SLE$_\kappa$, such as the eventual pairwise connectivity of its curves in the long-time limit.  Indeed, two multiple-SLE$_\kappa$ processes whose curves are conditioned to join in different connectivities obey the same stochastic PDEs driven by the martingales of condition \ref{condition1}.  Because which SLE$_\kappa$ partition function to use for item \ref{condition2} above is the only unspecified feature of these equations, we expect that this choice influences the eventual boundary arc connectivity.

This supposition naturally leads us to conjecture the rank of the system (\ref{nullstateCFT}, \ref{wardidCFT}) previously considered in our CFT approach above.  As mentioned, there are $C_N$ possible boundary arc connectivities, which we enumerate one through $C_N$.  Thus, there must be at least one SLE$_\kappa$ partition function per connectivity that conditions the boundary arcs to join pairwise in that connectivity almost surely.  But furthermore, if our SLE$_\kappa$ partition function does not influence any of the boundary arcs' other large-scale properties, then there may be at most one SLE$_\kappa$ partition function $\Pi_\varsigma$, called the \emph{$\varsigma$th connectivity weight}, that conditions the boundary arcs to join in, say, the $\varsigma$th connectivity.  If this is true, then we anticipate that the set $\{\Pi_1,\Pi_2,\ldots,\Pi_{C_N}\}$ is a basis for the solution space of the system (\ref{nullstateCFT}, \ref{wardidCFT}), and the rank of the system is therefore $C_N$.  Proving this last statement is one of the principal goals of this article and its sequels \cite{florkleb2, florkleb3, florkleb4}.

\subsection{Objectives and organization} \label{objorg}

So far, we have used the application of the system (\ref{nullstateCFT}, \ref{wardidCFT}) to critical lattice models and multiple SLE$_\kappa$ to anticipate some of the properties of its solution space.  In this section, we set the stage for our proof of some of these properties by declaring the goals, describing the organization, and establishing some notation conventions for this article and its sequels \cite{florkleb2, florkleb3, florkleb4}.  Inserting (\ref{theta1}) in the system (\ref{nullstateCFT}, \ref{wardidCFT}) gives the $2N$ null-state PDEs
\be\label{nullstate}\Bigg[\frac{\kappa}{4}\partial_j^2+\sum_{k\neq j}^{2N}\left(\frac{\partial_k}{x_k-x_j}-\frac{(6-\kappa)/2\kappa}{(x_k-x_j)^2}\right)\Bigg]F(\boldsymbol{x})=0,\quad j\in\{1,2,\ldots,2N\},\ee
with $\boldsymbol{x}:=(x_1,x_2,\ldots,x_{2N})$ and $\kappa>0$ (however we consider only $\kappa\in(0,8)$ in this article), and the three conformal Ward identities
\be\label{wardid}\sum_{k=1}^{2N}\partial_kF(\boldsymbol{x})=0,\quad \sum_{k=1}^{2N}\left[x_k\partial_k+\frac{6-\kappa}{2\kappa}\right]F(\boldsymbol{x})=0,\quad \sum_{k=1}^{2N}\left[x_k^2\partial_k+\frac{(6-\kappa)x_k}{\kappa}\right]F(\boldsymbol{x})=0.\ee 
We call the $j$th null-state PDE among (\ref{nullstateCFT}) \emph{the null-state PDE centered on $x_j$.}  Although this system (\ref{nullstate}, \ref{wardid}) arises in CFT in a way that is typically non-rigorous, our treatment of this system here and in \cite{florkleb2,florkleb3,florkleb4} is completely rigorous.  Before declaring what we intend to prove about this system of PDEs, we observe some important facts about it.
\begin{itemize}
\item The subsystem of $2N$ null-state PDEs (\ref{nullstate}) is undefined on the locus of {\it diagonal points} in $\mathbb{R}^{2N}$, or points with at least two of its coordinates equal.  We let $\Omega$ be the complement of the locus of diagonal points in $\mathbb{R}^{2N}$.  Then the diagonal points make up the boundary $\partial\Omega$, and all together, they divide $\Omega$ into connected components, each of the form
\be\label{components}\Omega_\sigma:=\{\boldsymbol{x}\in\Omega\,| \,x_{\sigma(1)}<x_{\sigma(2)}<\ldots< x_{\sigma(2N-1)}< x_{\sigma(2N)}\}\ee
for some permutation $\sigma\in S_{2N}$.  By symmetry, it suffices to restrict the domain of our solutions to the component $\Omega_0:=\Omega_{\sigma_0}$ corresponding to the identity permutation $\sigma_0$.  That is, we take $x_i<x_j$ whenever $i<j$ without loss of generality.  In this article and its sequel \cite{florkleb2}, we refer to $\boldsymbol{x}:=(x_1,x_2,\ldots,x_{2N})\in\Omega_\sigma$ as a point in a component of $\Omega$ and $x_i$ as the $i$th coordinate of that point, but in the sequels \cite{florkleb3, florkleb4}, we refer to $x_i$ as a point.
\item The subsystem (\ref{nullstate}) is elliptic, so all of its solutions exhibit strong regularity.  Indeed, after summing over all $2N$ null-state PDEs, we find that any solution satisfies a linear homogeneous strictly elliptic PDE whose coefficients are analytic in any connected component of $\Omega$.  (In fact, the principal part of this PDE is simply the Laplacian.)  It follows from the theorem of Hans Lewy \cite{berssch} that all of its solutions are (real) analytic in any connected component of $\Omega$.  We use this fact to exchange the order of integration and differentiation in many of the integral equations that we encounter here and in \cite{florkleb2}.

\item We may explicitly solve the conformal Ward identities (\ref{wardid}) via the method of characteristics.  It follows that any function $F:\Omega_0\rightarrow\mathbb{R}$ that satisfies these identities must have the form
\be\label{covar}F(\boldsymbol{x})=G(\lambda_1,\lambda_2,\ldots,\lambda_{2N-3})\prod_{j=1}^{2N}|x_j-x_{\tau(j)}|^{(\kappa-6)/2\kappa},\ee
where $\{\lambda_1,\lambda_2,\ldots,\lambda_{2N-3}\}$ is any set of $2N-3$ independent cross-ratios that we may form from $x_1,$ $x_2,\ldots,x_{2N}$, where $G(\lambda_1,\lambda_2,\ldots,\lambda_{2N-3})$ is a (real) analytic function of $\boldsymbol{x}\in\Omega_0$, and where $\tau$ is any pairing (i.e., a permutation $\tau\in S_{2N}$ other than the identity with $\tau=\tau^{-1}$) of the indices $1,2,\ldots,2N$.

\item We suppose that $f$ is a M\"{o}bius transformation sending the upper half-plane onto itself, and we define $x'_i:=f(x_i)$ and $\boldsymbol{x}':=(x_1',x_2',\ldots,x_{2N}')$.  Then the mapping $T:\Omega_0\rightarrow\Omega$ defined by $T(\boldsymbol{x})=\boldsymbol{x}'$ sends $\Omega_0$ onto a possibly different connected component $T(\Omega_0)$ of $\Omega$.  Because the cross-ratios $\lambda_1$, $\lambda_2,\ldots,\lambda_{2N-3}$ are invariant under $f$, the right side of (\ref{covar}) evaluated at any $\boldsymbol{x}'\in T(\Omega_0)$ is well-defined.  If we enumerate the permutations in $S_{2N}$ so $\Omega_0$, $\Omega_{\sigma_1}$, $\Omega_{\sigma_2},\ldots,\Omega_{\sigma_M}$ are all of the components of $\Omega$ that may be reached from $\Omega_0$ by such a transformation $T$, then we use (\ref{covar}) to extend $F$ to the function 
\be\label{hatF}\hat{F}:\bigcup_{j=0}^M\Omega_{\sigma_j}\rightarrow\mathbb{R},\quad \hat{F}(\boldsymbol{x}):=G(\lambda_1,\lambda_2,\ldots,\lambda_{2N-3})\prod_{j=1}^{2N}|x_j-x_{\tau(j)}|^{(\kappa-6)/2\kappa}.\ee
It is evident that because $F$ in (\ref{covar}) satisfies the system of PDEs (\ref{nullstate}, \ref{wardid}), $\hat{F}$ must satisfy this system too on each component of $\Omega$ in its domain.  Now, it is easy to show that (\ref{hatF}) transforms covariantly with respect to conformal bijections of the upper half-plane onto itself, with each of the $2N$ independent variables having \emph{conformal weight} $\theta_1$ (\ref{theta1}).  In other words, the functional equation (where $\partial f(x):=\partial f(x)/\partial x$)
\be\label{transform}\hat{F}(\boldsymbol{x}')=\partial f(x_1)^{-\theta_1}\partial f(x_2)^{-\theta_2}\dotsm\partial f(x_{2N})^{-\theta_1}\hat{F}(\boldsymbol{x}),\quad\theta_1:=(6-\kappa)/2\kappa\ee
holds whenever $f$ is a M\"{o}bius transformation taking the upper half-plane onto itself.  Such transformations are compositions of translation by $a\in\mathbb{R}$, dilation by $b>0$, and the inversion $x\mapsto-1/x$ (all of which have positive-valued derivatives).  Hence, $F(\boldsymbol{x})$ is invariant as we translate the coordinates of $\boldsymbol{x}$ by the same amount and is covariant with conformal weight $\theta_1$ (\ref{theta1}) as we dilate all of them by the same factor or invert all of them.  The first, second, and third Ward identities (\ref{wardid}) (counting from the left) respectively induce these three properties.

\item We may directly compute the solution space for the system of PDEs (\ref{nullstate}, \ref{wardid}) with domain $\Omega_0$ for $N\in\{1,2\}$ \cite{bbk}.  
\begin{itemize}
\item In the $N=1$ case, we use the first conformal Ward identity of (\ref{wardid}) (counting from the left) to reduce either PDE in (\ref{nullstate}) to a second order Euler differential equation in the one variable $x_2-x_1$.  The Euler equation has two characteristic powers (given in (\ref{p}) below), and the second conformal Ward identity of (\ref{wardid}) permits only the power $1-6/\kappa$.  Thus, the solution space is
\be\label{S1}\mathcal{S}_1=\{F:\Omega_0\rightarrow\mathbb{R}\,|\, \text{$F(x_1,x_2)=C(x_2-x_1)^{1-6/\kappa}$ for some $C\in\mathbb{R}$}\}.\ee
It is easy to show that the elements of $\mathcal{S}_1$ satisfy the third conformal Ward identity of (\ref{wardid}).
\item In the $N=2$ case, the conformal Ward identities demand that our solutions have the form (\ref{covar}), which we write as
\be\label{4ptansatz}F(\boldsymbol{x})=(x_4-x_2)^{1-6/\kappa}(x_3-x_1)^{1-6/\kappa}G\left(\frac{(x_2-x_1)(x_4-x_3)}{(x_3-x_1)(x_4-x_2)}\right),\ee
with $G$ an unspecified function.  By substituting (\ref{4ptansatz}) into any one of the null-state PDEs, we find that $[\lambda(1-\lambda)]^{-2/\kappa}G(\lambda)$ satisfies a second order hypergeometric differential equation.  This restricts $G$ to a linear combination of two possible functions $G_1$ and $G_2$ given by
\be\label{Pi1}G_1(\lambda)=G_2(1-\lambda)=\lambda^{2/\kappa}(1-\lambda)^{1-6/\kappa}\,_2F_1\,\left(\frac{4}{\kappa},1-\frac{4}{\kappa};\frac{8}{\kappa}\,\bigg|\,\lambda\right),\ee
with $_2F_1$ the Gauss hypergeometric function \cite{morsefesh}.  Thus, with $\lambda:=(x_2-x_1)(x_4-x_3)/(x_3-x_1)(x_4-x_2)$, the solution space is
\be\label{S2}\mathcal{S}_2=\{F:\Omega_0\rightarrow\mathbb{R}\,|\, \text{$F(\boldsymbol{x})=[(x_4-x_2)(x_3-x_1)]^{1-6/\kappa}[C_1G_1(\lambda)+C_2G_2(\lambda)]$ for some $C_1,C_2\in\mathbb{R}$}\}.\ee
\end{itemize}
\end{itemize}

Now we define the solution space for the system of PDEs (\ref{nullstate}, \ref{wardid}) that we wish to rigorously characterize in this article:
\begin{defn}\label{solutionspace} Let $\mathcal{S}_N$ denote the vector space over the real numbers of all functions $F:\Omega_0\rightarrow\mathbb{R}$ 
\begin{itemize}
\item that satisfy the system of PDEs (\ref{nullstate}, \ref{wardid}) (in the classical sense of \cite{giltru}), and
\item for which there exist positive constants $C$ and $p$ (which we may choose to be as large as needed) such that
\be\label{powerlaw} |F(\boldsymbol{x})|\leq C\prod_{i<j}^{2N}|x_j-x_i|^{\mu_{ij}(p)}\quad\text{with}\quad\mu_{ij}(p):=\begin{cases}-p, & |x_i-x_j|<1 \\ +p, & |x_i-x_j|\geq1\end{cases}\quad\text{for all $\boldsymbol{x}\in\Omega_0.$}\ee
\end{itemize}
\end{defn}
Sometimes, we write $\mathcal{S}_N(\kappa_0)$ to specify the particular solution space $\mathcal{S}_N$ with $\kappa=\kappa_0$ and use a similar notation for subsets of $\mathcal{S}_N$ too.  But usually, we suppress reference to the parameter $\kappa$ and simply write $\mathcal{S}_N$.

One may explicitly construct many putative elements of $\mathcal{S}_N$ by using the Coulomb gas formalism first proposed by V.S.\ Dotsenko and V.A.\ Fateev \cite{df1,df2}.  This method is non-rigorous, but in \cite{dub}, J.\ Dub\'{e}dat gave a proof that these ``candidate solutions" indeed satisfy the system of PDEs (\ref{nullstate}, \ref{wardid}).  We call these solutions \emph{Coulomb gas solutions.}  

The goal of this article and its sequels \cite{florkleb2,florkleb3,florkleb4} is to completely determine the space $\mathcal{S}_N$ for all $\kappa\in(0,8)$.  By ``determine," we mean
\begin{enumerate}
\item\label{item1} Rigorously prove that $\mathcal{S}_N$ is spanned by real-valued Coulomb gas solution.
\item\label{item2} Rigorously prove that $\dim\mathcal{S}_N=C_N$.
\item\label{item3} Argue that $\mathcal{S}_N$ has a basis $\mathscr{B}_N:=\{\Pi_1,\Pi_2,\ldots,\Pi_{C_N}\}$ of $C_N$ connectivity weights and find formulas for all of the connectivity weights.
\end{enumerate}
(For all $\kappa\geq8$, the multiple-SLE$_\kappa$ curves are space-filling almost surely \cite{rohshr,knk}.  Although we suspect that the findings of this article and its sequels \cite{florkleb2, florkleb3, florkleb4} are true for all $\kappa\geq8$, our proofs do not carry over to this range.)  Goals \ref{item1} and \ref{item2} determine the size and content of $\mathcal{S}_N$.  In this article, we prove the upper bound $\dim\mathcal{S}_N\leq C_N$.  To obtain that upper bound in this article, we construct a basis $\mathscr{B}_N^*$ for the dual space $\mathcal{S}_N^*$ of linear functionals acting on $\mathcal{S}_N$.  The derivation of this bound is contained entirely within this article, except for the proof of lemma \ref{alltwoleglem} below, which we defer to the second article \cite{florkleb2} of this series.  In the third article \cite{florkleb3}, we use the results of this article and \cite{florkleb2} to achieve goals \ref{item1} and \ref{item2} above.  In the fourth article \cite{florkleb4}, we investigate connectivity weights among other topics, and we use them to predict a formula for the probability of a particular multiple-SLE$_\kappa$ boundary arc connectivity.  A heuristic, though non-rigorous, argument shows that the basis $\mathscr{B}_N$ for $\mathcal{S}_N$ that is dual to $\mathscr{B}_N^*$ comprises all of the connectivity weights.  This realization gives a direct method for their computation, as desired in goal \ref{item3}.

In a future article \cite{fkz}, we use the connectivity weights to derive continuum-limit crossing-probability formulas for critical lattice models (such as percolation, Potts models, and random cluster models) in a polygon with an FFBC.  We verify our predictions with high-precision computer simulations of the $Q\in\{2,3\}$ critical random cluster model in a hexagon, finding excellent agreement.

In appendix \ref{preliminaries}, we survey some of the CFT methodologies used to study critical lattice models.  This formalism (non-rigorously) anticipates many of our results, so we often interpret our findings in context with CFT throughout this article.  The reader who is not familiar with this approach but wishes to understand our asides to it may consult this appendix.  (We emphasize that, in spite of our occasional references to CFT, all of our proofs are rigorous, and none of them use assumptions from CFT.  Rather, CFT tells us what ought to be true, and we prove those facts using rigorous methods.)

\subsection{A survey of our approach}\label{survey}

In this section, we motivate our method to achieve goals \ref{item1}--\ref{item3} (stated in the previous section), momentarily restricting our attention to percolation ($\kappa=6$) \cite{grim,weintru,smir2} for simplicity.  To begin, we choose one of the $C_N$ available boundary arc connectivities in a $2N$-sided polygon $\mathcal{P}$ with vertices $w_1,$ $w_2,\ldots,w_{2N}\in\mathbb{C}$.  Topological considerations show that there are at least two sides of $\mathcal{P}$ whose two adjacent vertices are endpoints of a common boundary arc (i.e., multiple-SLE$_\kappa$ curve), and we let $[w_i,w_{i+1}]$ be such a side.

Next, we investigate what happens as the vertices $w_i$ and $w_{i+1}$ approach each other. What we find depends on the boundary condition for $[w_i,w_{i+1}]$.  If this side is wired, then the adjacent free sides $[w_{i-1},w_i]$ and $[w_{i+1},w_{i+2}]$ fuse into one contiguous free side $[w_{i-1},w_{i+2}]$ of a $(2N-2)$-sided polygon $\mathcal{P}'$, and the isolated boundary cluster previously anchored to $[w_i,w_{i+1}]$ contracts away.  Or if $[w_i,w_{i+1}]$ is free, then the adjacent wired sides fuse into one contiguous wired side of $\mathcal{P}'$, and the boundary clusters previously anchored to these sides fuse into one boundary cluster anchored to this new wired side.  In either situation, the original crossing configuration for the $2N$-sided polygon goes to a crossing configuration for a $(2N-2)$-sided polygon, and the connectivity weight for the former configuration goes to the connectivity weight for the latter configuration (figure \ref{CollapseSeq}).  (In percolation, ``connectivity weight" and ``crossing-probability formula" are synonymous.)  If we repeat this process $N-1$ more times, then we end with a zero-sided polygon, or disk, whose boundary is either all wired or all free.  The disk trivially exhibits just one ``boundary arc connectivity" with zero curves, so this cumulative process sends the original connectivity weight to one.

\begin{figure}[b!]
\centering
\includegraphics[scale=0.23]{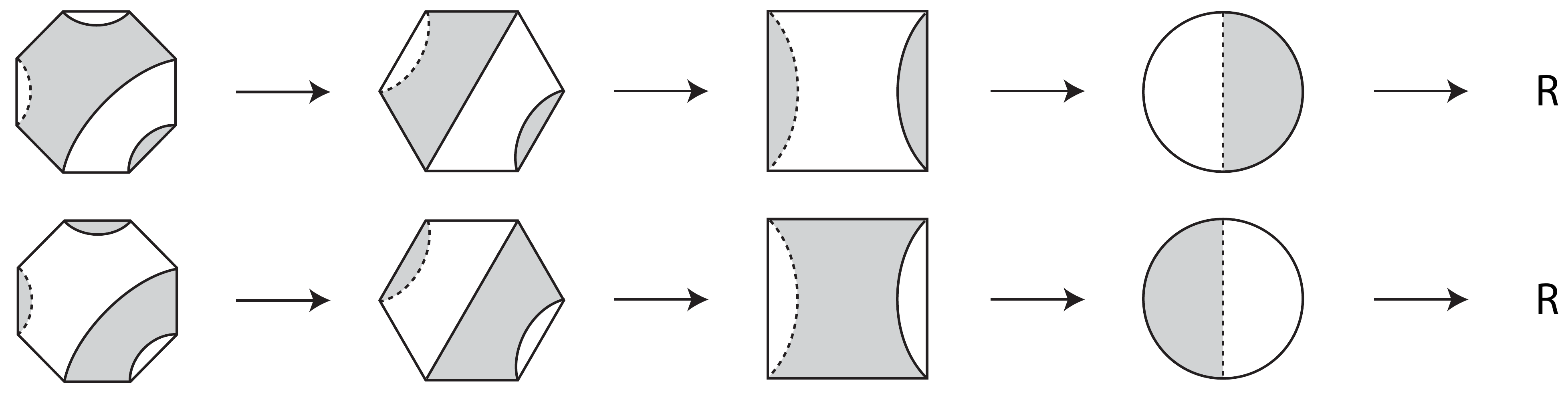}
\caption{The map that sends an octagon connectivity weight to a nonzero real number.  In each step, we bring together the endpoints of the dashed boundary arc.}\label{CollapseSeq}
\end{figure}

Now, we might bring together in a different order the same pairs of vertices that approach each other in the previous paragraph.  But because all of these variations send the same connectivity weight to one, we anticipate that all of them are different realizations of the same map.

Next, we investigate what happens if two adjacent vertices $w_i$ and $w_{i+1}$ that are not endpoints of a common boundary arc approach each other.  Using the same boundary arc connectivity as before, we observe one of two outcomes.  If the side $[w_i,w_{i+1}]$ is wired, then the adjacent free sides $[w_{i-1},w_i]$ and $[w_{i+1},w_{i+2}]$ of $\mathcal{P}$ do not fuse into one contiguous free segment.  Instead, they remain separated by an infinitesimal wired segment centered on the point $w_i=w_{i+1}$ within the side $[w_{i-1},w_{i+2}]$ of $\mathcal{P}'$, and the boundary cluster previously anchored to $[w_i,w_{i+1}]$ now anchors to this infinitesimal segment.  Or if $[w_i,w_{i+1}]$ is free, then the adjacent wired sides do not fuse into one contiguous wired segment.  Instead, they remain separated by an infinitesimal free segment centered on $w_i$, and the boundary clusters that originally anchored to the adjacent wired sides remain separated by this segment.  The likelihood of witnessing either of these two configurations in $\mathcal{P}'$, with respect to the point $w_i$ on its boundary, is zero.  Hence, pulling together two vertices not connected by a common boundary arc sends the connectivity weight for the original configuration in $\mathcal{P}$ to zero.  In CFT, this corresponds to the appearance of only the two-leg fusion channel in the OPE of the one-leg corner operators at $w_i$ and $w_{i+1}$. (See appendix \ref{preliminaries}.)

\begin{figure}[t!]
\centering
\includegraphics[scale=0.3]{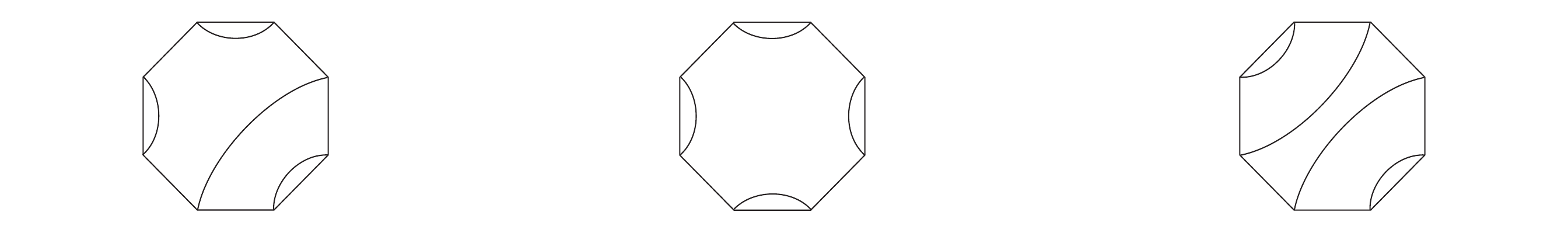}
\caption{Three boundary arc connectivity diagrams for the octagon.  We find the other $C_4-3=11$ diagrams by rotating one of these three.}
\label{Diagrams}
\end{figure}

All of these mappings that pull pairs of adjacent vertices of $\mathcal{P}$ together are subject to one constraint: the vertices $w_{i_{2j-1}}$ and $w_{i_{2j}}$ of the $2(N-j+1)$-sided polygon $\mathcal{P}_j$ to be brought together at the $j$th step of this mapping cannot be separated from each other by any other vertices within the boundary of $\mathcal{P}_j$.  If we imagine connecting $w_{i_{2j-1}}$ and $w_{i_{2j}}$ with an arc crossing the interior of $\mathcal{P}$ for each $j\in\{1,2,\ldots,N\}$, then this condition is satisfied if and only if we may draw these arcs in $\mathcal{P}$ so they do not intersect.  Furthermore, two mappings that bring the same pairs of vertices together in a different order have the same such arc connectivities, of which there are $C_N$ (\ref{catalan}) (figure \ref{Diagrams}).  Now, if changing the order in which we collapse the sides of $\mathcal{P}$ does not change the image of any of these mappings, then there are effectively only $C_N$ distinct mappings.  Assuming that the order indeed does not matter, we enumerate these mappings as we enumerated their corresponding boundary arc connectivities, we denote the $\varsigma$th of them by $[\mathscr{L}_\varsigma]$, and we let $\mathscr{B}_N^*:=\{[\mathscr{L}_1],[\mathscr{L}_2],\ldots,[\mathscr{L}_{C_N}]\}.$

If the arc connectivity for some specified mapping in $\mathscr{B}_N^*$ matches (resp.\ does not match) the boundary arc connectivity for some specified connectivity weight, then our arguments imply that this mapping sends that connectivity weight to one (resp.\ zero).  Hence, we anticipate that $[\mathscr{L}_\varsigma]\Pi_\vartheta=\delta_{\varsigma,\vartheta}$ for all $\varsigma,\vartheta\in\{1,2,\ldots,C_N\}$ and, assuming that $\mathcal{S}_N^*$ is finite-dimensional with basis $\mathscr{B}_N^*$, the set $\mathscr{B}_N$ of connectivity weights is the basis of $\mathcal{S}_N$ dual to $\mathscr{B}_N^*$.  (Therefore, in the sequel \cite{florkleb4}, we define the $\varsigma$th connectivity weight $\Pi_\varsigma$ to be the element of $\mathcal{S}_N$ dual to $[\mathscr{L}_\varsigma]\in\mathscr{B}_N^*$.  Both here and in that article, we assume that this precise definition agrees with the multiple-SLE$_\kappa$ definition given above.)

If $\mathscr{B}_N^*$ is a basis for $\mathcal{S}_N^*$, then this duality relation implies that $[\mathscr{L}_\varsigma]F$ is the coefficient of the $\varsigma$th connectivity weight $\Pi_\varsigma$ in the decomposition of $F\in\mathcal{S}_N$ over $\mathscr{B}_N$.  Furthermore, the linear mapping $v:\mathcal{S}_N\rightarrow\mathbb{R}^{C_N}$ with the $\varsigma$th coordinate of $v(F)$ equaling $[\mathscr{L}_\varsigma]F$ is a bijection.  That is, although $v$ destroys the pointwise information contained in each element of $\mathcal{S}_N$, it preserves the linear relations between these elements.

This reasoning motivates our strategy for proving goals \ref{item1}--\ref{item3} in section \ref{objorg} but in what follows, we order our steps differently because we do not know how to prove that $\mathscr{B}_N^*$ is a basis for $\mathcal{S}_N^*$ \emph{a priori}.  So after constructing the elements of $\mathscr{B}_N^*$ in sections \ref{boundary behavior} and \ref{dualspace}, we prove that the linear mapping $v$ is injective in section \ref{upperbound} (deferring there the proof of lemma \ref{alltwoleglem} to \cite{florkleb2}).   Then the dimension theorem of linear algebra bounds the dimension of $\mathcal{S}_N$ by $C_N$.  Finally, to prove that the dimension of $\mathcal{S}_N$ is indeed $C_N$ in \cite{florkleb3}, we use the Coulomb gas formalism to construct $C_N$ explicit elements of $\mathcal{S}_N$ and then use results of this article to prove that they are linearly independent.  That $\mathscr{B}_N^*$ is a basis for $\mathcal{S}_N$ follows from these results.

\section{Boundary behavior of solutions}\label{boundary behavior}

Motivated by the observations of section \ref{survey}, we investigate the behavior of elements of $\mathcal{S}_N$ near certain points in the boundary of $\Omega_0$.  (See the discussion surrounding (\ref{components}).)  If we conformally map the $2N$-sided polygon $\mathcal{P}$ onto the upper half-plane, with its $i$th vertex $w_i$ sent to the $i$th coordinate of $\boldsymbol{x}\in\Omega_0$, then the action of bringing together the vertices $w_{i+1}$ and $w_i$ sends $\boldsymbol{x}$ to the boundary point $(x_1,x_2,\ldots,x_{i-1},x_i,x_i,x_{i+2},\ldots,x_{2N})\in\partial\Omega_0$.  This point is in the hyperplane within $\partial\Omega_0$, whose points have only the $i$th and $(i+1)$th coordinates equal.  Hence, to implement the mappings described in the previous section for any $F\in\mathcal{S}_N$, we must study the limit of $F(\boldsymbol{x})$ as $x_{i+1}\rightarrow x_i$ for any $i\in\{1,2,\ldots,2N-1\}$ first.

Interpreting $F\in\mathcal{S}_N$ as a half-plane correlation function of $2N$ one-leg boundary operators (\ref{2Npoint}), we anticipate this limit using CFT. (These correlation functions appear, e.g., on the right side of (\ref{partratio}).  See appendix \ref{preliminaries} for further details and a review of the CFT nomenclature that we refer to here.)  We envisage the $i$th coordinate of $\boldsymbol{x}\in\Omega_0$ as hosting a one-leg boundary operator $\psi_1(x_i)$.  If we send $x_{i+1}\rightarrow x_i$ for some $i\in\{1,2,\ldots,2N-1\}$, then the operators $\psi_1(x_i)$ and $\psi_1(x_{i+1})$ fuse into some combination of an identity operator $\psi_0(x_i)$ (which is actually independent of $x_i$) and a two-leg boundary operator at $\psi_2(x_i)$.  After inserting their OPE into the $2N$-point function $F$, we find the Frobenius series expansion
\bea\label{onelegcorr}F(x_1,x_2,\ldots,x_{2N})&=&\langle\psi_1(x_1)\psi_1(x_2)\dotsm\psi_1(x_i)\psi_1(x_{i+1})\dotsm\psi_1(x_{2N})\rangle\\
\label{0fuse}&=& C_{11}^0(x_{i+1}-x_i)^{-2\theta_1+\theta_0}\langle\psi_1(x_1)\psi_1(x_2)\dotsm\psi_0(x_i)\psi_1(x_{i+2})\dotsm\psi_{2N}(x_{2N})\rangle+\dotsm\\
\label{2fuse}&+& C_{11}^2(x_{i+1}-x_i)^{-2\theta_1+\theta_2}\langle\psi_1(x_1)\psi_1(x_2)\dotsm\psi_2(x_i)\psi_1(x_{i+2})\dotsm\psi_{2N}(x_{2N})\rangle+\dotsm.\eea
Here, $C_{11}^0$ and $C_{11}^2$ are arbitrary real constants (for our present purposes), $\theta_s$ is the conformal weight of the $s$-leg boundary operator (\ref{bdysleg})
\be\label{bdyslegformula}\theta_s=\frac{s(2s+4-\kappa)}{2\kappa},\ee
Below, (\ref{p}) gives the powers appearing in (\ref{0fuse}, \ref{2fuse}).  Loosely speaking, if appropriately normalized, we refer to either (\ref{0fuse}) or (\ref{2fuse}) in CFT as a ``conformal block," and (\ref{0fuse}) and (\ref{2fuse}) correspond to the identity and two-leg fusion channels of the constituent one-leg boundary operators respectively 

Motivated by this interpretation of $F$, we suppose that any $F\in\mathcal{S}_N$ has a Frobenius series expansion in $x_{i+1}$ centered on $x_i$:
\bea \label{F0}F(x_1,x_2,\ldots,x_{2N})&=&(x_{i+1}-x_i)^{-p}F_0(x_1,x_2,\ldots, x_i,x_{i+2},\ldots,x_{2N})\\
\label{F1}&+&(x_{i+1}-x_i)^{-p+1}F_1(x_1,x_2,\ldots, x_i,x_{i+2},\ldots,x_{2N})\\
\label{F2}&+&(x_{i+1}-x_i)^{-p+2}F_2(x_1,x_2,\ldots, x_i,x_{i+2},\ldots,x_{2N})+\dotsm.\eea
Using the null-state PDEs (\ref{nullstate}) centered on $x_i$ and $x_{i+1}$ (i.e., with $j=i$ and $j=i+1$) and collecting the leading order contributions, we find from either equation that
\be\label{powers}\bigg[\frac{\kappa}{4}(-p)(-p-1)-p-\frac{6-\kappa}{2\kappa}\bigg]F_0=0.\ee
Solving this equation for $-p$, we find the two powers
\be \label{p}-p=\begin{cases}p_1= -2\theta_1+\theta_0=1-6/\kappa,& \quad\text{the identity channel,} \\  p_2= -2\theta_1+\theta_2=2/\kappa,&\quad \text{the two-leg channel,} \end{cases}\ee
that appear in (\ref{0fuse}) and (\ref{2fuse}) respectively.  In this article, we restrict our attention to the range $\kappa\in(0,8)$ over which $p_1=1-6/\kappa$ is the smaller power.  At the next order, the null-state PDEs centered on $x_i$ and $x_{i+1}$ respectively give
\bea\label{pfirst}\frac{\kappa p}{2}\partial_iF_0+\bigg[\frac{\kappa}{4}(-p+1)(-p)+(-p+1)-\frac{6-\kappa}{2\kappa}\bigg]F_1&=&0,\\
\label{psecond}-\partial_iF_0+\bigg[\frac{\kappa}{4}(-p+1)(-p)+(-p+1)-\frac{6-\kappa}{2\kappa}\bigg]F_1&=&0.\eea
Taking their difference when $-p=p_1=1-6/\kappa$ gives $\partial_i F_0=0$, so $F_1=0$ immediately follows if $\kappa\neq4$.  (We re-examine the case with $\kappa=4$ more closely in section \red{II} of \cite{florkleb4}.)  In the CFT language, the condition that $F_1=0$ is equivalent to the vanishing of the level-one descendant of the identity operator, and the condition $\partial_i F_0=0$ implies that the identity operator is nonlocal.  Comparing (\ref{F0}) with (\ref{0fuse}) suggest that we interpret $F_0$ as a $(2N-2)$-point function of one-leg boundary operators.  If this supposition is true, then $F_0$ must satisfy the system of PDEs (\ref{nullstate}, \ref{wardid}) in the coordinates of $\{x_j\}_{j\neq i,i+1}$ with $N$ replaced by $N-1$.  This observation echoes our previous claim that the limit $x_{i+1}\rightarrow x_i$ (previously $w_{i+1}\rightarrow w_i$) sends a connectivity weight for a $2N$-sided polygon to that of a $(2N-2)$-sided polygon.  If $-p=p_2=2/\kappa$ instead, then (\ref{pfirst}) and (\ref{psecond}) are identical, so $\partial_iF$ is typically not zero.  In the CFT language, this implies that the two-leg operator is local.  We focus our attention on the $-p=p_1$ case of (\ref{p}) for now and postpone consideration of the $-p= p_2$ case to \cite{florkleb2,florkleb4}.

These heuristic calculations suggest that for all $F\in\mathcal{S}_N$, if we let $x_{i+1}$ approach $x_i$ with the values of $x_i$ and the other coordinates of $\boldsymbol{x}\in\Omega_0$ fixed, then $F(\boldsymbol{x})$ either grows or decays with power $1-6/\kappa$ or greater.  Lemma \ref{boundedlem} below establishes this fact, but before we prove it, we introduce some convenient notation.

\begin{defn}\label{pidef} We define $\pi_i:\Omega\rightarrow\mathbb{R}^{2N-1}$ and $\pi_{ij}:\Omega\rightarrow\mathbb{R}^{2N-2}$ to be the projection maps removing the $i$th coordinate and both the $i$th and $j$th coordinates respectively from $\boldsymbol{x}\in\Omega$:
\be\begin{aligned} \pi_i(\boldsymbol{x})&=(x_1,x_2,\ldots,x_{i-1},x_{i+1},\ldots,x_{2N}),\\
 \pi_{ij}(\boldsymbol{x})&=(x_1,x_2,\ldots,x_{i-1},x_{i+1},\ldots,x_{j-1},x_{j+1},\ldots,x_{2N}).\end{aligned}\ee
More generally, we define $\pi_{i_1,i_2,\ldots i_M}:\Omega\rightarrow\mathbb{R}^{2N-M}$ to be the projection map removing the coordinates with indices $i_1,$ $i_2,\ldots,i_M$ from $\boldsymbol{x}\in\Omega$. 
\end{defn}
\noindent
In this article, we often identify $\pi_{i+1}(\Omega_0)$ with the subset of the boundary of $\Omega_0$ whose points have only two coordinates, the $i$th and the $(i+1)$th, equal.  Furthermore, we sometimes identify $\pi_{i,i+1}(\Omega_0)$, explicitly given by
 \be\pi_{i,i+1}(\Omega_0):=\{(x_1,x_2,\ldots,x_{i-1},x_{i+2},\ldots,x_{2N})\in\mathbb{R}^{2N-2}\,|\,x_1<x_2<\ldots<x_{i-1}<x_{i+2}<\ldots<x_{2N}\},\ee
with the subset of the boundary of $\Omega_0$ whose points have only three coordinates, the $(i-1)$th, the $i$th and the $(i+1)$th, equal.

\begin{lem}\label{boundedlem}  Suppose that $\kappa\in(0,8)$ and $F\in\mathcal{S}_N$, and for some $i\in\{1,2,\ldots,2N-1\}$, let
\be\label{xdelta}\boldsymbol{x}_\delta:=(x_1,x_2,\ldots,x_i,x_i+\delta,x_{i+2},\ldots,x_{2N}).\ee
Then for all $j,k\not\in\{i,i+1\}$ and any compact subset $\mathcal{K}$ of $\pi_{i+1}(\Omega_0)$, the supremums
\be\label{SNconditions}\sup_\mathcal{K}|F(\boldsymbol{x}_\delta)|, \quad\sup_\mathcal{K}|\partial_j F(\boldsymbol{x}_\delta)|, \quad \sup_\mathcal{K}|\partial_j\partial_kF(\boldsymbol{x}_\delta)|\ee
are $O(\delta^{1-6/\kappa})$ as $\delta\downarrow0$. 
\end{lem}

\begin{proof} For each point $\boldsymbol{x}_\delta\in\Omega_0$, we let $x:=x_i$, we relabel the coordinates in $\{x_j\}_{j\neq i,i+1}$ as $\{\xi_1,\xi_2,\ldots,\xi_{2N-2}\}$ in increasing order, and we let $\boldsymbol{\xi}:=(\xi_1,\xi_2,\ldots,\xi_{2N-2})$.  With this new notation, we define
\be\label{tilde} {\rm F}(\boldsymbol{\xi};x,\delta):=F(\xi_1,\xi_2,\ldots,\xi_{i-1},x,x+\delta,\xi_i,\ldots,\xi_{2N-2}).\ee
Finally, we choose an arbitrary compact subset $\mathcal{K}\subset\pi_{i+1}(\Omega_0)$, and with $p>0$ determined by (\ref{powerlaw}), we choose bounded open sets $ U_0,$ $ U_1,\ldots, U_m$, with $m:=\lceil q\rceil$ and $q:=p+1-6/\kappa$, such that they are sequentially compactly embedded:
\be\label{compactembedd}\mathcal{K}\subset\subset U_m\subset\subset U_{m-1}\subset\subset\ldots\subset\subset U_0\subset\subset\pi_{i+1}(\Omega_0).\ee
(We choose $p$ large enough so $m>2$.)  Our goal is to prove that the quantities in (\ref{SNconditions}) are $O(\delta^{1-6/\kappa})$ as $\delta\downarrow0$.

To begin, we write the null-state PDE (\ref{nullstate}) centered on $x_{i+1}$ as $\mathcal{L}[ {\rm F}]=\mathcal{M}[ {\rm F}]$, where $\mathcal{L}[{\rm F}]$ contains all of the terms that are seemingly largest when $\delta\downarrow0$.  With $\partial_j$ referring to a derivative with respect to $\xi_j$, we have
\bea\label{firsti+1}\mathcal{L}[{\rm F}](\boldsymbol{\xi};x,\delta)&:=&\left[\frac{\kappa}{4}\partial_\delta^2+\frac{\partial_\delta}{\delta}-\frac{(6-\kappa)/2\kappa}{\delta^2}\right]{\rm F}(\boldsymbol{\xi};x,\delta),\\
\label{secondi+1}\mathcal{M}[{\rm F}](\boldsymbol{\xi};x,\delta)&:=&\Bigg[\frac{\partial_x}{\delta}+\sum_{j=1}^{2N-2}\left(\frac{(6-\kappa)/2\kappa}{(\xi_j-x-\delta)^2}-\frac{\partial_j}{\xi_j-x-\delta}\right)\Bigg]{\rm F}(\boldsymbol{\xi};x,\delta).\eea
Thinking of $\delta$ as a time variable propagating backwards from an initial condition at some $0<b<\inf_\mathcal{K}(\xi_i-x)$ to zero, we invert the Euler differential operator $\mathcal{L}$ with a Green function that satisfies the adjoint problem 
\be\label{adjointL}\mathcal{L}^*[G](\delta,\eta):=\left[\frac{\kappa}{4}\partial_\eta^2-\frac{\partial_\eta}{\eta}-\frac{(6-\kappa)/2\kappa-1}{\eta^2}\right]G(\delta,\eta)=D(\eta-\delta),\quad G(\delta,0)=\partial_\eta G(\delta,0)=0,\ee
with $D$ the Dirac delta function, and with $\delta>0$ a parameter.  In order to satisfy (\ref{adjointL}), $G$ must be continuous at $\eta=\delta$, and $\partial_{\eta}G$ must have a jump discontinuity of $4/\kappa$ at $\eta=\delta$.  The unique solution to this initial value problem is
\be\label{greenfunc}G(\delta,\eta)=\frac{4\eta}{8-\kappa}\Theta(\eta-\delta)\left[\left(\frac{\delta}{\eta}\right)^{1-6/\kappa}-\hspace{.2cm}\left(\frac{\delta}{\eta}\right)^{2/\kappa}\right]\ee
for all $0<\delta,\eta<b$.  The Heaviside step function $\Theta$ enforces causality (i.e., $G(\delta,\eta)=0$ when $\eta\leq\delta$).  Using the usual Green identity \cite{folland}, we find that 
\be\label{secondline} {\rm F}(\boldsymbol{\xi};x,\delta)=\sideset{}{_\delta^b}\int G(\delta,\eta)\mathcal{M}[ {\rm F}](\boldsymbol{\xi};x,\eta)\,{\rm d}\eta-\frac{\kappa}{4}\left[G(\delta,b)\,\partial_\delta {\rm F}(\boldsymbol{\xi};x,b)-\partial_\eta G(\delta,b) {\rm F}(\boldsymbol{\xi};x,b)\right]-\frac{1}{b}G(\delta,b) {\rm F}(\boldsymbol{\xi};x,b)\ee
for all $0<\delta<b$.  All terms on the right side except the definite integral are manifestly $O(\delta^{1-6/\kappa})$ as $\delta\downarrow0$, so we only need to bound the former. After estimating the coefficients in the integrand and estimating $G$ for $\kappa<8$, we find
\be\label{abstildeF}| {\rm F}(\boldsymbol{\xi};x,\delta)|\leq \frac{4}{8-\kappa}\sideset{}{_\delta^b}\int \left(\frac{\delta}{\eta}\right)^{1-6/\kappa}|\eta\mathcal{M}[{\rm F}](\boldsymbol{\xi};x,\eta)|\,{\rm d}\eta+O(\delta^{1-6/\kappa})\ee
for all $0<\delta<b$.  It is then natural to define 
\be\label{FtoH}H(\boldsymbol{\xi};x,\delta):=\delta^{6/\kappa-1} {\rm F}(\boldsymbol{\xi};x,\delta),\ee
so proving the lemma amounts to showing that the supremums $|H(\boldsymbol{\xi};x,\delta)|,$ $|\partial_jH(\boldsymbol{\xi};x,\delta)|$, and $|\partial_j\partial_kH(\boldsymbol{\xi};x,\delta)|$ over $\mathcal{K}$ are bounded functions of $\delta\in(0,b)$.  After writing (\ref{abstildeF}) in terms of $H(\boldsymbol{\xi};x,\delta)$, we find
\be\label{estH}\sup_{U_n}|H(\boldsymbol{\xi};x,\delta)|\leq c_n+\frac{4}{8-\kappa}\sideset{}{_\delta^b}\int\sup_{U_n}|\eta\mathcal{M}[H](\boldsymbol{\xi};x,\eta)|\,{\rm d}\eta\ee
for all $0<\delta<b$ and some positive constants $c_n$.

Next, we bound terms in the integrand of (\ref{estH}) that contain derivatives of $H$.  For $j\not\in\{i,i+1\}$, the null-state PDE centered on $x_j$ becomes (now for $j\in\{1,2,\ldots,2N-2\}$)
\begin{multline}\label{notii+1}\Bigg[\frac{\kappa}{4}\partial_j^2+\sum_{k\neq j}\left(\frac{\partial_k}{\xi_k-\xi_j}-\frac{(6-\kappa)/2\kappa}{(\xi_k-\xi_j)^2}\right)+\frac{\partial_x}{x-\xi_j}-\frac{\delta\partial_\delta}{(x-\xi_j)(x+\delta-\xi_j)}\\
-\frac{(6-\kappa)/2\kappa}{(x-\xi_j)^2}-\frac{(6-\kappa)/2\kappa}{(x+\delta-\xi_j)^2}+\frac{6/\kappa-1}{(x-\xi_j)(x+\delta-\xi_j)}\Bigg] H(\boldsymbol{\xi};x,\delta)=0, \end{multline}
while that centered on $x_i$ becomes
\be\label{ith}\Biggl[\frac{\kappa}{4}(\partial_x-\partial_\delta)^2+\frac{\partial_\delta}{\delta}+\frac{(6-\kappa)(\partial_x-\partial_\delta)}{2\delta}+\sum_k\left(\frac{\partial_k}{\xi_k-x}-\frac{(6-\kappa)/2\kappa}{(\xi_k-x)^2}\right)\Biggl] H(\boldsymbol{\xi};x,\delta)=0,\ee
and that centered on $x_{i+1}$ becomes
\be\label{i+1th}\Bigg[\frac{\kappa}{4}\partial_\delta^2-\frac{(\partial_x-\partial_\delta)}{\delta}-\frac{(6-\kappa)\partial_\delta}{2\delta}+\sum_k\left(\frac{\partial_k}{\xi_k-x-\delta}-\frac{(6-\kappa)/2\kappa}{(\xi_k-x-\delta)^2}\right)\Bigg] H(\boldsymbol{\xi};x,\delta)=0.\ee
Also, the three conformal Ward identities (\ref{wardid}) become
\begin{align}\label{w1}&\bigg[\sideset{}{_k}\sum\partial_k+\partial_x\bigg]H(\boldsymbol{\xi};x,\delta)=0,\\
\label{w2}&\bigg[\sideset{}{_k}\sum(\xi_k\partial_k+(6-\kappa)/2\kappa)+x\partial_x+\delta\partial_\delta\bigg]H(\boldsymbol{\xi};x,\delta)=0,\\
\label{w3}&\bigg[\sideset{}{_k}\sum(\xi_k^2\partial_k+(6-\kappa)\xi_k/\kappa)+x^2\partial_x +(2x+\delta)\delta\partial_\delta\bigg]H(\boldsymbol{\xi};x,\delta)=0.\end{align}
Summing (\ref{notii+1}) over $j\in\{1,2,\ldots,2N-2\}$ and using (\ref{w1}, \ref{w2}) to isolate $\partial_xH$ and $\delta\partial_\delta H$ in terms of $H$ and its other derivatives, we find that $H$ obeys a strictly elliptic linear PDE in the coordinates of $\boldsymbol{\xi}$ and with $x$ and $\delta$ as parameters:
\begin{multline}\label{prePDE}\sum_j\Bigg[\frac{\kappa}{4}\partial_j^2+\sum_{k\neq j}\left(\frac{\partial_k}{\xi_k-\xi_j}-\frac{(6-\kappa)/2\kappa}{(\xi_k-\xi_j)^2}\right)-\sum_k\frac{\partial_k}{x-\xi_j}-\sum_k\frac{(x-\xi_k)\partial_k}{(x-\xi_j)(x+\delta-\xi_j)}\\
+\frac{(N-1)(6/\kappa-1)}{(x-\xi_j)(x+\delta-\xi_j)}-\frac{(6-\kappa)/2\kappa}{(x-\xi_j)^2}-\frac{(6-\kappa)/2\kappa}{(x+\delta-\xi_j)^2}+\frac{6/\kappa-1}{(x-\xi_j)(x+\delta-\xi_j)}\Bigg]H(\boldsymbol{\xi};x,\delta)=0.\end{multline}
The coefficients of this PDE are bounded over $U_n$ and do not blow up or vanish as $\delta\downarrow0$.  Therefore, the Schauder interior estimate (Cor.\ 6.3 of \cite{giltru}) implies that with $d_n:=\text{dist}(\partial U_n,\partial U_{n-1})$, the inequality
\be\label{Schauder}d_{n+1}\sup_{ U_{n+1}}|\partial^\varpi H(\boldsymbol{\xi};x,\delta;y)|\leq C_n\sup_{ U_n}|H(\boldsymbol{\xi};x,\delta;y)|\quad\begin{array}{l}\\ \end{array}\ee
holds for all $0<\delta<b$, where $C_n$ is some constant and $\varpi$ is any multi-index for the coordinates of $\boldsymbol{\xi}$ with length $|\varpi|\leq2$.  The conformal Ward identities (\ref{w1}, \ref{w2}) imply that $\varpi$ may involve $x$ and the derivatives $\delta\partial_\delta$ too. That is,
\be\label{multiindex}\partial^\varpi\in\left\{\begin{array}{c}\partial_j,\quad \partial_x,\quad \delta\partial_\delta, \quad\partial_j^2,\quad\partial_x^2,\\ (\delta\partial_\delta)^2,\quad\partial_j\partial_k, \quad \partial_j\partial_x,\quad\partial_j\delta\partial_\delta,\quad\partial_x\delta\partial_\delta \end{array}\right\}.\ee

Equations (\ref{powerlaw}) and (\ref{FtoH}) imply that both sides of (\ref{Schauder}) with $n=0$ are $O(\delta^{-q})$ as $\delta\downarrow0$, where $q:=p+1-6/\kappa$.  Now we improve these bounds.  After inserting (\ref{Schauder}) with $n=0$ into (\ref{secondi+1}, \ref{estH}) with $n=1$ and integrating, we find that
\be\label{thefirstest}\sup_{U_1}|H(\boldsymbol{\xi};x,\delta)|=O(\delta^{-q+1})\quad\Longrightarrow\quad\sup_{U_2}|\partial^\varpi H(\boldsymbol{\xi};x,\delta)|=O(\delta^{-q+1}),\ee
with the right estimate of (\ref{thefirstest}) following from (\ref{Schauder}).  After repeating this process another $m-2$ times, we find that the left side of (\ref{estH}) with $n=m-1$ is $O(\delta^{-q+m-1})$ as $\delta\downarrow0$.  Repeating one last time finally gives
\be\label{supU}\sup_{ U_m}|H(\boldsymbol{\xi};x,\delta)|=O(1)\quad\text{as $\delta\downarrow0$}\quad\Longrightarrow\quad\sup_{\mathcal{K}}|\partial^\varpi H(\boldsymbol{\xi};x,\delta)|=O(1)\quad\text{as $\delta\downarrow0$},\ee
thanks to the compact embedding $\mathcal{K}\subset\subset U_m$ (\ref{compactembedd}).  Equation (\ref{supU}) with (\ref{tilde}, \ref{FtoH}, \ref{multiindex}) then implies (\ref{SNconditions}).
\end{proof}

The proof of lemma \ref{boundedlem} never uses the null-state PDE centered on $x_i$ (\ref{nullstate}) or the third conformal Ward identity (\ref{wardid}).  Therefore, the lemma is true for not just $F\in\mathcal{S}_N$ but also for $F:\Omega_0\rightarrow\mathbb{R}$ satisfying only the other $2N-1$ null-state PDEs of (\ref{nullstate}), the first two conformal Ward identities of (\ref{wardid}), and the power-law bound (\ref{powerlaw}).  In spite of their  omission here, we use both of these PDEs in the proofs of some lemmas below in this article and in \cite{florkleb2}.

Having proven lemma \ref{boundedlem}, we prove next what the analysis that precedes that lemma suggests: the limit of $(x_{i+1}-x_i)^{6/\kappa-1}F(\boldsymbol{x})$ as $x_{i+1}\rightarrow x_i$ exists and is independent of $x_i$.  If this limit is not (resp.\ is) zero, then in CFT parlance, we say that the identity operator appears (resp.\ does not appear) in the OPE of $\psi_1(x_{i+1})$ with $\psi_1(x_i)$.  (Unlike the proof of lemma \ref{boundedlem}, we use the null-state PDE centered on $x_i$ to prove this lemma.)
\begin{lem}\label{limitlem}Suppose that $\kappa\in(0,8)$ and $F\in\mathcal{S}_N$, and let $\boldsymbol{x}_\delta$ be defined as in (\ref{xdelta}).  Then for all $j\not\in\{i,i+1\}$, the limits
\be\label{Fprimed}\lim_{\delta\downarrow0}\delta^{6/\kappa-1}F(\boldsymbol{x}_\delta),\quad\lim_{\delta\downarrow0}\delta^{6/\kappa-1}\partial_jF(\boldsymbol{x}_\delta),\quad\lim_{\delta\downarrow0}\delta^{6/\kappa-1}\partial_j^2F(\boldsymbol{x}_\delta)\ee
exist and are approached uniformly over every compact subset of $\pi_{i+1}(\Omega_0)$.  Furthermore, the limit of $\delta^{6/\kappa-1}F(\boldsymbol{x}_\delta)$ as $\delta\downarrow0$ does not depend on $x_i$.
\end{lem}

\begin{proof} With $H,\boldsymbol{\xi},x,\delta$, and $b$ defined in the proof of lemma \ref{boundedlem}, we first show that $H(\boldsymbol{\xi};x,\delta)$ has a limit as $\delta\downarrow0$.  Because $H(\boldsymbol{\xi};x,\delta)$ is bounded on $0<\delta<b$, we must only show that its superior and inferior limits are equal.  Now, we find
\be\label{tildeH} H(\boldsymbol{\xi};x,\delta)=H(\boldsymbol{\xi};x,b)-\frac{\kappa}{4}J(\delta,b)\partial_\delta H(\boldsymbol{\xi};x,b)+\sideset{}{_\delta^b}\int J(\delta,\eta)\mathcal{M}[H](\boldsymbol{\xi};x,\eta)\,{\rm d}\eta,\ee
for all $0<\delta<b$ from substituting (\ref{FtoH}) into (\ref{secondline}), where we define $\mathcal{M}$ in (\ref{secondi+1}), and where $J(\delta,\eta)$ is the Green function (\ref{greenfunc}) multiplied by $(\delta/\eta)^{6/\kappa-1}$:
\be\label{Jgreenfunc}J(\delta,\eta)=\frac{4\eta}{8-\kappa}\Theta(\eta-\delta)\left[1-\left(\frac{\delta}{\eta}\right)^{8/\kappa-1}\right].\ee
According to lemma \ref{boundedlem}, $\eta\mathcal{M}[H](\boldsymbol{\xi};x,\eta)$ is bounded on $0<\eta<b$, and because $\kappa<8$, the bracketed factor of (\ref{Jgreenfunc}) is less than one.  Thus, after estimating (\ref{tildeH}) and taking the supremum over this range, we find that
\be\label{supdiff}\sup_{0<\delta<b}| H(\boldsymbol{\xi};x,\delta)- H(\boldsymbol{\xi};x,b)|\leq\frac{\kappa}{8-\kappa}|b\,\partial_\delta H(\boldsymbol{\xi};x,b)|+\frac{4b}{8-\kappa}\sup_{0<\eta<b}|\eta\mathcal{M}[H](\boldsymbol{\xi};x,\eta)|.\ee
Next, we show that $b\hspace{.03cm}\partial_\delta H(\boldsymbol{\xi};x,b)$ vanishes as $b\downarrow0$.  After differentiating (\ref{tildeH}) with respect to $\delta$ (because $H$ is analytic, we may differentiate under the integral sign), we find that for all $0<\delta<b$,
\be\delta\label{partialyH}\partial_\delta H(\boldsymbol{\xi};x,\delta)=\left(\frac{\delta}{b}\right)^{8/\kappa-1}b\hspace{.03cm}\partial_\delta H(\boldsymbol{\xi};x,b)-\frac{4}{\kappa}\sideset{}{_\delta^b}\int\left(\frac{\delta}{\eta}\right)^{8/\kappa-1}\eta\mathcal{M}[H](\boldsymbol{\xi};x,\eta)\,{\rm d}\eta.\ee
Then after taking the magnitude of both sides of (\ref{partialyH}), estimating the magnitude of $\eta\mathcal{M}[H](\boldsymbol{\xi};x,\eta)$ by its supremum over $0<\delta<b$, and integrating, we find an upper bound of $|\delta\partial_\delta H(\boldsymbol{\xi};x,\delta)|$ that vanishes as $\delta\downarrow0$:
\bea\label{gotozero}&&|\delta\partial_\delta H(\boldsymbol{\xi};x,\delta)|\leq\left(\frac{\delta}{b}\right)^{8/\kappa-1}|b\hspace{.03cm}\partial_\delta H(\boldsymbol{\xi};x,b)|+\frac{4b}{\kappa}\sup_{0<\eta<b}|\eta\mathcal{M}[H](\boldsymbol{\xi};x,\eta)|\left(\frac{(\delta/b)^{8/\kappa-1}+\delta/b}{|2-8/\kappa|}\right)\\
\label{limepsilon}\Longrightarrow\quad&&|\delta\partial_\delta H(\boldsymbol{\xi};x,\delta)|\xrightarrow[\delta\downarrow0]{}0.\eea 
(If $\kappa=4$, then $(\delta/b)\log(b/\delta)$ replaces the last factor in (\ref{gotozero}), and (\ref{limepsilon}) follows again.)  Hence, both sides of (\ref{supdiff}) vanish as $b\downarrow0$, the inferior and superior limits of $H(\boldsymbol{\xi};x,\delta)$ as $\delta\downarrow0$ are equal, and this desired limit indeed exists: 
\be H(\boldsymbol{\xi};x,0):=\lim_{\delta\downarrow0}H(\boldsymbol{\xi};x,\delta).\ee

By replacing $\delta$ with zero and then $b$ with $\delta$ in (\ref{tildeH}) and taking the supremum over a compact subset $\mathcal{K}$ of $\pi_{i+1}(\Omega_0)$, we find that for all $0<\delta<b$,
\be\label{suppartial}\sup_\mathcal{K}|H(\boldsymbol{\xi};x,\delta)-H(\boldsymbol{\xi};x,0)|\leq\frac{\kappa}{8-\kappa} \sup_\mathcal{K}|\delta\partial_\delta H(\boldsymbol{\xi};x,\delta)|+\frac{4}{8-\kappa}\sideset{}{_0^\delta}\int\sup_\mathcal{K}|\eta\mathcal{M}[H](\boldsymbol{\xi};x,\eta)|\,{\rm d}\eta.\ee
Lemma \ref{boundedlem} implies that the integrand of (\ref{suppartial}) is bounded over $0<\eta<\delta$, so the definite integral vanishes as $\delta\downarrow0$.  Next, after taking the supremum of (\ref{gotozero}) over $\mathcal{K}$ and using lemma \ref{boundedlem} again, we see that the supremum of $|\delta\partial_\delta H(\boldsymbol{\xi};x,\delta)|$ over $\mathcal{K}$ vanishes as $\delta\downarrow0$.  Hence, the left side of (\ref{suppartial}) vanishes as $\delta\downarrow0$, proving that the limit $H(\boldsymbol{\xi};x,0)$ is approached uniformly over $\mathcal{K}$. Because $\mathcal{K}$ may be any compact subset of $\pi_{i+1}(\Omega_0)\subset\partial\Omega_0$, it follows that $H$ extends continuously to $\Omega_0\cup\pi_{i+1}(\Omega_0)$ if we naturally define $H(\boldsymbol{\xi},x):=H(\boldsymbol{\xi};x,0)$ for all $(\boldsymbol{\xi},x)\in\pi_{i+1}(\Omega_0)$.

We recycle these arguments to show that for all $j\in\{1,2,\ldots,2N-2\}$, $\partial_j H(\boldsymbol{\xi};x,\delta)$ and $\partial_j^2H(\boldsymbol{\xi};x,\delta)$ approach limits as $\delta\downarrow0$ uniformly over compact subsets of $\pi_{i+1}(\Omega_0)$.  By taking the $j$th partial derivative of (\ref{tildeH}) (because $H$ is analytic in $\Omega_0$, we may exchange the order of integration and differentiation), we find an equation similar to (\ref{tildeH}) but with a few changes.  First, the supremum of the new integrand over compact subsets of $\pi_{i+1}(\Omega_0)$ is also bounded on $0<\eta<b$, according to lemma \ref{boundedlem}.  Second, $\delta\partial_\delta H(\boldsymbol{\xi};x,\delta)$ in the first term of (\ref{tildeH}) is replaced with $\delta\partial_\delta\partial_j H(\boldsymbol{\xi};x,\delta)$.  By taking the $j$th partial derivative of (\ref{partialyH}) and following the reasoning that led to (\ref{gotozero}, \ref{limepsilon}), we find that the supremum of $\delta\partial_\delta\partial_j H(\boldsymbol{\xi};x,\delta)$ over compact subsets of $\pi_{i+1}(\Omega_0)$ vanishes as $\delta\downarrow0$ too.  Hence, we may reuse all of the reasoning presented above to show that $\partial_j H(\boldsymbol{\xi};x,\delta)$ approaches a limit as $\delta\downarrow0$ uniformly over compact subsets of $\pi_{i+1}(\Omega_0)$.  Finally, we use the null-state PDE centered on $x_j$ (\ref{notii+1}) and (\ref{w1}, \ref{w2}) to isolate $\partial_j^2 H(\boldsymbol{\xi};x,\delta)$ in terms of lower-order derivatives $\partial_k H(\boldsymbol{\xi};x,\delta)$ and $H(\boldsymbol{\xi};x,\delta)$ and thus prove that the $j$th second derivative of $H$ approaches a limit as $\delta\downarrow0$ uniformly over compact subsets of $\pi_{i+1}(\Omega_0)$.  Thus, we have proven that the limits in (\ref{Fprimed}) exist and are approached uniformly over compact subsets of $\pi_{i+1}(\Omega_0)$.

Finally, to prove that the limit $H(\boldsymbol{\xi};x):=H(\boldsymbol{\xi};x,0)$ does not depend on $x$, it suffices to show that $\partial_xH(\boldsymbol{\xi};x,\delta)$ vanishes as $\delta\downarrow0$.  Indeed, if this is true, then by sending $\delta\downarrow0$ in (\ref{w1}) and using the uniformness of the limits in (\ref{Fprimed}) to commute the derivative with respect to $x$ with this limit, we find that $\partial_xH(\boldsymbol{\xi};x)=0$ too.  To prove that $H(\boldsymbol{\xi};x,\delta)$ vanishes as $\delta\downarrow0$, we subtract (\ref{i+1th}) from (\ref{ith}) to find the following PDE:
\be\label{diffpde}\left[\frac{\kappa}{4}\partial_x-\frac{\kappa}{2}\partial_\delta+\frac{8-\kappa}{2\delta}\right]\partial_x H(\boldsymbol{\xi};x,\delta)=
\sum_k\left[\frac{\delta\partial_k}{(\xi_k-x)(\xi_k-x-\delta)}+\frac{[\delta+2(x-\xi_k)]\delta(6-\kappa)/2\kappa}{(\xi_k-x)^2(\xi_k-x-\delta)^2}\right] H(\boldsymbol{\xi};x,\delta).\ee
We choose a positive $a<\min\{x-\xi_{i-1},(b-\delta)/2\}$, and we let $Z(t):=\partial_x H(\boldsymbol{\xi};x-t,\delta+2t)$ for $t\in[0,a]$.  Upon evaluating (\ref{diffpde}) at $(\boldsymbol{\xi};x,\delta)\mapsto(\boldsymbol{\xi},x-t,\delta+2t)$ and multiplying both sides by $-4(\delta+2t)^{1-8/\kappa}/\kappa$, we find
\begin{multline}\frac{{\rm d}}{{\rm d}t}\left[(\delta+2t)^{1-8/\kappa}Z(t)\right]=-\frac{4}{\kappa}(\delta+2t)^{2-8/\kappa}\sum_k\Bigg[\frac{\partial_k}{(\xi_k-x+t)(\xi_k-x-\delta-t)}\\
+\frac{[\delta+2(x-\xi_k)](6-\kappa)/2\kappa}{(\xi_k-x+t)^2(\xi_k-x-\delta-t)^2}\Bigg]H(\boldsymbol{\xi};x-t,\delta+2t).\end{multline}
Integrating both sides with respect to $t$ from 0 to $a$, we have
\begin{multline}\partial_x H(\boldsymbol{\xi};x,\delta)=\left(\frac{\delta}{\delta+2a}\right)^{8/\kappa-1}\partial_x H(\boldsymbol{\xi};x-a,\delta+2a)+\frac{4}{\kappa}\delta^{8/\kappa-1}\\
\times\sideset{}{_0^a}\int {\rm d}t\,(\delta+2t)^{2-8/\kappa}\sum_k\Bigg[\frac{\partial_k}{(\xi_k-x+t)(\xi_k-x-\delta-t)}+\frac{[\delta+2(x-\xi_k)](6-\kappa)/2\kappa}{(\xi_k-x+t)^2(\xi_k-x-\delta-t)^2}\Bigg]H(\boldsymbol{\xi};x-t,\delta+2t).\end{multline}
Because the sum inside of the integrand is bounded over $(t,\delta)\in[0,a]\times[0,b]$, we have that for some positive-valued function $\Phi(\boldsymbol{\xi},x,a,b)$,
\be|\partial_x H(\boldsymbol{\xi};x,\delta)|\leq\left(\frac{\delta}{\delta+2a}\right)^{8/\kappa-1}|\partial_x H(\boldsymbol{\xi};x-a,\delta+2a)|+\Phi(\boldsymbol{\xi},x,a,b)\delta^{8/\kappa-1}\sideset{}{_0^a}\int {\rm d}t\,(\delta+2t)^{2-8/\kappa},\ee
from which it immediately follows that
\be\label{twoterms}\partial_x H(\boldsymbol{\xi};x,\delta)=O(\delta^{8/\kappa-1})+O(\delta^2)\rightarrow0\quad\text{as $\delta\downarrow0$}.\ee
Thus, the limit of $H(\boldsymbol{\xi};x,\delta)$ as $\delta\downarrow0$ does not depend on $x$.\end{proof}

The last equation (\ref{twoterms}) of the preceding proof may be interpreted in CFT parlance as the two possible fusion channels (i.e., indicial powers for the Frobenius series in (\ref{0fuse}, \ref{2fuse})).  Indeed, if we divide (\ref{twoterms}) by $\delta^{6/\kappa-1}$ to revert from $H$ to $F$ (\ref{FtoH}), then the first term in (\ref{twoterms}) corresponds to the two-leg fusion channel because $-2\theta_1+\theta_2=2/\kappa=p_2$ (\ref{p}).  On the other hand, the power $-2\theta_1+\theta_0+2=3-6/\kappa$ of the second term is not $-2\theta_1+\theta_0=1-6/\kappa=p_1$ (\ref{p}), but it is still indicative of the identity channel for the following reason.  Just before stating lemma \ref{boundedlem}, we show that if $F$ admits a Frobenius series expansion (\ref{F0}--\ref{F2}) with leading power $-p=1-6/\kappa$, then the function $F_0$ in the leading term (\ref{F0}) is independent of $x_i$ and the function $F_1$ in the following term (\ref{F1}) is zero.  Thus, the leading term (\ref{F2}) in $\partial_iF$ has the indicial power $3-6/\kappa$.

In the discussion preceding lemma \ref{boundedlem}, the CFT interpretation of $F\in\mathcal{S}_N$ as a $2N$-point function of one-leg boundary operators (\ref{onelegcorr}) suggests that the limit of $(x_{i+1}-x_i)^{6/\kappa-1}F(\boldsymbol{x})$ as $x_{i+1}\rightarrow x_i$, if it is not zero, is a $(2N-2)$-point function of one-leg boundary operators.  This motivates the following lemma.  (In its proof, we employ a variant of the inversion map $x_i\mapsto-1/x_i$.  The third conformal Ward identity of (\ref{wardid}) derives in part from the covariance of $F\in\mathcal{S}_N$ under this map.  Here, we use  this identity for the first time in this article.)

\begin{lem}\label{SN-1lem}Suppose that $\kappa\in(0,8)$ and $F\in\mathcal{S}_N$.  Then for all $i\in\{1,2,\ldots,2N-1\}$, the function $F_0:\pi_{i,i+1}(\Omega_0)\rightarrow\mathbb{R}$ defined by 
\be\label{thefirstlim}(F_0\circ\pi_{i,i+1})(\boldsymbol{x}):=\lim_{x_i\rightarrow x_{i-1}}\,\,\,\lim_{x_{i+1}\rightarrow x_i}(x_{i+1}-x_i)^{6/\kappa-1}F(\boldsymbol{x})\ee
(the limit sending $x_i\rightarrow x_{i-1}$ is used only if $i>1$ and is trivial according to lemma \ref{limitlem}), is an element of $\mathcal{S}_{N-1}$, and the function $F_0:\pi_{1,2N}(\Omega_0)\rightarrow\mathbb{R}$ defined by
\be\label{thesecondlim}(F_0\circ\pi_{1,2N})(\boldsymbol{x}):=\lim_{R\rightarrow\infty}(2R)^{6/\kappa-1}F(x_1=-R,x_2,\ldots,x_{2N}=R)\ee
is also an element of $\mathcal{S}_{N-1}$.
\end{lem}

\begin{proof}
We define $H,\boldsymbol{\xi},x,\delta$, and $H(\boldsymbol{\xi};x,0)$ as in the proof of lemma \ref{boundedlem}.  To begin, we prove that the first limit (\ref{thefirstlim}) is in $\mathcal{S}_{N-1}$.  By sending $\delta\downarrow0$ in (\ref{notii+1}) and (\ref{w1}--\ref{w3}) and using the limits (\ref{limepsilon}, \ref{twoterms}), we find equations almost identical to the $2N-2$ null-state PDEs in the coordinates of $\boldsymbol{\xi}$,
\be\label{prenullstate}\Bigg[\frac{\kappa}{4}\lim_{\delta\downarrow0}\partial_j^2+\sum_{k\neq j}\left(\frac{1}{\xi_k-\xi_j}\lim_{\delta\downarrow0}\partial_k-\frac{(6-\kappa)/2\kappa}{(\xi_k-\xi_j)^2}\lim_{\delta\downarrow0}\right)\Bigg]H(\boldsymbol{\xi};x,\delta)=0,\ee
and the three conformal Ward identities also in the coordinates of $\boldsymbol{\xi}$:
\be\begin{gathered}\label{prewardid}\sum_k\lim_{\delta\downarrow0}\partial_k H(\boldsymbol{\xi};x,\delta)=0,\quad\sum_k\left[\xi_k\lim_{\delta\downarrow0}\partial_k+\bigg(\frac{6-\kappa}{2\kappa}\bigg)\lim_{\delta\downarrow0}\right]H(\boldsymbol{\xi};x,\delta)=0,\\
\sum_k\left[\xi_k^2\lim_{\delta\downarrow0}\partial_k+\bigg(\frac{6-\kappa}{\kappa}\bigg)\xi_k\lim_{\delta\downarrow0}\right]H(\boldsymbol{\xi};x,\delta)=0.\end{gathered}\ee
According to lemma \ref{limitlem}, $H(\boldsymbol{\xi};x,\delta)$ and each of its first and second derivatives approach their limits as $\delta\downarrow0$ uniformly over compact subsets of $\pi_{i+1}(\Omega_0)$.  Hence, we may commute each limit with each differentiation in (\ref{prenullstate}, \ref{prewardid}) to find that the limit $H(\boldsymbol{\xi};x,0)$ satisfies the $(2N-2)$ null-state PDEs (\ref{nullstate}) and the three conformal Ward identities (\ref{wardid}) in the coordinates of $\boldsymbol{\xi}$.  It is also evident that this limit satisfies a power-law bound of the type (\ref{powerlaw}) in the coordinates of $\boldsymbol{\xi}$.  Thus, it is an element of $\mathcal{S}_{N-1}$.

Finally, we prove (\ref{thesecondlim}) by using (\ref{transform}), which we now write as
\be\label{transformch2}F(\boldsymbol{x})=|\partial f(x_1)|^{(6-\kappa)/2\kappa}|\partial f(x_2)|^{(6-\kappa)/2\kappa}\dotsm |\partial f(x_{2N})|^{(6-\kappa)/2\kappa}\hat{F}(\boldsymbol{x}')\ee
with $\hat{F}$ defined in (\ref{hatF}), $x':=f(x)$, $\boldsymbol{x}':=(x_1',x_2',\ldots,x_{2N}')$, and for our present purposes, with $f$ the M\"{o}bius transformation
\be\label{fmapping}f(x)=\frac{(x_{2N-1}+1-x_{2N-2})(x-x_{2N})}{(x_{2N}-x_{2N-2})(x-x_{2N-1}-1)}.\ee
If $x_{2N-1}+1<x_{2N}$, which we may assume because we are sending $x_{2N}=R\rightarrow\infty$ in (\ref{thesecondlim}), then this transformation (\ref{fmapping}) cyclically permutes the coordinates of $\boldsymbol{x}$ rightward along the real axis so $x_{2N}'=0<x_1'<x_2'<\ldots<x_{2N-2}'=1<x_{2N-1}'$.   (We note that $\boldsymbol{x}'$ is not in $\Omega_0$ because $x_{2N}'<x_1'$.)  From (\ref{transformch2}), we find
\begin{multline}\label{asympsides}(2R)^{6/\kappa-1}F(-R,x_2,x_3,\ldots,x_{2N-1},R)\underset{R\rightarrow\infty}{\sim}|\partial f(x_2)|^{(6-\kappa)/2\kappa}|\partial f(x_3)|^{(6-\kappa)/2\kappa}\dotsm\,\\
\dotsm\,|\partial f(x_{2N-1})|^{(6-\kappa)/2\kappa}(x_1'-x_{2N}')^{6/\kappa-1}\hat{F}(x_1',x_2',\ldots,x_{2N}').\end{multline}
In the primed coordinates, $x_1'\rightarrow x_{2N}'$ as $R\rightarrow\infty$.  Because $\hat{F}$ satisfies the system of PDEs (\ref{nullstate}, \ref{wardid}) and obeys the bound (\ref{powerlaw}) in the primed coordinates, we may invoke the result of the previous paragraph to conclude that
\begin{multline}\label{limrightside}\lim_{R\rightarrow\infty}(2R)^{6/\kappa-1}F(-R,x_2,x_3,\ldots,x_{2N-1},R)=|\partial f_0(x_2)|^{(6-\kappa)/2\kappa}|\partial f_0(x_3)|^{(6-\kappa)/2\kappa}\dotsm\,\\
\dotsm\,|\partial f_0(x_{2N-1})|^{(6-\kappa)/2\kappa}F_0(X_2',X_3',\ldots,X_{2N-1}'),\quad X_j':=\lim_{R\rightarrow\infty}x_j',\end{multline}
for some $F_0\in\mathcal{S}_{N-1}$.  Here, $f_0$ is the limit of (\ref{fmapping}) evaluated at $x_{2N}=R$ as $R\rightarrow\infty$, so $f_0(x_j)=X_j'$ for all $j\neq 1,2N$.   Transformation law (\ref{transformch2}) adapted to $F_0$ becomes the functional equation
\be\label{transformadapt}|\partial f_0(x_2)|^{(6-\kappa)/2\kappa}|\partial f_0(x_3)|^{(6-\kappa)/2\kappa}\dotsm\,|\partial f_0(x_{2N-1})|^{(6-\kappa)/2\kappa}F_0(X_2',X_3',\ldots,X_{2N-1}')=F_0(x_2,x_3,\ldots,x_{2N-1}).\ee
(We have removed the hat that would appear above $F_0$ in (\ref{limrightside}) and on the left side of (\ref{transformadapt}) because $X_2'<X_3'<\ldots<X_{2N-1}'$.)  Combining (\ref{limrightside}) and (\ref{transformadapt}), we find that (\ref{thesecondlim}) is an element of $\mathcal{S}_{N-1}$.
\end{proof}
For notational convenience, we omit explicit reference to the trivial limit $x_i\rightarrow x_{i-1}$ in (\ref{thefirstlim}) from now on.

\section{Construction of the dual space $\mathcal{S}_N^*$}\label{dualspace}

Inspired by the special property conveyed in lemma \ref{SN-1lem}, we construct elements of the dual space $\mathcal{S}_N^*$ as follows.  Starting with any element of $\mathcal{S}_N$, we take the limit (\ref{thefirstlim}) or (\ref{thesecondlim}) to find an element of $\mathcal{S}_{N-1}$.  Then we repeat this process $N-1$ more times until we ultimately arrive with an element of $\mathcal{S}_0=\mathbb{R}$.  This sequence of limits is a linear mapping $\mathscr{L}$ sending $\mathcal{S}_N$ into the real numbers and is thus an element of $\mathcal{S}_N^*$.  (In the sequel \cite{florkleb3}, we show that the set of all such mappings is a basis for $\mathcal{S}_N^*$, as the title of this section suggests.)

We must carefully order the $N$ limits of $\mathscr{L}$ in order for the action of this map on elements of $\mathcal{S}_N$ to be well-defined.  Intuitively, we best understand this restriction by drawing a disk and marking on its boundary $2N$ points $P=\{p_1,p_2,\ldots,p_{2N}\}$ in counterclockwise order and with $p_i$ corresponding to the $i$th coordinate of $\boldsymbol{x}\in\Omega_0$ (figure \ref{CircleCsls}).  Then the first limit of $\mathscr{L}$ may bring together any two points $p_i,p_j\in P$ not separated from each other within the disk boundary by other points of $P$. Next, the second limit of $\mathscr{L}$ may bring together any two points $p_k,p_l\in P\setminus\{p_i,p_j\}$ not separated from each other within the disk boundary by other points of $P\setminus\{p_i,p_j\}$, and so on.  We note that if the limits of $\mathscr{L}$ are ordered this way, then we can join the points of $P$ pairwise with $N$ non-intersecting arcs in the disk, where the endpoints of the $j$th arc are the two points in $P$ brought together by the $j$th limit of $\mathscr{L}$.  We call this diagram an \emph{interior arc connectivity diagram}, and we imagine that the $j$th limit of $\mathscr{L}$ contracts the $j$th arc of this diagram to a point (figure \ref{CircleCsls}).

\begin{figure}[t]
\centering
\includegraphics[scale=0.25]{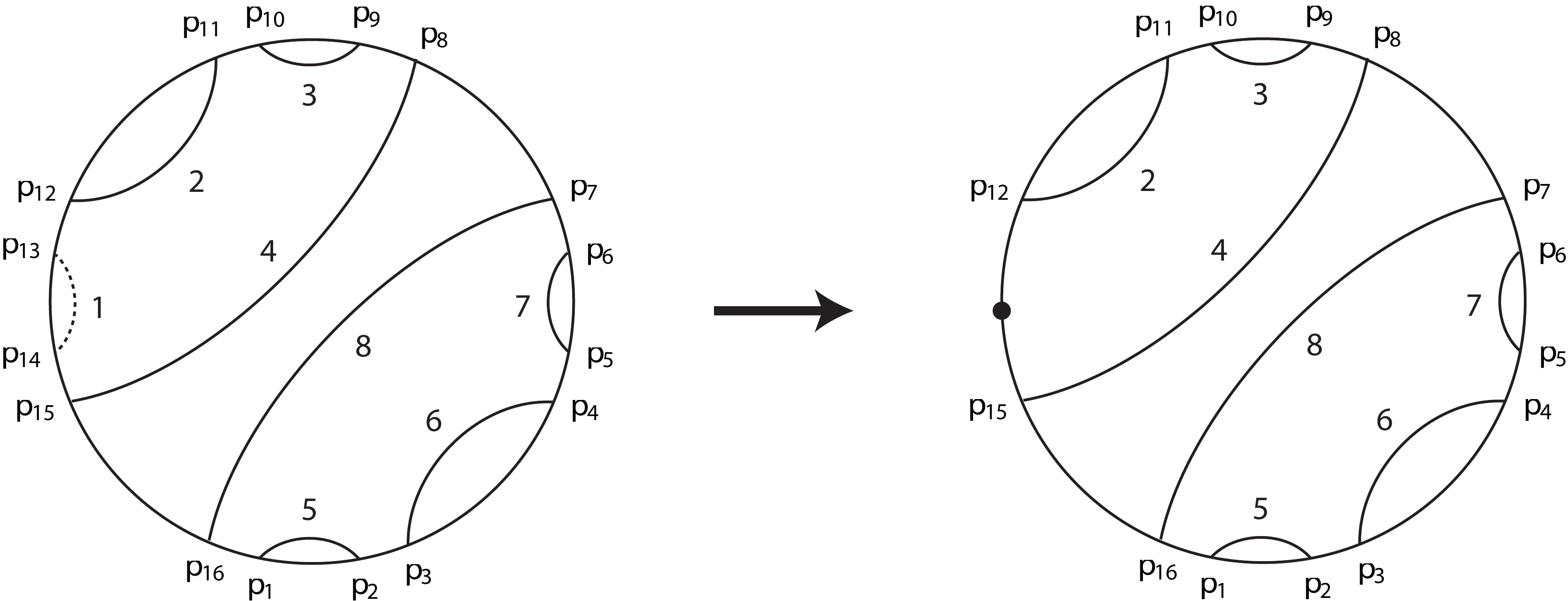}
\caption{An interior arc connectivity diagram for some $\mathscr{L}\in\mathcal{S}_8^*$.  The $j$th limit of $\mathscr{L}$ brings together the endpoints of the $j$th arc.  The figure illustrates the first limit, bringing together the endpoints of the first arc.}
\label{CircleCsls}
\end{figure}

Now two natural questions arise.  First, does an interior arc connectivity diagram always give rise to an element of $\mathcal{S}_N^*$?  And second, if two mappings $\mathscr{L},\mathscr{L}'\in\mathcal{S}_N^*$, constructed according to the previous paragraph, share an interior arc connectivity diagram (thus they are distinguished only by the ordering of their limits), then does $\mathscr{L}'F=\mathscr{L}F$ for all $F\in\mathcal{S}_N$?  In anticipation of an affirmative answer to both questions, we conveniently use the interior arc connectivity diagrams to formally state the restrictions on the ordering of the limits in $\mathscr{L}$.  Because the points brought together by the limits of $\mathscr{L}$ are in the real axis, we consider these diagrams in the upper half-plane first.

\begin{defn} With $x_1<x_2<\ldots<x_{2N}$ the coordinates of a point $\boldsymbol{x}\in\Omega_0$, an \emph{interior arc half-plane diagram on $x_1,$ $x_2,\ldots,x_{2N}$} is a collection of $N$ curves, called \emph{arcs}, in the closure of the upper half-plane such that
\begin{itemize}
\item each arc has its two endpoints among $x_1,$ $x_2,\ldots,x_{2N}$,
\item the two endpoints of each arc are distinct,
\item no point among $x_1,$ $x_2,\ldots,x_{2N}$ is an endpoint of two different arcs,
\item no two different arcs intersect each other, and no arc intersects itself.
\end{itemize}
\end{defn}
\noindent
There are $C_N$ topologically distinct interior arc half-plane diagrams on $x_1,$ $x_2,\ldots,x_{2N}$ \cite{fgg}, where $C_N$ is the $N$th Catalan number (\ref{catalan}).

\begin{defn} After we enumerate all $C_N$ of the interior arc half-plane diagrams on $x_1,$ $x_2,\ldots,x_{2N}$, we define the \emph{$\varsigma$th connectivity} to be the $\varsigma$th of these.  Furthermore, we define the \emph{$\varsigma$th interior arc polygon diagram} to be the image of the $\varsigma$th diagram under a continuous map that sends the upper half-plane (resp.\ real axis, resp.\ coordinates $x_1,$ $x_2,\ldots,x_{2N}$) onto the interior (resp.\ boundary, resp.\ vertices) of a $2N$-sided regular polygon.  We call either diagram the \emph{$\varsigma$th interior arc connectivity diagram}.
\end{defn}

\begin{notation}For any $a<b\in\mathbb{R}$, we write $(b,a)$ for the complement of $[a,b]$ in the one-point compactification of $\mathbb{R}$.\end{notation}

\begin{defn}\label{cslsdefn} With $x_1<x_2<\ldots<x_{2N}$ the coordinates of a point $\boldsymbol{x}\in\Omega_0$, we choose one of the $C_N$ available half-plane arc connectivity diagrams on $x_1,$ $x_2,\ldots,x_{2N}$, enumerate the $N$ arcs of this diagram in some arbitrary way, and let $x_{i_{2j-1}}<x_{i_{2j}}$ be the endpoints of the $j$th arc.  
\begin{enumerate}[I.]
\item We let $\pi$ be a projection (definition \ref{pidef}) that removes all coordinates of $\boldsymbol{x}\in\Omega_0$ that are in $(x_{i_{2j-1}},x_{i_{2j}})$ and some even number of coordinates of $\boldsymbol{x}$ that are in $(x_{i_{2j}},x_{i_{2j-1}})$.  If $\pi$ removes $2M$ coordinates with $M\in\{1,2,\ldots,N\}$, then for $\kappa\in(0,8)$ and $F\in\mathcal{S}_{N-M}$, we define $\bar{\ell}_j:\mathcal{S}_{N-M}\rightarrow\mathcal{S}_{N-M-1}$ by
\be \label{pionF1}(\bar{\ell}_jF\circ\pi_{i_{2j-1},i_{2j}}\circ\pi)(\boldsymbol{x})\,\,\,:=\lim_{x_{i_{2j}}\rightarrow x_{i_{2j-1}}}(x_{i_{2j}}-x_{i_{2j-1}})^{6/\kappa-1}(F\circ\pi)(\boldsymbol{x}).\ee
Lemma \ref{limitlem} guarantees that this limit exists and does not depend on $x_{i_{2j-1}}$, and lemma \ref{SN-1lem} guarantees that this limit is in $\mathcal{S}_{N-M-1}$.
\item Or we let $\pi$ be a projection that removes all coordinates of $\boldsymbol{x}\in\Omega_0$ that are in $(x_{i_{2j}},x_{i_{2j-1}})$ and some even number of coordinates of $\boldsymbol{x}$ that are in $(x_{i_{2j-1}},x_{i_{2j}})$.  If $\pi$ removes $2M$ coordinates with $M\in\{1,2,\ldots,N\}$, then for $\kappa\in(0,8)$ and $F\in\mathcal{S}_{N-M}$, we define $\underline{\ell}_j:\mathcal{S}_{N-M}\rightarrow\mathcal{S}_{N-M-1}$ by
\be\label{pionF2}(\underline{\ell}_jF\circ\pi_{i_{2j-1},i_{2j}}\circ\pi)(\boldsymbol{x})\,\,\,:=\lim_{R\rightarrow\infty}(2R)^{6/\kappa-1}(F\circ\pi)(\boldsymbol{x})\Big|_{(x_{i_{2j-1}},x_{i_{2j}})=(-R,R)}.\ee
Lemma \ref{limitlem} guarantees that this limit exists, and lemma \ref{SN-1lem} guarantees that this limit is in $\mathcal{S}_{N-M-1}$.  
\end{enumerate}
In either case, we say that \emph{$\ell_j$ collapses the interval $(x_{i_{2j-1}},x_{i_{2j}})$}, where $\ell_j$ is a generic label for either $\bar{\ell}_j$ or $\underline{\ell}_j$.  Next, we enumerate the labels $j\in\{1,2,\ldots,N\}$ of the $N$ arcs as $j_1$, $j_2,\ldots,j_N$ and in such a way that the following is true for each $k\in\{1,2,\ldots,N\}$:
\begin{enumerate}
\item\label{thefirstitem}$\displaystyle{\bigcup_{l>k}^N\{x_{i_{2j_l-1}},x_{i_{2j_l}}\}}\subset(x_{i_{2j_k-1}},x_{i_{2j_k}}),$ or
\item\label{theseconditem}$\displaystyle{\bigcup_{l>k}^N\{x_{i_{2j_l-1}},x_{i_{2j_l}}\}}\subset(x_{i_{2j_k}},x_{i_{2j_k-1}})$.
\end{enumerate}
Then for $\kappa\in(0,8)$ and $M\in\{1,2,\ldots, N\}$, we define 
\be\label{Ldefn}\mathscr{L}:\mathcal{S}_N\rightarrow\mathcal{S}_{N-M},\quad\mathscr{L}F:=\ell_{j_M}\ell_{j_{M-1}}\dotsm\ell_{j_2}\ell_{j_1}F,\ee
and we call $\mathscr{L}$ an \emph{allowable sequence of $M$ limits involving the coordinates $x_{i_{2j_1-1}}$, $x_{i_{2j_1}}$, $x_{i_{2j_2-1}}$, $x_{i_{2j_2}},\ldots,x_{i_{2j_{M}-1}}$, $x_{i_{2j_M}}$}.  In this definition (\ref{Ldefn}) for $\mathscr{L}$,
\begin{itemize}
\item if item \ref{thefirstitem} above is true for a particular $k\in\{1,2,\ldots,\min\{M,N-1\}\}$, then $\ell_{j_k}=\underline{\ell}_{j_k}$, 
\item if item \ref{theseconditem} above is true for a particular $k\in\{1,2,\ldots,\min\{M,N-1\}\}$, then $\ell_{j_k}=\bar{\ell}_{j_k}$.
\end{itemize}
If $M=N$, then the aforementioned rule does not specify whether $\ell_{j_N}=\bar{\ell}_{j_N}$ or $\ell_{j_N}=\underline{\ell}_{j_N}$.  But also if $M=N$, then $\ell_{j_N}$ acts on an element of $\mathcal{S}_1$.  According to (\ref{S1}), this element is a multiple of 
\be F(x_{i_{j_{2N-1}}},x_{i_{j_{2N}}})=(x_{i_{j_{2N}}}-\,\,x_{i_{j_{2N-1}}})^{1-6/\kappa},\ee
the image of which under $\bar{\ell}_{j_N}$ equals its image under $\underline{\ell}_{j_N}$.  Thus, we may take either $\ell_{j_N}=\bar{\ell}_{j_N}$ or $\ell_{j_N}=\underline{\ell}_{j_N}$ in $\mathscr{L}$.  Finally, we define the \emph{half-plane} (resp.\ \emph{polygon}) \emph{diagram for $\mathscr{L}$} to be the interior arc half-plane (resp.\ polygon) diagram used to construct $\mathscr{L}$ with the arcs whose endpoints are among 
\be\{x_{i_{2j_{M+1}-1}},x_{i_{2j_{M+1}}},x_{i_{2j_{M+2}-1}},x_{i_{2j_{M+2}}},\ldots,x_{i_{2j_N-1}},x_{i_{2j_N}}\}\ee
deleted if $M<N$.  We call either diagram the \emph{diagram for $\mathscr{L}$}.
\end{defn}
\noindent
The part of definition \ref{cslsdefn} that defines an allowable sequence of $M$ limits anticipates the following lemma.

\begin{figure}[b]
\centering
\includegraphics[scale=0.3]{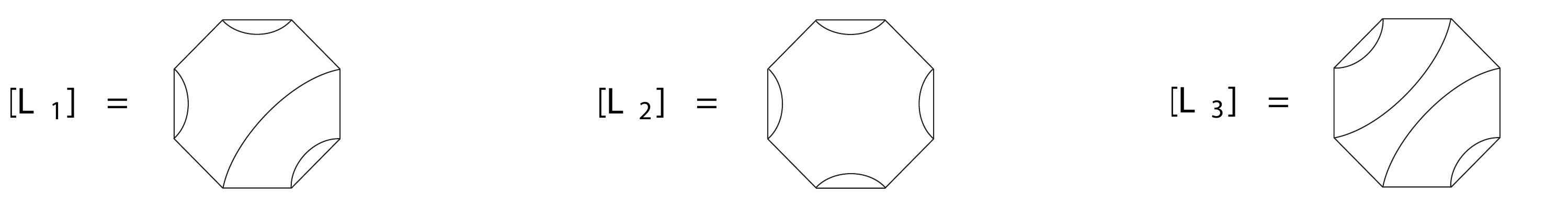}
\caption{Polygon diagrams for three different equivalence classes of allowable sequences of $N=4$ limits.  The other $C_4-3=11$ diagrams are found by rotating one of these three.}
\label{Csls}
\end{figure}

\begin{lem}\label{cslslem} Suppose that $\kappa\in(0,8)$ and $F\in\mathcal{S}_N$, and let $\mathscr{L}=\ell_{j_M}\ell_{j_{M-1}}\dotsm\ell_{j_2}\ell_{j_1}$ be an allowable sequence of $M\leq N$ limits.  Then the limit $\mathscr{L}F$ exists and is in $\mathcal{S}_{N-M}$.  In particular, if $M=N$, then the limit is a real number.
\end{lem}

\begin{proof} Items \ref{thefirstitem} and \ref{theseconditem} in definition \ref{cslsdefn} imply that the two coordinates $x_{i_{2j_1-1}}$ and $x_{i_{2j_1}}$ satisfy either of the following:
\begin{itemize}
\item $x_{i_{2j_1-1}}=x_i$ and $x_{i_{2j_1}}=x_{i+1}$ for some $i\in\{1,2,\ldots,2N-1\}$, and therefore $\ell_{j_1}=\bar{\ell}_{j_1}$, or
\item $x_{i_{2j_1-1}}=x_1$ and $x_{i_{2j_1}}=x_{2N}$, and therefore $\ell_{j_1}=\underline{\ell}_{j_1}$.
\end{itemize}
In either case, lemmas \ref{limitlem} and \ref{SN-1lem} guarantee that the limit $F_0:=\ell_{j_1}F$ exists and is in $\mathcal{S}_{N-1}$.  Furthermore, $\ell_{j_M}\ell_{j_{M-1}}\dotsm\ell_{j_2}$ is, according to definition \ref{cslsdefn}, an allowable sequence of limits involving the coordinates $x_{i_{2j_2-1}}$, $x_{i_{2j_2}}$, $x_{i_{2j_3-1}}$, $x_{i_{2j_3}}, \ldots, x_{i_{2j_M-1}}$, $x_{i_{2j_M}}$.  Therefore, items \ref{thefirstitem} and \ref{theseconditem} in definition \ref{cslsdefn} imply that the two coordinates $x_{i_{2j_2-1}}$ and $x_{i_{2j_2}}$ satisfy either of the following:
\begin{itemize}
\item item \ref{thefirstitem} in definition \ref{cslsdefn}, so $\ell_{j_1}=\bar{\ell}_{j_1}$ and there are no points among $x_{i_{2j_3-1}}$, $x_{i_{2j_3}}$,  $x_{i_{2j_4-1}}$,  $x_{i_{2j_4}} \ldots, x_{i_{2j_M-1}}$, $x_{i_{2j_M}}$ in $(x_{i_{2j_2-1}},x_{i_{2j_2}})$, or
\item item \ref{theseconditem}  in definition \ref{cslsdefn}, so $\ell_{j_1}=\underline{\ell}_{j_1}$ and there are no points among $x_{i_{2j_3-1}}$, $x_{i_{2j_3}}$,  $x_{i_{2j_4-1}}$,  $x_{i_{2j_4}} \ldots, x_{i_{2j_M-1}}$, $x_{i_{2j_M}}$ in $(x_{i_{2j_2}},x_{i_{2j_2-1}}).$
\end{itemize}
In either case, lemma \ref{limitlem} guarantees that the limit $\ell_{j_2}F_0=\ell_{j_2}\ell_{j_1}F$ exists, and lemma \ref{SN-1lem} guarantees that it is in $\mathcal{S}_{N-2}$.  We repeat this reasoning $M-2$ more times to prove the lemma.
\end{proof}

The vector space over the real numbers generated by the allowable sequences of limits acting on $\mathcal{S}_N$ has a natural partition into equivalence classes given by the following definition.  (See figure \ref{Csls} for an example.)
\begin{defn}\label{equivdefn}We say that two allowable sequences of $M\in\{1,2,\ldots,N\}$ limits $\mathscr{L}$ and $\mathscr{L}'$ involving the same coordinates of $\boldsymbol{x}\in\Omega_0$ are \emph{equivalent} if their diagrams are identical.  This defines an equivalence relation on the set of all allowable sequences of $M$ limits, and we represent the equivalence class containing $\mathscr{L}$ by $[\mathscr{L}]$.  If $M=N$, then we enumerate these equivalence classes so the arcs in the diagram for $[\mathscr{L}_\varsigma]$ join in the $\varsigma$th connectivity, and we let $\mathscr{B}^*_N:=\{[\mathscr{L}_1],[\mathscr{L}_2],\ldots,[\mathscr{L}_{C_N}]\}$.
\end{defn}
\noindent
Because there are exactly $C_N$ interior arc connectivity diagrams, with $C_N$ the $N$th Catalan number (\ref{catalan}), it immediately follows that the cardinality of $\mathscr{B}_N^*$ is $C_N$.

\begin{lem}\label{welldef} Suppose that $\kappa\in(0,8)$ and $F\in\mathcal{S}_N$, and let $[\mathscr{L}]$ be an equivalence class of allowable sequences of $M\in\{1,2,\ldots,N\}$ limits.  Then $[\mathscr{L}]F$ is well-defined in the sense that $\mathscr{L}'F=\mathscr{L}''F$ for all $\mathscr{L}',\mathscr{L}''\in[\mathscr{L}]$.
\end{lem}
\begin{proof}If $N=1$ or 2, then we may prove the lemma by working directly with the elements of $\mathcal{S}_N$, all of which we know explicitly (\ref{S1}, \ref{S2}).  Therefore, we assume that $N>2$ throughout this proof.

\begin{figure}[b!]
\centering
\includegraphics[scale=0.27]{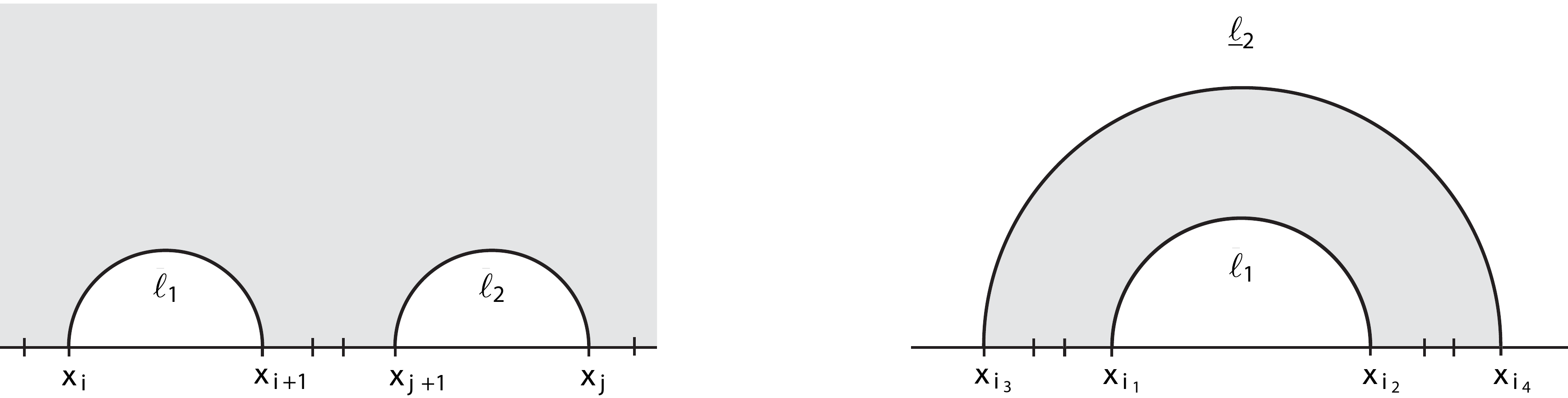}
\caption{The two possible half-plane diagrams for an equivalence class $[\mathscr{L}]$ with $M=2$ limits.  All other coordinates not involved in $[\mathscr{L}]$, indicated with ticks, lie on the part of the real axis touching the gray area.}
\label{NestUnnest}
\end{figure}

The proof is by induction on $M$.  To begin, we suppose that $M=2<N$, so $[\mathscr{L}]$ has at most two elements.  We further assume that $[\mathscr{L}]$ has exactly two elements (or else $[\mathscr{L}]$ would have one element, so there would be nothing to prove) and the arcs in the half-plane diagram for $[\mathscr{L}]$ are un-nested.  We enumerate these arcs so the left arc has $j=1$, the right arc has $j=2$, and $[\mathscr{L}]=\{\bar{\ell}_2\bar{\ell}_1,\bar{\ell}_1\bar{\ell}_2\}$.  Then because the two arcs are un-nested, we have $x_{i_1}<x_{i_2}<x_{i_3}<x_{i_4}$.  Furthermore, because $[\mathscr{L}]$ has exactly two elements and the two arcs in its diagram are un-nested, definition \ref{cslsdefn} implies that the $2N-4$ coordinates of $\boldsymbol{x}\in\Omega_0$ not involved in $[\mathscr{L}]$ must lie outside of $(x_{i_1},x_{i_2})\cup(x_{i_3},x_{i_4})$.  Thus, we have $i_1=i$, $i_2=i+1$, $i_3=j$, and $i_4=j+1$ for some $i,j\in\{1,2,\ldots,2N-1\}$ with $j>i+1$ (figure \ref{NestUnnest}).  

Now, to prove the lemma for the case discussed in the previous paragraph, we must show that $\bar{\ell}_1\bar{\ell}_2F=\bar{\ell}_2\bar{\ell}_1F$, or equivalently
\be\label{precommute}\lim_{x_{i+1}\rightarrow x_i}\lim_{x_{j+1}\rightarrow x_j}(x_{i+1}-x_i)^{6/\kappa-1}(x_{j+1}-x_j)^{6/\kappa-1}F(\boldsymbol{x})\,\,\,=\lim_{x_{j+1}\rightarrow x_j}\lim_{x_{i+1}\rightarrow x_i}(x_{i+1}-x_i)^{6/\kappa-1}(x_{j+1}-x_j)^{6/\kappa-1}F(\boldsymbol{x}).\ee
In the only other possible scenario with $M=2$ and $[\mathscr{L}]$ having exactly two elements, one of the two arcs in the half-plane diagram for $[\mathscr{L}]$ nests the other arc, and the $2N-4$ coordinates of $\boldsymbol{x}\in\Omega_0$ not involved in $[\mathscr{L}]$ must lie inside the outer arc and outside the inner arc (figure \ref{NestUnnest}).  After conformally transforming so neither image of the two arcs nests the other, we use (\ref{precommute}) to prove the equivalent statement $\bar{\ell}_1\underline{\ell}_2F=\underline{\ell}_2\bar{\ell}_1F$, where we have labeled the inner and outer arc as the first and second arc respectively.  Thus, to prove the lemma with $M=2$, it suffices to prove (\ref{precommute}).

In our proof of (\ref{precommute}), we let $x:=x_i$, $y:=x_j,\delta:=x_{i+1}-x_i$, $\epsilon:=x_{j+1}-x_j$, we relabel the other $2N-4$ coordinates of $\{x_k\}_{k\neq i,i+1,j,j+1}$ in increasing order by $\{\xi_1,\xi_2,\ldots,\xi_{2N-4}\}$, and we let $\boldsymbol{\xi}:=(\xi_1,\xi_2,\ldots,\xi_{2N-4})$. (This definition of $\boldsymbol{\xi}$ resembles, but is not the same as, that in the proof of lemmas \ref{boundedlem}--\ref{SN-1lem}.)  We let
\be\label{Idef}I(\boldsymbol{\xi};x,\delta;y,\epsilon)=\epsilon^{6/\kappa-1}\delta^{6/\kappa-1}F(\xi_1,\xi_2,\ldots,\xi_{i-1},x,x+\delta,\xi_i,\ldots,\xi_{j-3},y,y+\epsilon,\xi_{j-2},\ldots,\xi_{2N-4}),\ee
while restricting to $0<\delta,\epsilon<b$, where $b$ is small enough to ensure that $x_{i+1}$ and $x_{j+1}$ are respectively less than $x_{i+2}$ and $x_{j+2}$.  Expressed in terms of these quantities, (\ref{precommute}) becomes
\be\label{commute}\lim_{\delta\downarrow0}\lim_{\epsilon\downarrow0}I(\boldsymbol{\xi};x,\delta;y,\epsilon)=\lim_{\epsilon\downarrow0}\lim_{\delta\downarrow0}I(\boldsymbol{\xi};x,\delta;y,\epsilon),\ee
which we wish to prove.  Now, the integral equation (\ref{tildeH}) expressed in terms of $I$, $\boldsymbol{\xi}$, $x$, $\delta$, $y$, and $\epsilon$ becomes 
\be\label{I}I(\boldsymbol{\xi};x,\delta;y,\epsilon)= I(\boldsymbol{\xi};x,b;y,\epsilon)-\frac{\kappa}{4}J(\delta,b)\partial_\delta I(\boldsymbol{\xi};x,b;y,\epsilon)+\sideset{}{_\delta^b}\int J(\delta,\eta)\mathcal{N}[I](\boldsymbol{\xi};x,\eta;y,\epsilon)\,{\rm d}\eta,\ee
for all $0<\delta<b$, where $J$ is the Green function (\ref{Jgreenfunc}) and $\mathcal{N}$ is the differential operator
\begin{multline}\label{N}\mathcal{N}[I](\boldsymbol{\xi};x,\eta;y,\epsilon):=\Bigg[\frac{\partial_x}{\eta}-\sum_k\left(\frac{\partial_k}{\xi_k-x-\eta}-\frac{(6-\kappa)/2\kappa}{(\xi_k-x-\eta)^2}\right)-\frac{\partial_y}{y-x-\eta}+\frac{(6-\kappa)/2\kappa}{(y-x-\eta)^2}\\
+\frac{\epsilon\partial_{\epsilon}}{(y-x-\eta)(y+\epsilon-x-\eta)}+\frac{(6-\kappa)/2\kappa}{(y+\epsilon-x-\eta)^2}+\frac{1-6/\kappa}{(y-x-\eta)(y+\epsilon-x-\eta)}\Bigg]I(\boldsymbol{\xi};x,\eta;y,\epsilon).\end{multline}
With $x$, $y$, and the coordinates of $\boldsymbol{\xi}$ fixed to distinct values, we prove (\ref{commute}) by showing that $I(\boldsymbol{\xi};x,\delta;y,\epsilon)$ approaches its limit $I(\boldsymbol{\xi};x,0;y,\epsilon)$ as $\delta\downarrow0$ (guaranteed to exist by lemma \ref{boundedlem}) uniformly over $0<\epsilon<b$.  Now, we find
\be\label{Isup}\sup_{0<\epsilon< b}|I(\boldsymbol{\xi};x,\delta;y,\epsilon)-I(\boldsymbol{\xi};x,0;y,\epsilon)|\leq\frac{\kappa}{8-\kappa}\sup_{0<\epsilon< b}|\delta\partial_\delta I(\boldsymbol{\xi};x,\delta;y,\epsilon)|+\frac{4}{8-\kappa}\sideset{}{_0^\delta}\int\sup_{0<\epsilon< b}|\eta\mathcal{N}[I](\boldsymbol{\xi};x,\eta;y,\epsilon)|\,{\rm d}\eta\ee
from (\ref{I}) after replacing $\delta$ with zero, then replacing $b$ with $\delta$, and then estimating both sides of the equation.  Hence, to prove uniformness, it suffices to show that the right side of (\ref{Isup}) vanishes as $\delta\downarrow0$.  

First, we prove that the integral in (\ref{Isup}) vanishes as $\delta\downarrow0$ by showing that its integrand is bounded over $0<\eta<b$.  The proof of this statement resembles the proof of lemma \ref{boundedlem}, but it has a key difference.  In the proof of lemma \ref{boundedlem}, constants arising from the Schauder interior estimate (\ref{Schauder}) grow without bound as $\epsilon\downarrow0$ because the coefficients of the strictly elliptic PDE (\ref{prePDE}) (recast in terms of the variables used in this proof) grow without bound in this limit.  To avoid this issue, we construct a new strictly elliptic PDE whose coefficients are bounded both as $\delta\downarrow0$ (as they are for (\ref{prePDE}) in the proof of lemma \ref{boundedlem}) \emph{and} as $\epsilon\downarrow0$ (as they are not for (\ref{prePDE}) in the proof of lemma \ref{boundedlem}).

For $k\not\in\{i,i+1,j,j+1\}$, the null-state PDE centered on $x_k$ becomes (now with $k\in\{1,2,\ldots,2N-4\}$)
\begin{multline}\label{Inotii+1}\Bigg[\frac{\kappa}{4}\partial_k^2+\sum_{l\neq k}\left(\frac{\partial_l}{\xi_l-\xi_k}-\frac{(6-\kappa)/2\kappa}{(\xi_l-\xi_k)^2}\right)\\
\begin{aligned}&+\frac{\partial_x}{x-\xi_k}-\frac{\delta\partial_\delta}{(x-\xi_k)(x+\delta-\xi_k)}+\frac{6/\kappa-1}{(x-\xi_k)(x+\delta-\xi_k)}-\frac{(6-\kappa)/2\kappa}{(x-\xi_k)^2}-\frac{(6-\kappa)/2\kappa}{(x+\delta-\xi_k)^2}\\
&+\frac{\partial_y}{y-\xi_k}-\frac{\epsilon\partial_\epsilon}{(y-\xi_k)(y+\epsilon-\xi_k)}+\frac{6/\kappa-1}{(y-\xi_k)(y+\epsilon-\xi_k)}-\frac{(6-\kappa)/2\kappa}{(y-\xi_k)^2}-\frac{(6-\kappa)/2\kappa}{(y+\epsilon-\xi_k)^2}\Bigg] I(\boldsymbol{\xi};x,\delta;y,\epsilon)=0,\end{aligned}\end{multline}
the null-state PDE centered on $x_i$ becomes
\begin{multline}\label{Iith}\Bigg[\frac{\kappa}{4}(\partial_x-\partial_\delta)^2+\frac{\partial_\delta}{\delta}+\frac{(6-\kappa)(\partial_x-\partial_\delta)}{2\delta}+\sum_l\left(\frac{\partial_l}{\xi_l-x}-\frac{(6-\kappa)/2\kappa}{(\xi_l-x)^2}\right)\\
+\frac{\partial_y}{y-x}-\frac{\epsilon\partial_\epsilon}{(y-x)(y+\epsilon-x)}-\frac{(6-\kappa)/2\kappa}{(y-x)^2}-\frac{(6-\kappa)/2\kappa}{(y+\epsilon-x)^2}+\frac{6/\kappa-1}{(y-x)(y+\epsilon-x)}\Bigg]I(\boldsymbol{\xi};x,\delta;y,\epsilon)=0,\end{multline}
the null-state PDE centered on $x_{i+1}$ becomes
\begin{multline}\label{Ii+1th}\Bigg[\frac{\kappa}{4}\partial_\delta^2-\frac{(\partial_x-\partial_\delta)}{\delta}-\frac{(6-\kappa)\partial_\delta}{2\delta}+\sum_k\left(\frac{\partial_l}{\xi_l-x-\delta}-\frac{(6-\kappa)/2\kappa}{(\xi_l-x-\delta)^2}\right)\\
+\frac{\partial_y}{y-x}-\frac{\epsilon\partial_\epsilon}{(y-x)(y+\epsilon-x)}-\frac{(6-\kappa)/2\kappa}{(y-x)^2}-\frac{(6-\kappa)/2\kappa}{(y+\epsilon-x)^2}+\frac{6/\kappa-1}{(y-x)(y+\epsilon-x)}\Bigg]I(\boldsymbol{\xi};x,\delta;y,\epsilon)=0,\end{multline}
and the null-state PDEs centered on $x_j$ and $x_{j+1}$ are found by replacing $(x,\delta;y,\epsilon)\mapsto(y,\epsilon;x,\delta)$ in (\ref{Iith}) and (\ref{Ii+1th}) respectively.
Also, the three conformal Ward identities (\ref{wardid}) become
\begin{align}\label{Iw1}&\bigg[\sideset{}{_k}\sum\partial_k+\partial_x+\partial_y\bigg] I(\boldsymbol{\xi};x,\delta;y,\epsilon)=0,\\
\label{Iw2}&\bigg[\sideset{}{_k}\sum(\xi_k\partial_k+(6-\kappa)/2\kappa)+x\partial_x+\delta\partial_\delta+y\partial_y+\epsilon\partial_\epsilon\bigg]I(\boldsymbol{\xi};x,\delta;y,\epsilon)=0,\\
\label{Iw3}&\bigg[\sideset{}{_k}\sum(\xi_k^2\partial_k+(6-\kappa)\xi_k/\kappa)+x^2\partial_x+(2x+\delta)\delta\partial_\delta+y^2\partial_y+(2y+\epsilon)\epsilon\partial_\epsilon\bigg]I(\boldsymbol{\xi};x,\delta;y,\epsilon)=0.\end{align}
The last two identities are most useful if we isolate $\delta\partial_\delta I$ and $\epsilon\partial_\epsilon I$ in terms of $I$ and its derivatives with respect to $x$, $y$, and the coordinates of $\boldsymbol{\xi}$.  We find
\bea\nonumber\delta\partial_\delta I(\boldsymbol{\xi};x,\delta;y,\epsilon)&=&\frac{1}{2(x-y)+\delta-\epsilon}\sum_k\bigg[(2y+\epsilon-\xi_k)\xi_k\partial_k+(6-\kappa)(2y+\epsilon-2\xi_k)/2\kappa\bigg]I(\boldsymbol{\xi};x,\delta;y,\epsilon)\\
\label{dpd}&+&\frac{1}{2(x-y)+\delta-\epsilon}\bigg[(2y+\epsilon-x)x\partial_x+(y+\epsilon)y\partial_y\bigg]I(\boldsymbol{\xi};x,\delta;y,\epsilon),\\
 \nonumber\epsilon\partial_\epsilon I(\boldsymbol{\xi};x,\delta;y,\epsilon)&=&\frac{1}{2(y-x)+\epsilon-\delta}\sum_k\bigg[(2x+\delta-\xi_k)\xi_k\partial_k+(6-\kappa)(2x+\delta-2\xi_k)/2\kappa\bigg]I(\boldsymbol{\xi};x,\delta;y,\epsilon)\\
\label{epe}&+&\frac{1}{2(y-x)+\epsilon-\delta}\bigg[(x+\delta)x\partial_x+(2x+\delta-y)y\partial_y\bigg]I(\boldsymbol{\xi};x,\delta;y,\epsilon).
\eea

We use PDEs (\ref{Iith}--\ref{Iw3}) to construct a new PDE that has $x,$ $y,$ and the coordinates of $\boldsymbol{\xi}$ as independent variables, that has $\delta$ and $\epsilon$ as parameters, and that is strictly elliptic in an arbitrarily chosen compact subset of $\pi_{i+1,j+1}(\Omega_0)$.  To begin, we subtract (\ref{Ii+1th}) from (\ref{Iith}) and multiply the result by $\delta$ to find
\begin{multline}\label{Idiffpde}\Bigg[\frac{\kappa}{4}\delta\partial_x^2-\frac{\kappa}{2}\partial_x\delta\partial_\delta+\frac{8-\kappa}{2}\partial_x-\sum_k\left(\frac{\delta^2\partial_k}{(\xi_k-x)(\xi_k-x-\delta)}+\frac{[\delta+2(x-\xi_k)]\delta^2(6-\kappa)/2\kappa}{(\xi_k-x)^2(\xi_k-x-\delta)^2}\right)\\
\begin{aligned}
&-\frac{\delta^2\partial_y}{(y-x)(y-x-\delta)}-\frac{[\delta+2(x-y)]\delta^2(6-\kappa)/2\kappa}{(y-x)^2(y-x-\delta)^2}-\frac{\delta^2(2x-2y+\delta-\epsilon)\epsilon\partial_\epsilon}{(y-x)(y-x-\delta)(y+\epsilon-x)(y+\epsilon-x-\delta)}\\
&-\frac{\delta^2(2x-2y+\delta-\epsilon)(1-6/\kappa)}{(y-x)(y-x-\delta)(y+\epsilon-x)(y+\epsilon-x-\delta)}-\frac{[\delta+2(x-y-\epsilon)]\delta^2(6-\kappa)/2\kappa}{(y+\epsilon-x)^2(y+\epsilon-x-\delta)^2}\Bigg]I(\boldsymbol{\xi};x,\delta;y,\epsilon)=0.\end{aligned}\end{multline}
Next, we use (\ref{Iw1}) to eliminate $\partial_yI$ from (\ref{dpd}), and we insert the result into (\ref{Idiffpde}) to generate a PDE whose principal part only contains $\partial_x^2 I$ and the mixed partial derivatives $\partial_x\partial_k I$ for $k\in\{1,2,\ldots,2N-4\}$.  The coefficient of the former term in the principal part is
\be\label{B}a(x,\delta;y,\epsilon)=\frac{\kappa}{4}\delta+\frac{\kappa}{2}\left(\frac{(x-y)^2-\epsilon(x-y)}{2(x-y)+\delta-\epsilon}\right),\ee
and we note that $a$ restricted to a compact subset of $\pi_{i+1,j+1}(\Omega_0)$ (so $x-y$ is bounded away from zero) does not vanish or grow without bound as $\epsilon\downarrow0$ or $\delta\downarrow0$.  The other coefficients in the principle part of this PDE exhibit this property too, and none of the other coefficients in this PDE grow without bound as $\epsilon\downarrow0$ or $\delta\downarrow0$ either.  Next, by replacing $(x,\delta;y,\epsilon)\mapsto(y,\epsilon;x,\delta)$, we generate another PDE whose principal part only contains $\partial_y^2 I$ (with coefficient $a(y,\epsilon;x,\delta)$) and the mixed partial derivatives $\partial_y\partial_k I$ for $k\in\{1,2,\ldots,2N-4\}$, and with the mentioned features of the companion PDE that generated it.

Next, we form a linear combination of the two PDEs that we constructed in the previous paragraph, with respective nonzero coefficients $a(x,\delta;y,\epsilon)^{-1}$ and $a(y,\epsilon;x,\delta)^{-1}$, and the $2N-4$ null-state PDEs in (\ref{Inotii+1}), each with the same coefficient $4c/\kappa$ for some $c>0$, to find a new PDE whose principal part only contains $\partial_x^2I,$ $\partial_y^2I,$ $\partial_k^2I$, and the mixed partial derivatives $\partial_x\partial_kI$ and $\partial_y\partial_kI$ with $k\in\{1,2,\ldots,2N-4\}$.  Furthermore, we may use (\ref{dpd}, \ref{epe}) again to replace the first derivatives $\delta\partial_\delta I$ and $\epsilon\partial_\epsilon I$ in this current PDE with linear combinations of first derivatives of $I$ in $x$, $y$, and the coordinates of $\boldsymbol{\xi}$.  This produces a final PDE (which is very complicated, so we do not display it here) for which $x,$ $y,$ and the coordinates of $\boldsymbol{\xi}$ are independent variables while $\delta$ and $\epsilon$ are simply parameters.  Moreover, the coefficients in the principle part of this final PDE do not vanish or grow without bound as $\epsilon\downarrow0$ or $\delta\downarrow0$, and none of the coefficients of the other terms grow without bound as $\epsilon\downarrow0$ or $\delta\downarrow0$.

Finally, we show that for any open set $U_0\subset\subset\pi_{i+1,j+1}(\Omega_0)$, there exists a choice for $c$ such that the PDE constructed in the previous paragraph is strictly elliptic in that open set for all $0<\delta,\epsilon<b$.  The coefficient matrix for the PDE's principal part is
\renewcommand{\kbldelim}{(}
\renewcommand{\kbrdelim}{)}
\be\label{coefmatrix}\kbordermatrix{
    & \xi_1 & \xi_2 & \hdots & \xi_{2N-5} & \xi_{2N-4} & x & y\\
    \xi_1 & c & 0 & \hdots & 0 & 0 & a_{x,1} & a_{y,1}\\
    \xi_2 & 0 & c & &  & 0 & a_{x,2} & a_{y,2} \\
    \vdots & \vdots &  & \ddots &  & \vdots & \vdots & \vdots \\
    \xi_{2N-5} & 0 &  &  & c & 0 & a_{x,2N-5} & a_{y,2N-5} \\
    \xi_{2N-4} & 0 & 0 & \hdots & 0 & c & a_{x,2N-4} & a_{y,2N-4} \\
   x & a_{x,1} & a_{x,2} & \hdots & a_{x,2N-5} & a_{x,2N-4} & 1 & 0\\
   y& a_{y,1} & a_{y,2} & \hdots & a_{y,2N-5} & a_{y,2N-4} & 0 & 1
  },
\ee
where $a_{x,k}$ (resp.\ $a_{y,k}$) is half of the coefficient of $\partial_x\partial_kI$ (resp.\ $\partial_y\partial_kI$).  According to Sylvester's criterion (Thm.\ 7.5.2 of \cite{horn}), this matrix is positive definite if all of its leading principal minors are positive.  Because $c>0$, the determinant of the first $2N-4$ leading principal minors are evidently positive.  We find the remaining two principal minors by using the determinant formula
\be\label{detformula}M=\left(\begin{matrix} A & B \\ C & D\end{matrix}\right)\quad\Longrightarrow\quad\det M=\det A\det(D-CA^{-1}B),\ee
where $A$ and $D$ are square blocks of the square matrix $M$, and where $B$ and $C$ are blocks that fill the part of $M$ above $D$ and beneath $A$ respectively.  Using this formula, we find that the $(2N-3)$th leading principal minor of the coefficient matrix (\ref{coefmatrix}) is 
\be\label{firstdet}c^{2N-4}\left(1-\frac{|a_x|^2}{c}\right),\ee
where $a_x$ is the vector in $\mathbb{R}^{2N-4}$ whose $k$th entry is $a_{x,k}$, and the $(2N-2)$th leading principal minor (this is simply the determinant of (\ref{coefmatrix})) is
\be\label{seconddet}c^{2N-4}\det\left[\left(\begin{matrix}1&0\\0&1\end{matrix}\right)-\frac{1}{c}\left(\begin{matrix}a_x\cdot a_x&a_x\cdot a_y\\ a_y\cdot a_x&a_y\cdot a_y\end{matrix}\right)\right],\ee
where $a_y$ is the vector in $\mathbb{R}^{2N-4}$ whose $k$th entry is $a_{y,k}$.  By choosing $c$ sufficiently large, we ensure that determinants (\ref{firstdet}) and (\ref{seconddet}) are greater than, say, one at all points in an open set $U_0\subset\subset\pi_{i+1,j+1}(\Omega_0)$ and for all $0<\delta,\epsilon<b$.  Thus, the matrix (\ref{coefmatrix}) is positive definite.  Furthermore, because all of the components of $a_x$ and $a_y$ are bounded on $U_0$, the eigenvalues of the matrix (\ref{coefmatrix}) (which is Hermitian and therefore diagonalizable) are also bounded on $U_0$.  This fact together with the fact that the product of these eigenvalues (equaling (\ref{seconddet})) is greater than one on $U_0$ imply that all of these eigenvalues are bounded away from zero over this set for all $0<\delta,\epsilon<b$.  Thus, the constructed PDE is strictly elliptic in $U_0$ for all $0<\delta,\epsilon<b$.

The existence of this PDE implies Schauder estimates.  We choose open sets $ U_0,$ $ U_1,\ldots, U_{2m+1}$, where $m:=\lceil q\rceil$ and $q:=p+1-6/\kappa$ (we recall that $p$ is given by (\ref{powerlaw}), and we choose $p$ big enough so $m>1$), such that 
\be U_{2m+1}\subset\subset U_{2m}\subset\subset\ldots\subset\subset U_0\subset\subset\pi_{i+1,j+1}(\Omega_0),\ee
and we let $d_n:=\text{dist}(\partial U_{n+1},\partial U_n)$.  With these choices, the Schauder interior estimate gives (Cor.\ 6.3 of \cite{giltru})
\be\label{ISchauder}d_{n+1}\sup_{ U_{n+1}}|\partial^\varpi I(\boldsymbol{\xi};x,\delta;y,\epsilon)|\leq C_n\sup_{ U_n}|I(\boldsymbol{\xi};x,\delta;y,\epsilon)|\ee
for all $0<\delta,\epsilon<b$ and $n\in\{0,1,\ldots,2m\}$, where $C_n$ is some constant and $\varpi$ is a multi-index for the coordinates of $\boldsymbol{\xi}$ and $x$ and $y$.  Furthermore, (\ref{dpd}, \ref{epe}) imply that $\varpi$ may include the derivatives $\delta\partial_\delta$ and $\epsilon\partial_\epsilon$ too.  Overall, we have
\be\label{multiindex2}\partial^\varpi\in\left\{\begin{array}{c}\partial_j,\quad \partial_x,\quad \partial_y,\quad \delta\partial_\delta,\quad\epsilon\partial_\epsilon, \quad\partial_j^2,\quad\partial_x^2,\quad\partial_y^2,\quad(\delta\partial_\delta)^2,\quad(\epsilon\partial_\epsilon)^2,\quad \partial_j\partial_k, \\ \partial_j\partial_x,\quad\partial_j\partial_y,\quad\partial_j\delta\partial_\delta,\quad\partial_j\epsilon\partial_\epsilon,\quad\partial_x\partial_y,\quad\partial_x\delta\partial_\delta,\quad\partial_x\epsilon\partial_\epsilon\quad\partial_y\delta\partial_\delta,\quad \partial_y\epsilon\partial_\epsilon,\quad\delta\partial_\delta\epsilon\partial_\epsilon \end{array}\right\}.\ee

Now we use (\ref{estI2}) to show that $I(\boldsymbol{\xi};x,\delta;y,\epsilon)=O(1)$ as $\delta,\epsilon\downarrow0$.  After taking the supremum of (\ref{I}) over $U_n$, we find that for all $0<\delta<b$,
\be\label{estI2}\sup_{U_n}|I(\boldsymbol{\xi};x,\delta;y,\epsilon)|\leq\sup_{U_n}|I(\boldsymbol{\xi};x,b;y,\epsilon)|+\frac{\kappa}{8-\kappa}\sup_{U_n}|\delta\partial_\delta I(\boldsymbol{\xi};x,b;y,\epsilon)|+\frac{4}{8-\kappa}\sideset{}{_\delta^b}\int \sup_{U_n}|\eta\mathcal{N}[I](\boldsymbol{\xi};x,\eta;y,\epsilon)|\,{\rm d}\eta.\ee
We use (\ref{ISchauder}) with $n=0$ to estimate the integrand of (\ref{estI2}) with $n=1$, finding that it is $O(\eta^{-q}\epsilon^{-q})$ as $\eta\downarrow0$ or $\epsilon\downarrow0$.  After inserting this estimate into (\ref{estI2}) and integrating, we find that (because $m>1$)
\be\sup_{ U_1}|I(\boldsymbol{\xi};x,\delta;y,\epsilon)|=O(\delta^{-q+1}\epsilon^{-q}).\ee
Just as we did in the proof of lemma \ref{boundedlem}, we repeat this process $m-1$ more times, using the subsets $U_m\subset\subset U_{m-1}\subset\subset\ldots\subset\subset U_0$ until we finally arrive with
\be\sup_{ U_m}|I(\boldsymbol{\xi};x,\delta;y,\epsilon)|=O(\epsilon^{-q})\quad\text{as $\delta,\epsilon\downarrow0$.}\ee
Now to decrease the power on $\epsilon$, we switch $(x,\delta;y,\epsilon)$ to $(y,\epsilon;x,\delta)$ in (\ref{estI2}), and we continue the previous steps for the $\epsilon$ variable, using the subsets $ U_{2m}\subset\subset U_{2m-1}\subset\subset\ldots\subset\subset U_m$.  We ultimately find
\be\label{almostlastest1}\sup_{ U_{2m}}|I(\boldsymbol{\xi};x,\delta;y,\epsilon)|=O(1)\quad\text{as $\delta,\epsilon\downarrow0$.}\ee
This fact followed by one last application of the Schauder estimate (\ref{ISchauder}) with $n=2m$ implies that the integrand of (\ref{Isup}) is a bounded function of $\eta$.  Hence, the definite integral of (\ref{Isup}) vanishes as $\delta\downarrow0$.

Now we argue that the first term on the right side of (\ref{Isup}) vanishes as $\delta\downarrow0$ too.  After re-expressing (\ref{partialyH}) in terms of the quantities used in this proof, we take the supremum of both sides over $0<\epsilon<b$, finding
\be\label{supdpd}\sup_{0<\epsilon<b}|\delta\partial_\delta I(\boldsymbol{\xi};x,\delta;y,\epsilon)|\leq\left(\frac{\delta}{b}\right)^{8/\kappa-1}\sup_{0<\epsilon<b}|b\hspace{.03cm}\partial_\delta I(\boldsymbol{\xi};x,b;y,\epsilon)|+\frac{4}{\kappa}\sideset{}{_\delta^b}\int\left(\frac{\delta}{\eta}\right)^{8/\kappa-1}\sup_{0<\epsilon<b}|\eta\mathcal{N}[I](\boldsymbol{\xi};x,\eta;y,\epsilon)|\,{\rm d}\eta\ee
for all $0<\delta<b$.  We just argued that the supremum in the integrand of (\ref{supdpd}) is bounded over $0<\eta<b$.  In light of this fact, we repeat the analysis of (\ref{gotozero}, \ref{limepsilon}) to show that the right side of (\ref{supdpd}), and thus the first term on the right side of (\ref{Isup}), vanishes as $\delta\downarrow0$.  From this fact and the finding of the previous paragraph, we conclude that the right side, and therefore the left side, of (\ref{Isup}) vanishes as $\delta\downarrow0$.  Thus,$I(\boldsymbol{\xi};x,\delta;y,\epsilon)$ converges to its limit as $\delta\downarrow0$ uniformly over $0<\epsilon<b$.  This fact proves the equality supposed in (\ref{precommute}, \ref{commute}).  Consequently, if each element of $[\mathscr{L}]$ has two limits, then $[\mathscr{L}]F$ is well-defined for all $F\in\mathcal{S}_N$.  This proves the lemma in the case $M=2$.

To prove the lemma for $M\in\{3,4,\ldots,N\}$, we use induction.  After selecting an $M<N$ from this set and an arbitrary $F\in\mathcal{S}_N$, we suppose that $[\mathscr{L}']F$ is well-defined for all equivalence classes $[\mathscr{L}']$ whose elements have fewer than $M$ limits, and we choose a different equivalence class $[\mathscr{L}]$ whose elements have exactly $M$ limits $\ell_1,$ $\ell_2,\ldots,\ell_M$.  Each element of $[\mathscr{L}]$ equals $\ell_m\mathscr{L}_{\overline{m}}$ for some $m\in\{1,2,\ldots,M\}$ and some allowable sequence $\mathscr{L}_{\overline{m}}$ of the $M-1$ limits of $\{\ell_j\}_{j\neq m}$.  We let 
\be\label{Adef}\mathcal{A}:=\{m\in\mathbb{Z}^+\,|\, \text{there is an element of $[\mathscr{L}]$ that executes the particular limit $\ell_m$ last}\}.\ee
For fixed $m\in \mathcal{A}$, all elements of $[\mathscr{L}]$ of the form $\ell_m\mathscr{L}_{\overline{m}}$ are equivalent by definition \ref{equivdefn}, and we denote their equivalence sub-class by $\ell_m[\mathscr{L}_{\overline{m}}]$, where  $[\mathscr{L}_{\overline{m}}]$ is the equivalence class for $\mathscr{L}_{\overline{m}}$.   By the induction hypothesis, $[\mathscr{L}_{\overline{m}}]F$ is well-defined for all $m\in \mathcal{A}$, so to finish the proof, we show that $\ell_m[\mathscr{L}_{\overline{m}}]F=\ell_n[\mathscr{L}_{\overline{n}}]F$ for each pair $(m,n)\in \mathcal{A}\times \mathcal{A}$.

Because $M<N$, the condition $m,n\in \mathcal{A}$ (i.e., we may collapse either the $m$th arc or the $n$th arc in the half-plane diagram for $\mathscr{L}$ last) and $m\neq n$ constrains the possible connectivities of the arcs in the half-plane diagram for $[\mathscr{L}]$.  (We recall from definition \ref{cslsdefn} that we have enumerated the $M$ arcs in this diagram, the limit $\ell_j$ corresponds to the $j$th arc, and the coordinates $x_{i_{2j-1}}<x_{i_{2j}}$ of $\boldsymbol{x}\in\Omega_0$ are the endpoints of the $j$th arc.)  Indeed, if the $m$th arc nests the $n$th arc in this diagram, then the following are true.
\begin{enumerate}
\item\label{firstit}All coordinates of $\boldsymbol{x}\in\Omega_0$ that are not involved in $[\mathscr{L}]$ reside in $(x_{i_{2m-1}},x_{i_{2n-1}})\cup(x_{i_{2n}},x_{i_{2m}})$.  Indeed, if one of these coordinates resides in $(x_{i_{2n-1}},x_{i_{2n}})$ instead, then $\ell_n=\underline{\ell}_n$ and $\ell_m=\underline{\ell}_m$ in all elements of $[\mathscr{L}]$.  Hence, the limit $\underline{\ell}_n$ necessarily follows $\underline{\ell}_m$ in all elements of $[\mathscr{L}]$, contradicting the supposition that $m\in \mathcal{A}$.  A similar argument shows that none of these coordinates reside in $(x_{i_{2m}},x_{i_{2m-1}})$ either, thus proving the claim.  
\item\label{secondit} We have $\ell_m=\underline{\ell}_m$, and $\ell_n=\bar{\ell}_n$ in all elements of $[\mathscr{L}]$.  Indeed, this follows immediately from condition \ref{firstit} above.
\item\label{thirdit} No arc in the half-plane diagram for $[\mathscr{L}]$ simultaneously has one endpoint in $(x_{i_{2m-1}},x_{i_{2n-1}})$ and its other endpoint in $(x_{i_{2n}},x_{i_{2m}})$.  Indeed, suppose that the contrary is true for the $j$th arc with $j\not\in\{m,n\}$.  Then according to definition \ref{cslsdefn}, either all or none of the coordinates of $\boldsymbol{x}$ that are not involved in $[\mathscr{L}]$ reside between the endpoints of the $j$th arc.  In the former case, $\ell_j=\underline{\ell}_j$, and $\underline{\ell}_j$ necessarily follows $\underline{\ell}_m$ in each element of $[\mathscr{L}]$, contradicting the supposition that $m\in \mathcal{A}$.  In the latter case, $\ell_j=\bar{\ell}_j$, and $\bar{\ell}_j$ necessarily follows $\bar{\ell}_n$ in each element of $[\mathscr{L}]$, contradicting the supposition that $n\in \mathcal{A}$.  
\item\label{fourthit} Both endpoints of the $j$th arc in the half-plane diagram for $[\mathscr{L}]$ with $j\not\in\{m,n\}$ reside in only one of the following four intervals: $(x_{i_{2m}},x_{i_{2m-1}})$, $(x_{i_{2m-1}},x_{i_{2n-1}})$, $(x_{i_{2n-1}},x_{i_{2n}})$, or $(x_{i_{2n}},x_{i_{2m}})$.  Indeed, were the endpoints to reside in different intervals, then because the $j$th arc cannot cross the $m$th arc or the $n$th arc, one endpoint must be in $(x_{i_{2m-1}},x_{i_{2n-1}})$ while the other must be in $(x_{i_{2n}},x_{i_{2m}})$, contradicting condition \ref{thirdit} above.
\item\label{fifthit} None of the coordinates of $\boldsymbol{x}\in\Omega_0$ that are not involved in $[\mathscr{L}]$ reside between the endpoints of the $j$th arc with $j\neq m$.  For if one such coordinate did reside there, then $\ell_j=\underline{\ell}_j$, and $\underline{\ell}_j$ necessarily follows $\underline{\ell}_m$ in all elements of $[\mathscr{L}]$, contradicting the supposition that $m\in \mathcal{A}$.
\end{enumerate}
(Figure \ref{ArcPossibilities} shows a half-plane diagram of an $[\mathscr{L}]$ with $M<N$ limits that satisfies conditions \ref{firstit}--\ref{fifthit} above.)

\begin{figure}[t]
\centering
\includegraphics[scale=0.27]{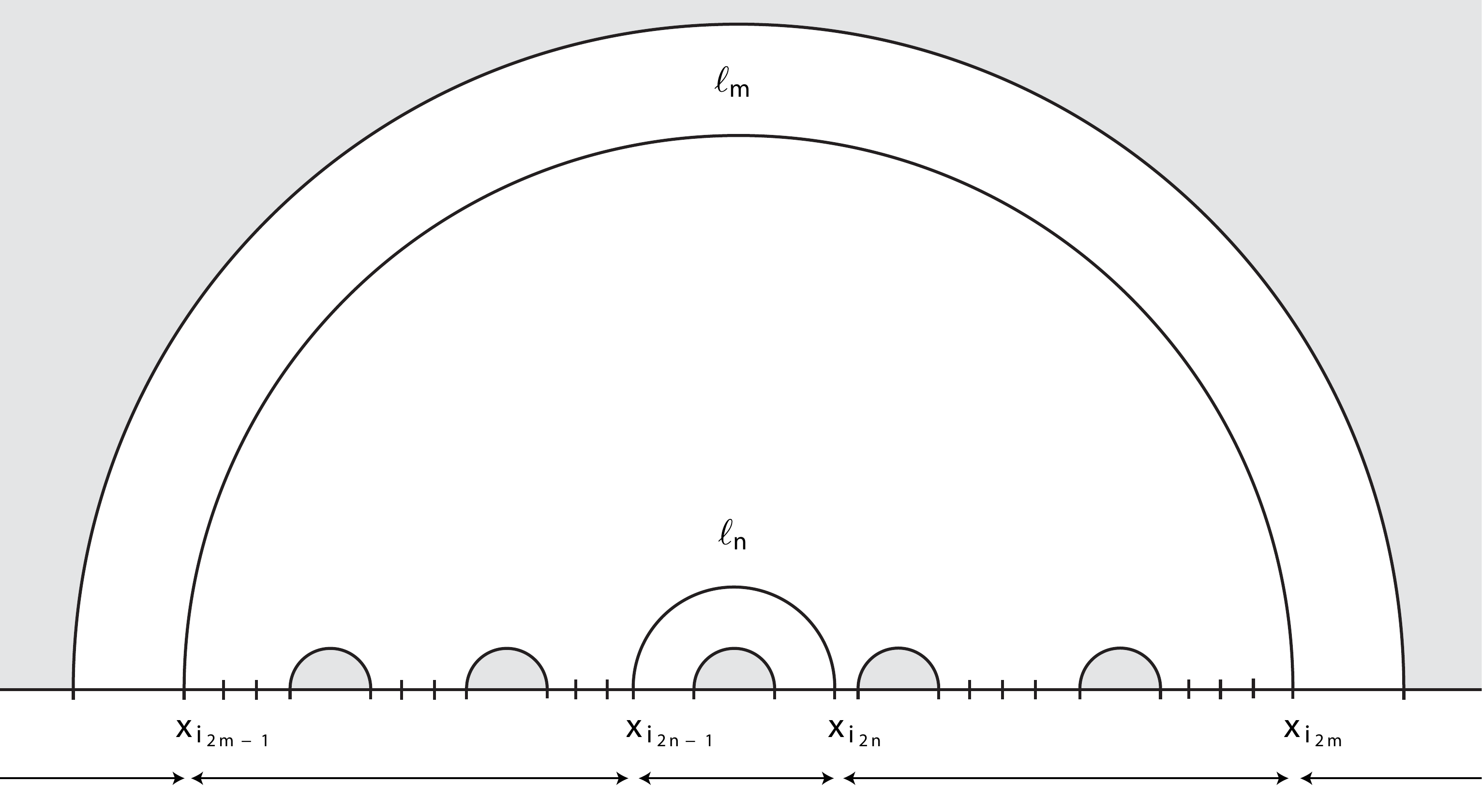}
\caption{The half-plane diagram for an $[\mathscr{L}]$ with $m,n\in\mathcal{A}$, showing the $m$th and $n$th arcs.  The ticks crossing the $x$-axis mark the coordinates not involved in $[\mathscr{L}]$.  (See item \ref{fifthit}.)  Thus, all of the other arcs (not all shown) must reside in the gray regions.}
\label{ArcPossibilities}
\end{figure}

Now, conditions \ref{fourthit} and \ref{fifthit} imply the existence of an allowable sequence $\mathscr{L}_{\overline{m},\overline{n}}$ of the $M-2$ limits of $\{\ell_j\}_{j\neq m,n}$ that contract away all but the $m$th and $n$th arcs in the half-plane diagram for $\mathscr{L}$.  Indeed, because all of these former arcs reside in the gray regions of figure \ref{ArcPossibilities}, we may separately contract away all of these arcs without colliding against the $m$th or $n$th arc or against points not involved in $\mathscr{L}$.  Furthermore, lemma \ref{cslslem} implies that $\mathscr{L}_{\overline{m},\overline{n}}F\in\mathcal{S}_{N-M+2}$.  Now, we proved the present lemma \ref{welldef} in the case $M=2$ earlier, and this result implies that if $\underline{\ell}_m\bar{\ell}_n$ and $\bar{\ell}_n\underline{\ell}_m$ are allowable sequences of limits acting on $G\in\mathcal{S}_{N'}$ with $N'>1$, then $\underline{\ell}_m\bar{\ell}_nG=\bar{\ell}_n\underline{\ell}_mG$.  With $N'=N-M+2$, we thus have
\be\label{equality}\mathscr{L}_{\overline{m},\overline{n}}F\in\mathcal{S}_{N-M+2}\quad\Longrightarrow\quad\underline{\ell}_m\bar{\ell}_n\mathscr{L}_{\overline{m},\overline{n}}F=\bar{\ell}_n\underline{\ell}_m\mathscr{L}_{\overline{m},\overline{n}}F.\ee
But also, $\bar{\ell}_n\mathscr{L}_{\overline{m},\overline{n}}\in[\mathscr{L}_{\overline{m}}]$ and $\underline{\ell}_m\mathscr{L}_{\overline{m},\overline{n}}\in[\mathscr{L}_{\overline{n}}]$.  The induction hypothesis implies that $[\mathscr{L}_{\overline{m}}]F$ and $[\mathscr{L}_{\overline{n}}]F$ are well-defined because $[\mathscr{L}_{\overline{m}}]$ and $[\mathscr{L}_{\overline{n}}]$ are allowable sequences of fewer than $M$ limits.  Thus, we may insert $\bar{\ell}_n\mathscr{L}_{\overline{m},\overline{n}}F=[\mathscr{L}_{\overline{m}}]F$ and $\underline{\ell}_m\mathscr{L}_{\overline{m},\overline{n}}F=[\mathscr{L}_{\overline{n}}]F$ in (\ref{equality}) to ultimately find that 
\be\label{finalcsls}\underline{\ell}_m[\mathscr{L}_{\overline{m}}]F=\bar{\ell}_n[\mathscr{L}_{\overline{n}}]F\quad\text{for $m,n\in \mathcal{A}$.}\ee
The proof of (\ref{finalcsls}) in the case where the arcs for $\ell_m$ and $\ell_n$ are un-nested proceeds similarly.  Thus, $[\mathscr{L}]F$ is well-defined if $[\mathscr{L}]$ involves $M<N$ limits.

If $N=M$ and the $m$th arc nests the $n$th arc, then the half-plane diagram for $[\mathscr{L}]$ still must satisfy  conditions \ref{firstit} and \ref{secondit} above, but it may now violate conditions \ref{thirdit} and \ref{fourthit}.  (Condition \ref{fifthit} no longer applies.)  If $[\mathscr{L}]$ satisfies all of conditions \ref{firstit}--\ref{fourthit}, then the argument of the previous paragraph proves that $[\mathscr{L}]F$ is well-defined if $M=N$.  If $[\mathscr{L}]$ does not satisfy conditions \ref{thirdit} and \ref{fourthit}, then we let the indices $j_1=m,$ $j_2$, $j_3,\ldots,j_{k-1}$, $j_k=n$ label from outermost to innermost the $k$ arcs with one endpoint in $[x_{i_{2m-1}},x_{i_{2n-1}}]$ and the other in $[x_{i_{2n}},x_{i_{2m}}]$. We note two facts.
\begin{enumerate}[I.]
\item\label{thefirst} If $M=N$, then it is easy to see that $\mathcal{A}=\{1,2,\ldots,N\}$ (\ref{Adef}).  Thus, $[\mathscr{L}_{\overline{j}}]$ is defined for each $j\in\{1,2,\ldots,N\}$.
\item\label{thesecond} If $M=N$, then $\bar{\ell}_j[\mathscr{L}_{\overline{j}}]F=\underline{\ell}_j[\mathscr{L}_{\overline{j}}]F$ for all $j\in\{1,2,\ldots,N\}$.  This follows immediately from the fact that $[\mathscr{L}_{\overline{j}}]F\in\mathcal{S}_1$ and (\ref{S1}).
\end{enumerate}
Now, because no arc separates the $j_{l-1}$th arc from the $j_l$th arc in the upper half-plane, the argument of the previous paragraph and item \ref{thefirst} says that $\underline{\ell}_{j_{l-1}}[\mathscr{L}_{\overline{j_{l-1}}}]F=\bar{\ell}_{j_l}[\mathscr{L}_{\overline{j_l}}]F$ for each $l\in\{2,3,\ldots,k\}$.  Following with item \ref{thesecond}, we find
\be \underline{\ell}_m[\mathscr{L}_{\overline{m}}]F=\underline{\ell}_{j_1}[\mathscr{L}_{\overline{j_1}}]F=\underline{\ell}_{j_2}[\mathscr{L}_{\overline{j_2}}]F=\quad\ldots\quad=\underline{\ell}_{j_k}[\mathscr{L}_{\overline{j_k}}]F=\bar{\ell}_n[\mathscr{L}_{\overline{n}}]F\ee
and thus conclude that $[\mathscr{L}]F$ is well-defined if $M=N$.  The proof of the case with the $m$th and $n$th arcs un-nested proceeds similarly.
\end{proof}

\section{An upper bound for $\dim\mathcal{S}_N$}\label{upperbound}

In this section we find an upper-bound for the dimension of $\mathcal{S}_N$, namely that $\dim\mathcal{S}_N\leq C_N$.  But first, we introduce some terminology motivated by CFT.  In CFT, if each endpoint of an interval $(x_i,x_{i+1})$ hosts a one-leg boundary operator, then upon sending $x_{i+1}\rightarrow x_i$, these operators fuse into a combination of an identity operator and a two-leg operator.  If only the identity (resp.\ two-leg) channel appears in the OPE, then we call $(x_i,x_{i+1})$ an \emph{identity interval} (resp.\ a \emph{two-leg interval}), and the correlation function with these one-leg boundary operators exhibits this OPE by admitting a Frobenius series expansion in powers of $x_{i+1}-x_i$ and with indicial power $p_1=-2\theta_1+\theta_0=1-6/\kappa$ (resp.\ $p_2=-2\theta_1-\theta_2=2/\kappa$) (\ref{p}).  (The powers of these two expansions differ by an integer only when $\kappa=8/r$ for some $r\in\mathbb{Z}^+$.  This invites the consideration of logarithmic CFT \cite{gur,matrid}.  We explore this situation more carefully in \cite{florkleb4}.)  Lemma \ref{limitlem} extends these notions to all elements of $\mathcal{S}_N$ (should any not be identified with correlation functions of one-leg boundary operators), motivating the following definition.

\begin{defn}\label{intervaldefn} With $\kappa\in(0,8)$, we choose $F\in\mathcal{S}_N$ and $i\in\{1,2,\ldots,2N-1\}$.  Interpreting $\pi_{i+1}(\Omega_0)$ as the part of the boundary of $\Omega_0$ whose points have only their $i$th and $(i+1)$th coordinates equal, we let
\be\label{secondH}H:\Omega_0\cup\pi_{i+1}(\Omega_0)\rightarrow\mathbb{R},\quad H(\boldsymbol{x}):=(x_{i+1}-x_i)^{6/\kappa-1}F(\boldsymbol{x})\quad\text{for $\boldsymbol{x}\in\Omega_0$},\ee
and continuously extend $H$ to $\pi_{i+1}(\Omega_0)$.  ($H$ as defined in (\ref{secondH}) is virtually identical to the $H$ defined in (\ref{FtoH}).)  We also let
\be\label{secondF0} F_0:\pi_{i+1}(\Omega_0)\rightarrow\mathbb{R},\quad (F_0\circ\pi_{i+1})(\boldsymbol{x}):=\lim_{x_{i+1}\rightarrow x_i}(x_{i+1}-x_i)^{6/\kappa-1}F(\boldsymbol{x}).\ee
($F_0$ as defined in (\ref{secondF0}) is virtually identical to the $F_0$ defined in (\ref{thefirstlim}).)
\begin{enumerate}
\item\label{cftinterval1} We define $(x_i,x_{i+1})$ to be a \emph{two-leg interval of $F$} if the limit $F_0$ (\ref{secondF0}) vanishes. 
\item\label{cftinterval2} For $8/\kappa\not\in\mathbb{Z}^+$, we define $(x_i,x_{i+1})$ to be an \emph{identity interval of $F$} if the limit $F_0$ (\ref{secondF0}) does not vanish and $H$ (\ref{secondH}) is analytic at every point in $\pi_{i+1}(\Omega_0)$.
\end{enumerate}
Letting $x_i'=f(x_i)$, where $f$ is the M\"{o}bius transformation defined in (\ref{fmapping}), we define the interval $(x_{2N},x_1)$ of $F$ to be the same type as its image interval $(x_{2N}',x_1')$ with respect to the function (\ref{hatF}) with the following restricted domain:
\be\hat{F}:f(\Omega_0)\rightarrow\mathbb{R},\quad f(\Omega_0):=\{\boldsymbol{x}'=(x_1',x_2',\ldots,x_{2N}')\in\Omega\,|\,\boldsymbol{x}\in\Omega_0\}.\ee
(This restricted function is an element of $\mathcal{S}_N$ after we rearrange the indices of the coordinates of $\boldsymbol{x}'$ so they increase as we go from the smallest coordinate of $\boldsymbol{x}'$ to its biggest.)
\end{defn}
\noindent
See appendix \ref{preliminaries} and sections \red{II} and \red{IV} of the sequel \cite{florkleb4} for a CFT interpretation of these definitions.  (If $8/\kappa\in\mathbb{Z}^+$, then we do not define the term ``identity interval" for now because of complications that arise if the two indicial powers (\ref{p}) differ by an integer.  See table \red{I} of \cite{florkleb3} and sections \red{II} and \red{IV} of \cite{florkleb4} for further details.)

The Green function (\ref{greenfunc}) used in the proof of lemma \ref{boundedlem} gives the power law for $F$ if $(x_i,x_{i+1})$ is a two-leg interval of it.  This is simply the two-leg power $p_2=-2\theta_1+\theta_2=2/\kappa$ of (\ref{p}).  We prove this claim in lemma \red{5} of \cite{florkleb2}.

\begin{lem}\label{alltwoleglem}  Suppose that $\kappa\in(0,8)$ and $F\in\mathcal{S}_N$ with $N>1$.  If all of $(x_2,x_3)$, $(x_3,x_4),\ldots,(x_{2N-2},x_{2N-1})$, and $(x_{2N-1},x_{2N})$  are two-leg intervals of $F$, then $F=0$.
\end{lem}

The physical motivation for this lemma is as follows.  The discussion in section \ref{survey} implies that (at least in the case of percolation where $\kappa=6$) the evolving curves anchored to the endpoints of a two-leg interval are conditioned to not join together and form one boundary arc in the long-time limit of multiple SLE$_\kappa$.  If all of the intervals among $(x_2,x_3)$, $(x_3,x_4),\ldots,(x_{2N-2},x_{2N-1})$ and $(x_{2N-1},x_{2N})$ are two-leg intervals, then there cannot be an interval among these whose endpoints are joined by a single boundary arc.  But, topological considerations show that no such boundary arc connectivity exists, thus implying lemma \ref{alltwoleglem}.  Lemma \ref{alltwoleglem} is therefore both simple and natural.  However, a simple proof does not seem possible, and in fact the entirety of \cite{florkleb2} is devoted to a proof of it.  In the next three paragraphs, we explore  this circumstance.

Lemma \ref{alltwoleglem} is not necessarily true if $F$ does not satisfy the conformal Ward identities (\ref{wardid}).  Indeed (if we can adapt the definition of a ``two-leg interval" to a solution of this smaller system (\ref{nullstate})), we have the following counterexample:
\be\label{counterexample}F(\boldsymbol{x})=\prod_{i<j}^{2N}(x_j-x_i)^{2/\kappa}.\ee
This function (\ref{counterexample}) satisfies the null-state PDEs (\ref{nullstate}) and only the first conformal Ward identity (\ref{wardid}), and each interval $(x_i,x_{i+1})$ with $i<2N$ is a two-leg interval of it.  Thus, a proof of lemma \ref{alltwoleglem} must use the other identities of (\ref{wardid}).

The weak maximum principle almost justifies lemma \ref{alltwoleglem} if $\kappa\in(0,6]$ and $(x_1,x_2)$ is also a two-leg interval of $F$, but the argument encounters technical difficulties.  Such a proof might proceed as follows.  With $x_1=a$, $x_{2N}=b$, and
\be\begin{gathered}\label{FR} \Omega_{a,b}:=\{(x_2,x_3,\ldots,x_{2N-1})\in\mathbb{R}^{2N-2}\,|\,a<x_2<x_3<\ldots<x_{2N-1}<b\},\\
F_R:\Omega_{a,b}\rightarrow\mathbb{R},\quad F_R(x_2,x_3,\ldots,x_{2N-1})=F(a,x_2,x_3,\ldots,x_{2N-1},b),\end{gathered}\ee
we use the first two Ward identities of (\ref{wardid}) to eliminate all derivatives with respect to $x_1$ and $x_{2N}$ in the null-state PDE (\ref{nullstate}) centered on $x_j$ with $j\in\{2,3,\ldots,2N-1\}$.  This PDE becomes
\begin{multline}\label{subelliptic}\Bigg[\frac{\kappa}{4}\partial_j^2+\sum_{k\neq j,1,2N}\left(\frac{\partial_k}{x_k-x_j}-\frac{(6-\kappa)/2\kappa}{(x_k-x_j)^2}\right)-\sum_{k\neq1,2N}\frac{(x_k-a)\partial_k}{(b-x_j)(x_j-a)}\\
+\sum_{k\neq 1,2N}\frac{\partial_k}{x_j-a}-\frac{N(6/\kappa-1)}{(b-x_j)(x_j-a)}-\frac{(6-\kappa)/2\kappa}{(x_j-a)^2}-\frac{(6-\kappa)/2\kappa}{(b-x_j)^2}\Bigg]F_R(x_2,x_3,\ldots,x_{2N-1})=0.\end{multline}
After summing (\ref{subelliptic}) over $j\in\{2,3,\ldots,2N-1\}$, we find a strictly elliptic PDE with a nonpositive constant term.  By hypothesis, all intervals of $F_R$ are two-leg intervals, and because of this, one may show (see lemma \red{5} of \cite{florkleb2}) that $F_R$ continuously extends to and equals zero on $\partial\Omega_{a,b}\setminus\pi_{1,2N}(E)$, where $E$ is the set of points $(a,x_2,x_3,\ldots,x_{2N-1},b)\in\partial\Omega_0$ with three or more coordinates equal.  If $F_R$ continuously extends to and equals zero on all of $\partial\Omega_{a,b}$, then the weak maximum principle \cite{giltru} implies that $F$ is zero.  However, we have not found a way to derive this vanishing extension, nor have we found bounds on the growth of $F_R$ near $\pi_{1,2N}(E)$ for which we could use a Phragm\'{e}n-Lindel\"{o}f maximum principle \cite{protter} to skirt this issue.  Furthermore, this principle would be difficult to apply because the coefficients of (\ref{subelliptic}) are not bounded as the points in $\pi_{1,2N}(E)$ are approached.

These arguments suggest that we must determine the behavior of $F\in\mathcal{S}_N$ as more than two coordinates simultaneously come together in order to prove lemma \ref{alltwoleglem}.  Moreover, these arguments do not limit the number of simultaneously converging coordinates that we must consider.  Explicitly determining the behavior of $F_R$ (\ref{FR}) as, for example, all coordinates of a point in $\Omega_{a,b}$ simultaneously approach $a$ is likely impossible to do with our present results.  Therefore, we instead prove this lemma in \cite{florkleb2} with a method based on a multiple-SLE$_\kappa$ argument.  This different method has a major advantage.  It only requires us to determine the behavior of $F(\boldsymbol{x})$ as no more than three coordinates of $\boldsymbol{x}$ simultaneously approach each other, an attainable goal.

If we assume lemma \ref{alltwoleglem}, then we immediately obtain the upper bound $\dim\mathcal{S}_N\leq C_N$.
\begin{lem}\label{dimupperboundlem} Suppose that $\kappa\in(0,8)$ and $F\in\mathcal{S}_N$, and let $v:\mathcal{S}_N\rightarrow\mathbb{R}^{C_N}$ be the map with the $\varsigma$th coordinate of $v(F)$ equaling $[\mathscr{L}_\varsigma]F$ for $[\mathscr{L}_\varsigma]\in\mathscr{B}_N^*$.  Then $v$ is a linear injection, and $\dim\mathcal{S}_N\leq C_N$.
\end{lem}
\begin{proof}The map $v$ is clearly linear.  To show that it is injective, we argue that its kernel is trivial.  Supposing that $F$ is not zero, we construct an allowable sequence of limits $\mathscr{L}$ such that $\mathscr{L}F$ is not zero as follows.  According to lemma \ref{alltwoleglem}, there is some index $i_1=i\in\{2,3,\ldots,2N\}$ such that the interval $(x_{i_1},x_{i_2})=(x_i,x_{i+1})$ is not a two-leg interval of $F$.  Invoking definition \ref{cslsdefn} and lemma \ref{SN-1lem}, we let $\bar{\ell}_1$ collapse this interval, obtaining $\bar{\ell}_1F\in\mathcal{S}_{N-1}\setminus\{0\}$.  Now because, $\bar{\ell}_1F$ is not zero, there is some other index $i_3\in\{2,3,\ldots,2N\}\setminus\{i_1,i_2\}$ such that the interval $(x_{i_3},x_{i_4})$ is not a two-leg interval of $\bar{\ell}_1F$.  We let $\bar{\ell}_2$ collapse this interval next, obtaining $\bar{\ell}_2\bar{\ell}_1F\in\mathcal{S}_{N-2}\setminus\{0\}$.  Repeating this process $N-3$ more times leaves us with $\bar{\ell}_{N-1}\dotsm\bar{\ell}_2\bar{\ell}_1F\in\mathcal{S}_1\setminus\{0\}$ with the form (\ref{S1}), and a final limit $\bar{\ell}_N$ sends it to a nonzero real number.  Now because $\mathscr{L}F:=\bar{\ell}_{N}\dotsm\bar{\ell}_2\bar{\ell}_1F\neq0$, it follows that $v(F)\neq0$, and $\ker v=\{0\}$.   Finally, the dimension theorem of linear algebra then implies that $\dim\mathcal{S}_N\leq C_N$.
\end{proof}

To finish, we state a corollary that justifies our exclusive consideration of the system of PDEs (\ref{nullstate}, \ref{wardid}) with an even number of independent variables.  The proofs of lemmas \ref{boundedlem}--\ref{SN-1lem}, \ref{cslslem}, \ref{welldef}, and \ref{alltwoleglem} do not required this number to be even, so they remain true if it is odd.  We therefore define $\mathcal{S}_{N+1/2}$ exactly as we define $\mathcal{S}_N$ for the system of PDEs (\ref{nullstate}, \ref{wardid}) but with $2N+1$ independent variables.  

\begin{cor}\label{trivialcor}Suppose that $\kappa\in(0,8)$ and $F\in\mathcal{S}_{N+1/2}$.  If $\kappa\neq6$, then $\mathcal{S}_{N+1/2}=\{0\}$, and if $\kappa=6$, then $\mathcal{S}_{N+1/2}=\mathbb{R}$.
\end{cor}

\begin{proof} We prove the lemma by induction.  First, if $\kappa\neq6$, then apparently only the trivial solution satisfies the system (\ref{nullstate}, \ref{wardid}) with one independent variable $x_1$, so $\mathcal{S}_{1/2}=\{0\}$.  Now, we suppose that $\mathcal{S}_{N-1/2}=\{0\}$ for some $N\geq1$.  Lemma \ref{SN-1lem} shows that collapsing any interval $(x_i,x_{i+1})$ with $i\in\{2,3,\ldots,2N\}$ sends $F\in\mathcal{S}_{N+1/2}$ to an element of $\mathcal{S}_{N-1/2}$, and by the induction hypothesis, this element is zero.  Thus, each of these intervals is a two-leg interval of $F$ according to definition \ref{intervaldefn}, so lemma \ref{alltwoleglem} implies that $F$ is zero.  We conclude that $\mathcal{S}_{N+1/2}=\{0\}$.

Next, we consider the case $\kappa=6$.  It is obvious from the PDEs (\ref{nullstate}, \ref{wardid}) with $\kappa=6$ that $\mathbb{R}$ is a subspace of $\mathcal{S}_{N+1/2}$ for all $N\geq0$.  Furthermore, only the real numbers satisfy the system (\ref{nullstate}, \ref{wardid}) with one independent variable $x_1$ or with three independent variables $x_1$, $x_2,$ and $x_3$, so $\mathcal{S}_{1/2}=\mathcal{S}_{3/2}=\mathbb{R}$.  Indeed, for the latter case, all solutions of the conformal Ward identities (\ref{wardid}) have the structure of a three-point CFT correlation function
\be\label{3pt}F(x_1,x_2,x_3)\propto(x_2-x_1)^{-h_1-h_2+h_3}(x_3-x_1)^{-h_3-h_1+h_2}(x_3-x_2)^{-h_2-h_3+h_1}\ee
with $h_j$ the conformal weight of $x_j$.  Because these conformal weights are $\theta_1=(6-\kappa)/2\kappa=0$, (\ref{3pt}) is a real number.  Now, we suppose that $\mathcal{S}_{N-1/2}=\mathbb{R}$ for some $N>1$.  Lemma \ref{SN-1lem} shows that collapsing any interval $(x_i,x_{i+1})$ with $i\in\{2,3,\ldots,2N\}$ sends $F\in\mathcal{S}_{N+1/2}$ to an element of $\mathcal{S}_{N-1/2}$, and by the induction hypothesis, this element is a real number.  Furthermore, for any $i,j\in\{2,3,\ldots,2N\}$ with $j>i+1$, we consider the two limits (definition \ref{cslsdefn} with $\kappa=6$)
\be\label{barij}(\bar{\ell}_1F\circ\pi_{i,i+1})(\boldsymbol{x})\,\,\,:=\lim_{x_{i+1}\rightarrow x_i}F(\boldsymbol{x})\in\mathcal{S}_{N-1/2}=\mathbb{R},\qquad(\bar{\ell}_2F\circ\pi_{j,j+1})(\boldsymbol{x})\,\,\,:=\lim_{x_{j+1}\rightarrow x_j}F(\boldsymbol{x})\in\mathcal{S}_{N-1/2}=\mathbb{R}.\ee
According to lemma \ref{welldef}, $\bar{\ell}_1\bar{\ell}_2F=\bar{\ell}_2\bar{\ell}_1F$ (\ref{precommute}).  But because $\bar{\ell}_1F,\bar{\ell}_2F\in\mathbb{R}$, we have $\bar{\ell}_1\bar{\ell}_2F=\bar{\ell}_2F$ and $\bar{\ell}_2\bar{\ell}_1F=\bar{\ell}_1F$ too.  It immediately follows that $\bar{\ell}_2F=\bar{\ell}_1F$ for all $i,j\in\{2,3,\ldots,2N\}$ with $j>i+1$.  With $2N>3$, this last equation implies that for all $i\in\{2,3,\ldots,2N\}$, $\bar{\ell}_1F$ is the same real number $F_0$, and because $\mathbb{R}\subset\mathcal{S}_{N+1/2}$, we have $F-F_0\in\mathcal{S}_{N+1/2}$.  Now for all $i\in\{2,3,\ldots,2N\}$, $\bar{\ell}_1(F-F_0)=0$, so $(x_i,x_{i+1})$ is a two-leg interval of $F-F_0$ by definition \ref{intervaldefn}.  Lemma \ref{alltwoleglem} then implies that $F-F_0=0$, so $F\in\mathbb{R}$ and $\mathcal{S}_{N+1/2}=\mathbb{R}$.
\end{proof}

We reconsider the induction argument in the proof of corollary \ref{trivialcor} with $\kappa=6$ and an even number of independent variables.  Because $\mathcal{S}_1=\mathbb{R}$ (\ref{S1}), this argument might seem to imply that $\mathcal{S}_N=\mathbb{R}$ for all $N>1$.  But as $\dim\mathcal{S}_2=2$, a fact found by explicitly solving the system (\ref{nullstate}, \ref{wardid}) with $N=2$ in section \ref{objorg}, the induction step that would take us from $\mathcal{S}_1$ to $\mathcal{S}_2$ evidently must break down.  Out of interest, we investigate how this happens.  We let $\bar{\ell}_1$, $\bar{\ell}_3$, and $\bar{\ell}_2$ collapse the intervals $(x_1,x_2)$, $(x_2,x_3)$, and $(x_3,x_4)$ respectively.  Thus, $\bar{\ell}_1$ (resp.\ $\bar{\ell}_2$) acts on $F\in\mathcal{S}_2$ as shown in (\ref{barij}) with $i=1$ (resp.\ $j=3$).  Now, lemma \ref{welldef} implies that $\bar{\ell}_1\bar{\ell}_2F=\bar{\ell}_2\bar{\ell}_1F$, so $\bar{\ell}_2F=\bar{\ell}_1F\in\mathcal{S}_1=\mathbb{R}$ are the same real number.  However, lemma \ref{welldef} does not give a similar equation relating $\bar{\ell}_3F$ to the other two limits because $(x_2,x_3)$ is adjacent to their respective intervals.  So while all three of these limits are real numbers, only $\bar{\ell}_2F$ and $\bar{\ell}_1F$ are necessarily equal while $\bar{\ell}_3F\in\mathcal{S}_1=\mathbb{R}$ may not equal either of them.  For example, if $F=1\in\mathcal{S}_2$, then $\bar{\ell}_2F=\bar{\ell}_1F=\bar{\ell}_3F$.  However, if $F$ is given by (\ref{4ptansatz}) with $G=G_1$ (\ref{Pi1}), then $\bar{\ell}_2F=\bar{\ell}_1F=0$ while $\bar{\ell}_3F\neq0$.  Thus, this induction step fails at $N=2$, and the solution space $\mathcal{S}_N$ with $N>2$ and $\kappa=6$ contains $\mathbb{R}$ as a proper subspace.

\section{Summary}

In this article, we study a solution space for the system of PDEs (\ref{nullstate}, \ref{wardid}) with $\kappa\in(0,8)$ the Schramm-L\"owner evolution (SLE$_\kappa$) parameter.  These PDEs govern multiple-SLE$_{\kappa}$ partition functions and CFT $2N$-point correlation functions of one-leg boundary operators (\ref{onelegcorr}).  In fact, systems (\ref{nullstate}) and (\ref{wardid}) are respectively the conformal field theory (CFT) null-state conditions and conformal Ward identities for these correlation functions.  Roughly speaking, such correlation functions are partition functions for continuum limits of statistical cluster or loop models such as percolation, or more generally the Potts models and O$(n)$ models, in polygons and at the statistical mechanical critical point.  These partition functions exclusively sum over side-alternating free/fixed boundary condition events.
 
Focusing on the space $\mathcal{S}_N$ of solutions for this system that grow no faster than a power law (\ref{powerlaw}), we use techniques of analysis and linear algebra to rigorously establish various facts about solutions in this space.  Although this system (\ref{nullstate}, \ref{wardid}) arises in CFT in a way that is typically non-rigorous, our treatment of this system here and in \cite{florkleb2,florkleb3,florkleb4} is completely rigorous.  In section \ref{boundary behavior}, we prove three lemmas showing that the behavior of any such solution $F(\boldsymbol{x})$ as $x_{i+1}\rightarrow x_i$ is consistent with the formalism of the CFT operator product expansion (OPE), which is generally assumed to be valid in the physics literature.  In lemma \ref{boundedlem}, we state that any solution $F(\boldsymbol{x})$ is $O((x_{i+1}-x_i)^{1-6/\kappa})$ as we bring together (only) two adjacent coordinates $x_i$ and $x_{i+1}$ of the point $\boldsymbol{x}$.   Lemmas \ref{limitlem} and \ref{SN-1lem} strengthen this result by proving, respectively, that the limit of $(x_{i+1}-x_i)^{6/\kappa-1}F(\boldsymbol{x})$ as  $x_{i+1} \to x_i$ exists  and is an element of $\mathcal{S}_{N-1}$. Next, in sections \ref{dualspace} and \ref{upperbound}, we prove that $\dim\mathcal{S}_N\leq C_N$, where $C_N$ is the $N$th Catalan number (\ref{catalan}), by explicitly constructing $C_N$ elements of the dual space $\mathcal{S}_N^*$ (definition \ref{cslsdefn}) and using them as components of a linear injective map $v$ that embeds $\mathcal{S}_N$ into $\mathbb{R}^{C_N}$ (lemma \ref{dimupperboundlem}).  We stress that our results are completely rigorous, in spite of our occasional asides to matters in CFT that relate to critical lattice models.

In part two \cite{florkleb2} of this series of articles, we prove lemma \ref{alltwoleglem}, and in part three \cite{florkleb3}, we  prove that if $\kappa\in(0,8)$, then $\dim\mathcal{S}_N=C_N$ and we may construct all elements of $\mathcal{S}_N$ using the CFT Coulomb gas (contour integral) formalism.  In part four \cite{florkleb4}, we derive Frobenius series expansions of solutions in $\mathcal{S}_N$, study special solutions called \emph{connectivity weights}, and investigate a connection between certain SLE$_\kappa$ speeds and CFT minimal models.  In a future article \cite{fkz}, we derive a general crossing-probability formula for critical statistical mechanical systems in polygons. 

During the writing of this article, we learned that K.\ Kyt\"ol\"a and E.\ Peltola recently obtained results very similar to ours by using a completely different approach based on quantum group methods \cite{kype,kype2}.

\section{Acknowledgements}

We thank J.\ J.\ H.\ Simmons and J.\ Rauch for insightful conversations and C.\ Townley Flores for proofreading the manuscript.

This work was supported by National Science Foundation Grants Nos.\ PHY-0855335 (SMF) and DMR-0536927 (PK and SMF).

\appendix{}

\section{The conformal field theory perspective}\label{preliminaries}

The purpose of this appendix is to connect the system of PDEs (\ref{nullstate}, \ref{wardid}), which we have examined from a purely mathematical point of view in this article, with certain aspects of CFT and its application to critical lattice models that we have mentioned.  Although this system, to our knowledge, first appeared in CFT, that area of study is not necessary to understand our mathematical results about it. However, to understand our results from the CFT point of view is still useful and interesting.   For this reason, we summarize here some basic facts from the application of CFT to critical lattice models in simply-connected domains with a boundary, keeping in mind that intuitive notions will prevail over rigorous analysis.  Ref.\ \cite{c3,bpz,fms,henkel,c1,c2,gruz,rgbw,fqsh,dotsenko} and many references therein say more on this subject.

In this appendix, we work with the continuum limit of the critical $Q$-state Potts model \cite{wu} in the upper half-plane.  To begin, we turn our attention to conformally invariant boundary conditions (BC) for this system, of which there are several \cite{c2,saulbau,bauber2}.  In the \emph{fixed}, or \emph{wired}, BC, we assign all boundary sites along the real axis to the same state, and we denote by ``$a$" the fixed BC with all boundary sites in state $a$.  Also, in the \emph{free} BC, we do not condition the boundary sites, and we denote this BC by ``$f$."  Finally, we may condition the boundary sites to uniformly exhibit any but one of the $Q$ possible states, and we denote by ``$\slashed{a}$" the BC with the state $a$ excluded.  The random cluster representation of the $Q$-state Potts model \cite{wu} has conformally invariant boundary conditions too.  In the \emph{fixed}, or \emph{wired}, BC, we activate all bonds between boundary lattice sites, and we denote by ``$a$" the fixed BC with all bonds in state $a$.  Also, in the \emph{free} BC, we do not condition bonds between boundary sites, and we denote this BC by ``$f$."

We may condition the BC to change at (or near) a point $x\in\mathbb{R}$ by inserting a \emph{boundary-condition-changing (BCC) operator} $\phi(x)$. (In CFT, the term ``operator" is used to refer to random objects that may or may not have operator characteristics in a given physical system.  They are perhaps best envisioned as random fields.)  BCC operators are examples of \emph{boundary operators}, or primary CFT operators with the system boundary as their domains.  Under a conformal bijection $f$ sending the upper half-plane onto another simply connected planar region, a primary operator transforms according to
\be \phi_h(x)\mapsto \phi_h'(f(x))=|\partial f(x)|^{-h}\phi_h(x),\ee
where $h$ is called the \emph{conformal weight} of $\phi$.  Among them, the \emph{Kac operator} $\phi_{r,s}$ is a (boundary) primary operator whose conformal weight equals the \emph{Kac weight} $h_{r,s}$, given by (here, $c(\kappa)$ is the CFT central charge (\ref{central}))
\be\label{hrs}h_{r,s}(\kappa)=\frac{1-c(\kappa)}{96}\Bigg[\Bigg(r+s+(r-s)\sqrt{\frac{25-c(\kappa)}{1-c(\kappa)}}\,\Bigg)^2-4\Bigg]=\frac{1}{16\kappa}\begin{cases}(\kappa r-4s)^2-(\kappa-4)^2,&\kappa>4\\(\kappa s-4r)^2-(\kappa-4)^2,&\kappa\leq4\end{cases},\quad r,s\in\mathbb{Z}^+.\ee
Kac operators are useful because correlation functions with them satisfy ``null-state" PDEs such as (\ref{nullstate}).  Also, many of them (and other primary operators of conformal weight (\ref{hrs}) with $r$ or $s$ equaling zero or a half-integer) induce certain BCCs or represent other physical quantities.  For instance, in the random cluster model, the $(1,2)$ Kac operators
\be\phi_{1,2}^{af}(x_0),\quad \phi_{1,2}^{fa}(x_0),\ee
respectively induce a \emph{fixed-to-free} and \emph{free-to-fixed} BCC at $x_0$.  Indeed, the former (resp.\ latter) changes the $a$ (resp.\ $f$) BC for $x<x_0$ to the $f$ (resp.\ $a$) BC for $x>x_0$ \cite{c2}.  And in the $Q$-state Potts model, the $(2,1)$ Kac operators, 
\be\phi_{2,1}^{a \slashed{a}}(x_0),\quad \phi_{2,1}^{\slashed{a}a}(x_0),\ee
induce similar ``fixed-to-free" and ``free-to-fixed" BCCs at $x_0$ (although these descriptions are not entirely accurate).  Indeed, the former (resp.\ latter) changes the $a$ (resp.\ $\slashed{a}$) BC for $x<x_0$ to the $\slashed{a}$ (resp.\ $a$) BC for $x>x_0$ \cite{c2}.

\begin{figure}[t!]
\centering
\includegraphics[scale=0.4]{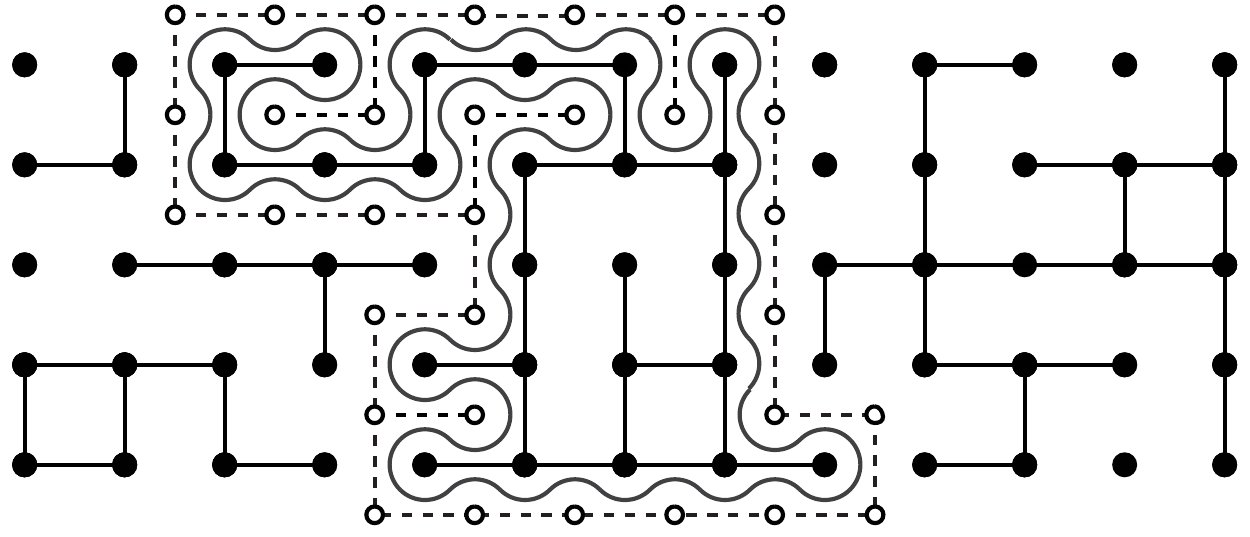}\\
\caption{The winding boundary loop surrounds an FK cluster.  If the states of the sites outside this loop are different from that inside, then the dashed boundary loop surrounds a spin cluster.  The part of either boundary loop that lies above the bottom row of lattice sites is the corresponding boundary arc.}
\label{FKboundaryloop}
\end{figure}

The fixed-to-free or free-to-fixed BCC operator $\phi_{1,2}(x_0)$ (resp.\ $\phi_{2,1}(x_0)$) bears a second interpretation.  As this BCC operator conditions an FK-cluster (resp.\ a spin-cluster) to anchor to the fixed side of $x_0$, it also conditions the interface, or \emph{boundary arc} (figure \ref{FKboundaryloop}), between that cluster and the rest of the system to anchor to $x_0$ \cite{gruz, rgbw}.  For this reason, we also call this BCC operator a \emph{one-leg boundary operator.}  The boundary arc fluctuates in the system domain with the law of multiple SLE$_\kappa$.  In particular, the speed $\kappa$ that models FK-cluster interfaces in the random cluster model and spin-cluster interfaces in the Potts model, both with $Q\in\{2,3,4\}$ states, is related to $Q$ through \cite{smir} 
\be\label{kappaQ}Q=\begin{cases}4\cos^2(4\pi/\kappa), & \text{$4\leq\kappa<8$ (random cluster)} \\ 4\cos^2(\pi\kappa/4), & \text{$0<\kappa\leq4$ (Potts model)}\end{cases}.\ee
With speeds $\kappa$ generating FK-cluster interfaces in the dense phase $(4,\infty)$ of SLE$_\kappa$ and speeds $\hat{\kappa}$ that generate spin-cluster interfaces in the dilute phase $(0,4]$ \cite{rohshr}, we collect $\phi_{1,2}(x_0)$ and $\phi_{2,1}(x_0)$ together with the notation
\be\label{oneleg}\psi_1(x_0):=\begin{cases}\phi_{1,2}(x_0),&\kappa>4 \\ \phi_{2,1}(x_0),&\kappa\leq 4\end{cases}.\ee
After choosing which side of $x_0$ is free, we adopt the convention that $\psi_1(x_0)$ changes the BC from free on that side to any of the fixed states on the other side of $x_0$.  Equation (\ref{hrs}) gives the conformal weight (\ref{theta1}) of $\psi_1$,
\be\theta_1:=\left\{\begin{array}{ll}h_{1,2}, &\kappa>4 \\ h_{2,1}, & \kappa\leq4\end{array}\right\}=\frac{6-\kappa}{2\kappa},\ee
which, incidentally, is a continuous function at the transition $\kappa=4$ from the dilute phase to the dense phase.  We call this conformal weight the \emph{one-leg boundary weight}.

We now explain how the one-leg boundary operator $\psi_1$ induces a BCC.  With a Potts or random cluster model on a lattice in the upper half-plane and $2N$ marked points $x_1<x_2<\ldots<x_{2N}$ in the system boundary, a free-fixed side alternating BC (FFBC) is the event that the intervals $(x_i,x_{i+1})$ with $i$ odd (or even, but not both) exhibit any of the available fixed BCs while the other intervals exhibit the free BC (or vice versa).  By conditioning various fixed intervals to exhibit the same state, we generate different FFBCs.  Now, because the BCCs are fixed to precise locations, the probability of the $\vartheta$th FFBC event is zero.  For this reason, we consider the collection of $\vartheta$th FFBC events with the $i$th BCC occurring within a very small distance $\epsilon_i$ from $x_i$.  If $Z_\vartheta(\epsilon_1,\epsilon_2,\ldots,\epsilon_{2N})$ is the partition function summing over this collection, and $Z_f$ is the free partition function summing over the entire sample space of the system, then the continuum limit of the ratio $Z_\vartheta/Z_f$ should behave as
\be\label{partratio}Z_\vartheta(\epsilon_1,\epsilon_2,\ldots,\epsilon_{2N})/Z_f\underset{\epsilon_i\rightarrow0}{\sim}c_1^{2N}\epsilon_1^{\theta_1}\epsilon_2^{\theta_1}\ldots\epsilon_{2N}^{\theta_1}\langle\psi_1(x_1)\psi_1(x_2)\ldots\psi_1(x_{2N})\rangle_{\vartheta},\ee
where $c_1$ is a non-universal constant associated with $\psi_1$.  We interpret (\ref{partratio}) to say that each one-leg boundary operator actually induces the BCC only somewhere near its location \cite{c3,c2}.  According to the CFT null-state condition \cite{bpz,fms,henkel}, the correlation function in (\ref{partratio}) satisfies the system of PDEs (\ref{nullstate}, \ref{wardid}) central to this article.

In CFT, primary boundary operators exhibit the following algebraic property called \emph{fusion}.  Within a correlation function, we may replace two primary operators $\phi(x_1)$ and $\phi(x_2)$ whose respective points $x_1$ and $x_2$ are close by a sum over other primary operators (and their ``descendants").  We write this replacement as
\be\label{opeform}\phi_1(x_1)\phi_2(x_2)\underset{x_2\rightarrow x_1}{\sim}\sum_{h}C_{12}^h|x_2-x_1|^{-h_1-h_2+h}\phi_h(x_1)+\dotsm,\ee
with $h_1,$ $h_2,$ and $h$ the conformal weights of the primary operators $\phi_1,\phi_2$, and $\phi_h$.
The sum on the right side of (\ref{opeform}), called an \emph{operator product expansion} (OPE), is the product of the fusion of $\phi_1$ with $\phi_2$ appearing on the left side of (\ref{opeform}), and we call each conformal weight $h$ in the sum a \emph{fusion channel}.  If the primary operators appearing on this left side are one-leg boundary operators, then replacing the left side of (\ref{opeform}) by its right side in the correlation function of (\ref{partratio}) implies that the correlation function equals a sum of Frobenius series in powers of $x_2-x_1$.  Because this correlation function satisfies the system of PDEs (\ref{nullstate}, \ref{wardid}), we are naturally led to ask which, if not all, solutions to this system have such expansions.  We use this question to guide the analysis that preceded lemma \ref{boundedlem} in section \ref{boundary behavior}.

\begin{figure}[b!]
\centering
\includegraphics[scale=0.25]{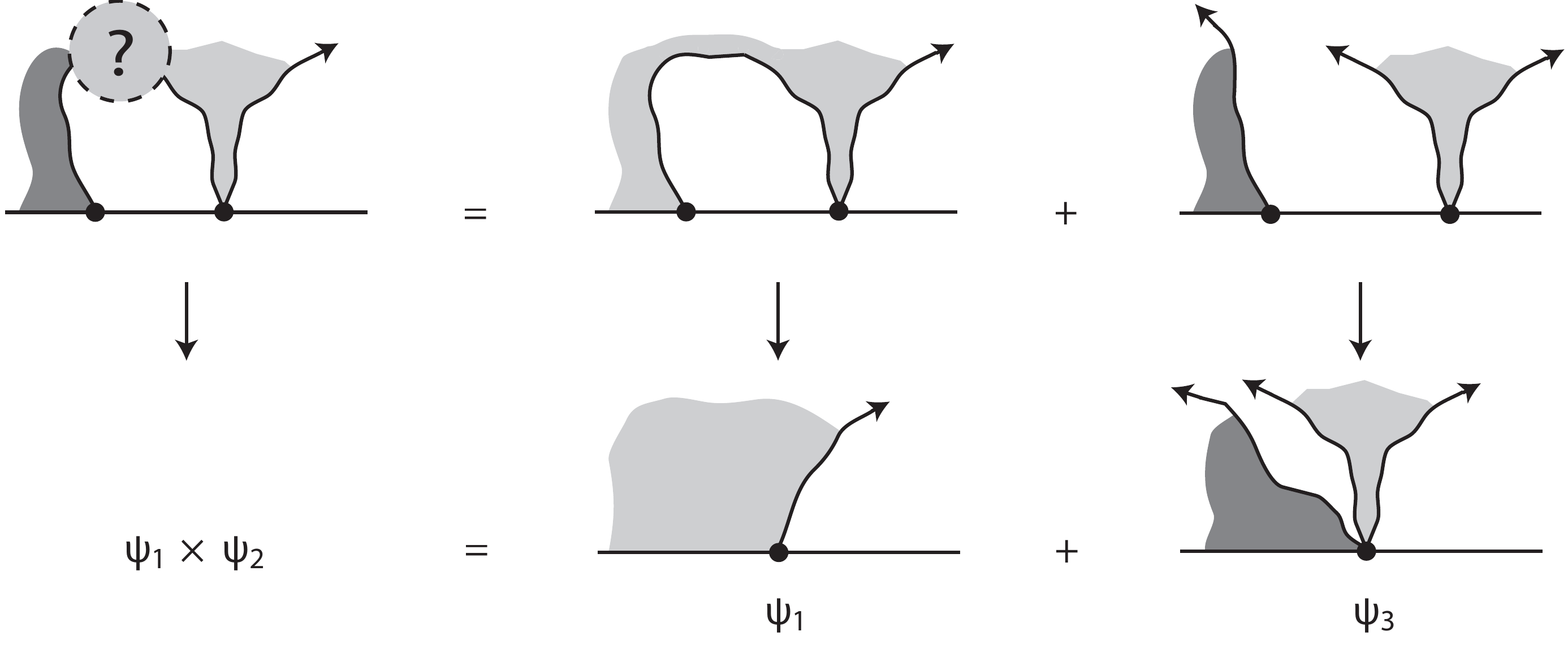}
\caption{The fusion rule $\psi_1\times\psi_2\sim\psi_1+\psi_3$.  Each $\psi_s$ sums over all possible boundary cluster states.  The $\psi_1$ channel appears only if two adjacent boundary clusters exhibit the same state.}
\label{ThreeLegFuse}
\end{figure}

\begin{figure}[t!]
\centering
\includegraphics[scale=0.25]{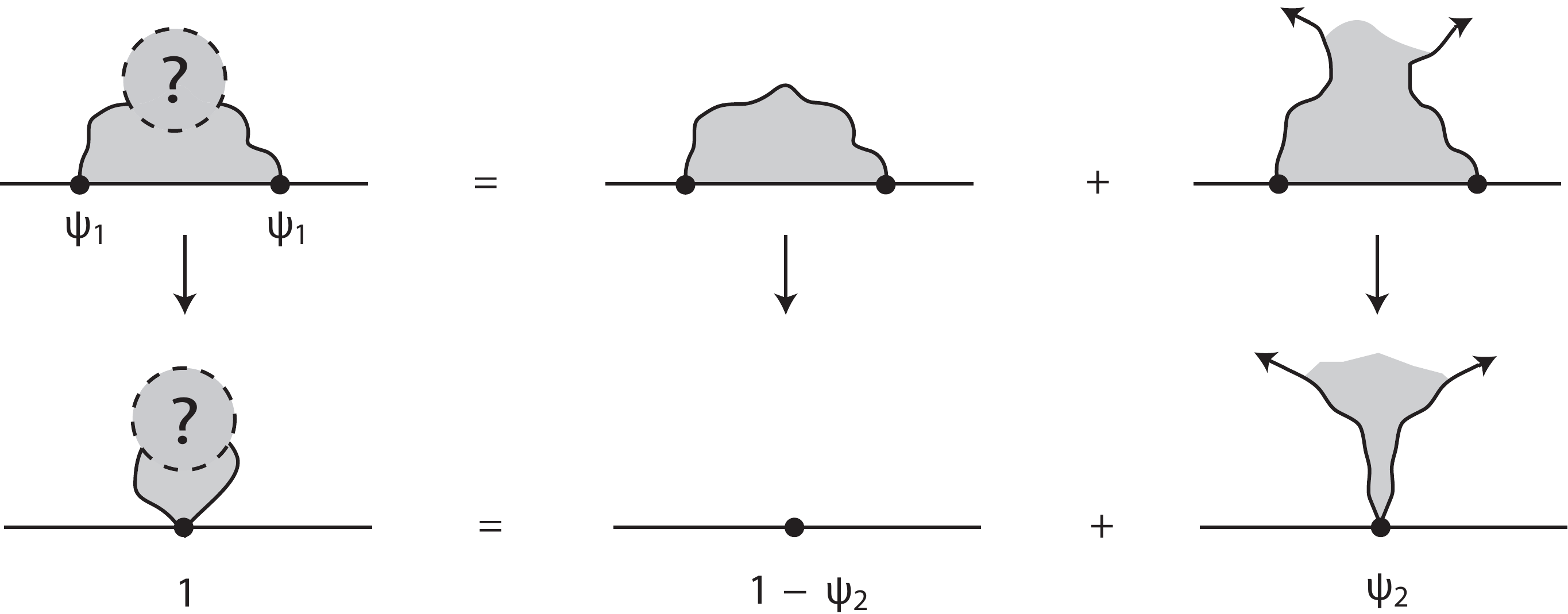}
\caption{The fusion rule $\psi_1\times\psi_1\sim\mathbf{1}+\psi_2.$  The identity (resp.\ two-leg) channel does not (resp.\ does) condition two adjacent boundary arcs to propagate (i.e., to not join together).}
\label{TwoLegFuse}
\end{figure}

If the two operators on the left side of (\ref{opeform}) are Kac operators (as is the case with one-leg boundary operators), then the content of their OPE is strongly constrained.  For example, the OPE of two $\phi_{1,2}$ (resp.\ $\phi_{2,1}$) Kac operators may only contain a $\phi_{1,1}$ (i.e., identity $\mathbf{1}$) Kac operator and/or a $\phi_{1,3}$ (resp.\ $\phi_{3,1}$) Kac operator (and their descendants).  We call this restriction a \emph{fusion rule}, and we write it as (after, as is customary in CFT, dropping the constant, the power of the difference $x_2-x_1$, and the descendant terms for notational conciseness)   
\be\label{thefusionrule}\left\{\begin{array}{l}\phi_{1,2}(x_1)\phi_{1,2}(x_2)\underset{x_2\rightarrow x_1}{\sim}\mathbf{1}+\phi_{1,3}(x_1) \\ 
\phi_{2,1}(x_1)\phi_{2,1}(x_2)\underset{x_2\rightarrow x_1}{\sim}\mathbf{1}+\phi_{3,1}(x_1)\end{array}\right..\ee
Ref.\ \cite{bpz,fms,henkel} gives general fusion rules for Kac operators.  Next, we interpret one-leg boundary operator fusion rules in context with their realization as BCC operators.

In the random cluster model, if a BCC from state $a$ to $f$ takes place at $x_1$, a BCC from $f$ to state $b$ takes place at $x_2$, and we send $x_2\rightarrow x_1$, then (\ref{thefusionrule}), expressed in terms of the boundary operators for these BCCs, reads
\be\label{12ope}\phi_{1,2}^{a f}(x_1)\phi_{1,2}^{fb}(x_2)\underset{x_2\rightarrow x_1}{\sim}\delta_{ab}\mathbf{1}^{aa}+\phi_{1,3}^{ab}(x_1),\ee
with this interpretation.   If $a\neq b$, then the two BCCs join into one passing from state $a$ to state $b$, and the $(1,3)$ Kac operator in (\ref{12ope}) implements this new BCC at $x_1$.  On the other hand, if $a=b$, then the BCCs annihilate each other as they meet, and the identity operator in (\ref{12ope}) captures this disappearance.  (An identity operator in a correlation function has no effect in the sense that is nonlocal and its removal does not alter the value of the correlation function.)  The second term on the right side of (\ref{12ope}) includes all configurations with an infinitesimal free segment $(x_1,x_2)$ that abuts a fixed segment on either end, so the FK boundary cluster anchored to the left fixed segment is disconnected from that anchored to the right fixed segment.  The identity term in (\ref{12ope}) includes these configurations too.  

If $a=b$ in the previous BCC scenario, then we may switch the order of the BCCs, thus finding a fusion rule similar to that of (\ref{12ope}):
\be\label{backward}\phi_{1,2}^{fa}(x_1)\phi_{1,2}^{a f}(x_2)\underset{x_2\rightarrow x_1}{\sim}\mathbf{1}^{ff}+\phi_{1,3}^{fa f}(x_1).\ee
The second boundary operator on the right side of (\ref{backward}) conditions the existence of an infinitesimal segment at $x_1$ fixed to state $a$, and this forces a state-$a$ FK boundary cluster to anchor to this infinitesimal segment.  

In either case, we may continue this process an arbitrary number of times to create an arbitrary number of BCCs proximal to $x_1$.  This implies the general fusion rule
\be\label{genfuserule}\phi_{1,2}(x_1)\phi_{1,s+1}(x_2)\underset{x_2\rightarrow x_1}{\sim}\phi_{1,s}(x_1)+\phi_{1,s+2}(x_1),\ee
with the superscripts indicating the BCCs suppressed.  In this rule $\phi_{1,s+1}(x_2)$ implements $s$ distinct BCCs clumped very near $x_2$, and fusing $\phi_{1,2}(x_1)$ with $\phi_{1,s+1}(x_2)$ either annihilates one of these BCCs or adds a new one.

In the Potts model, if a BCC from state $a$ to state $\slashed{a}$ takes place at $x_1$, a BCC from state $\slashed{a}$ to state $a$ takes place at $x_2$, and we send $x_2\rightarrow x_1$, then (\ref{thefusionrule}) expressed in terms of the boundary operators for these BCCs reads
\be\label{21ope}\phi_{2,1}^{a \slashed{a}}(x_1)\phi_{2,1}^{\slashed{a}a}(x_2)\underset{x_2\rightarrow x_1}{\sim}\mathbf{1}^{aa}+\phi_{3,1}^{a\slashed{a}a}(x_1).\ee
The interpretation of (\ref{21ope}) is identical to that of (\ref{12ope}), except that the $(3,1)$ Kac operator separates two disjoint boundary spin clusters with an infinitesimal free segment centered on $x_1$.  A similar adaption of (\ref{genfuserule}) gives
\be\label{genfuserule2}\phi_{2,1}(x_1)\phi_{s+1,1}(x_2)\underset{x_2\rightarrow x_1}{\sim}\phi_{s,1}(x_1)+\phi_{s+2,1}(x_1).\ee

Because these models have multiple-SLE$_\kappa$ descriptions too, interpreting these fusion rules in terms of boundary arc connectivities is natural.  Generalizing (\ref{oneleg}), we define the \emph{$s$-leg boundary operator} $\psi_s$ with \emph{$s$-leg boundary weight} $\theta_s$,
\be\label{bdysleg}\psi_s(x_0)=\begin{cases}\phi_{1,s+1}(x_0),&\kappa>4 \\ \phi_{s+1,1}(x_0),&\kappa\leq 4\end{cases},\quad \theta_s:=\left\{\begin{array}{ll}h_{1,s+1}, &\kappa>4 \\ h_{s+1,1}, & \kappa\leq4 \end{array}\right\}=\frac{s(2s+4-\kappa)}{2\kappa},\ee
so called because $\psi_s$ conditions $s$ distinct boundary arcs to emanate from the $s$ BCCs tightly accumulated around $x_0$.  These $s$ BCCs alternate from fixed-to-free to free-to-fixed or vice versa.  Fusion rules (\ref{genfuserule}, \ref{genfuserule2}) combine into
\be\label{legfuse}\psi_1(x_1)\psi_s(x_2)\underset{x_2\rightarrow x_1}{\sim}\psi_{s-1}(x_1)+\psi_{s+1}(x_1),\quad s\in\mathbb{Z}^+,\ee
which we interpret as follows.  If the boundary arc anchored to $x_1$ connects with the leftmost of the boundary arcs anchored to $x_2>x_1$, then this one boundary arc contracts to a point almost surely as $x_2\rightarrow x_1$, leaving $s-1$ boundary arcs emanating from $x_1$.  This explains the first term on the right side of (\ref{legfuse}).  On the other hand, if these two adjacent boundary arcs are disjoint, then we find $s+1$ distinct boundary arcs emanating from $x_1$ after sending $x_2\rightarrow x_1$.  This explains the second term on the right side of (\ref{legfuse}).  (See figure \ref{ThreeLegFuse} for the case $s=2$.)  We encounter this latter case again in \cite{florkleb2}.

The case $s=1$ is somewhat different.  On the right side of (\ref{legfuse}), we find a ``zero-leg" boundary operator (i.e., the identity) and a two-leg boundary operator.  Loosely speaking, this OPE may contain only one of these operators, or it may contain both.  If the OPE contains only the identity (resp.\ two-leg) operator, then $(x_1,x_2)$ is an identity (resp.\ two-leg) interval.  We examine these two cases more closely.  In particular, an interval is a two-leg interval of the CFT correlation function in (\ref{partratio}) if and only if the associated boundary arcs are almost surely \emph{propagating} \cite{florkleb4}.  That is, the boundary arcs do not join to form one curve (figure \ref{TwoLegFuse}).  The correspondence between a unique primary operator in the OPE and definite information on the connectivity of the associated boundary arcs is similar to the multi-leg fusion cases described above.  On the other hand, if an interval is an identity interval, then the associated boundary arcs are \emph{not} conditioned to be propagating almost surely or contractible almost surely (where ``contractible" means that the boundary arcs do join to form one curve (figure \ref{TwoLegFuse})) \cite{js}.  This corresponds to the notion in CFT that an identity operator does not condition boundary arc connectivities, a situation that is different from the multi-leg fusion cases described above.  We explore these matters more concretely in \cite{florkleb4}.

\end{document}